%% file: main.tex
\documentclass[11pt,letterpaper]{article}
\usepackage[utf8]{inputenc}			% Allows the use of UTF-8 characters
\usepackage{makeidx}
\usepackage[pdftex]{graphicx}
\usepackage{mathtools,amsfonts}
\usepackage{array}
\usepackage{multirow, hhline}

\usepackage{nicefrac}
\usepackage{tikz}
\usetikzlibrary{arrows}
\usetikzlibrary{calc}
\usepackage{paralist}

\usepackage{wrapfig}
\usepackage[margin=1in]{geometry}	% Defines margins to 1 inch
\usepackage{lipsum}			% provides filler text \lipsum[n]
\usepackage{comment}
\usepackage{chngpage}			% Allows width to be changed dynamically
\usepackage{amsfonts,amssymb,amsmath, amsthm}
\usepackage{amsopn}
\usepackage{color}
\usepackage{pgfplots}
\usepackage{pgfgantt}
\usepackage[noend]{algpseudocode}
\usepackage[Algorithm]{algorithm}
\usepackage{subcaption}
\usepackage{tabto}
\usepackage{csquotes} %quotations
\usepackage{multirow}
\usepackage{amsthm}
\usepackage[colorlinks,urlcolor=black,citecolor=black,linkcolor=black,menucolor=black]{hyperref}
\newtheorem{theorem}{Theorem}

\newtheorem{corollary}[theorem]{Corollary}
\newtheorem{observation}[theorem]{Observation}
\newtheorem{definition}[theorem]{Definition}

\newtheorem{lemma}[theorem]{Lemma}
\newtheorem{problem}[theorem]{Problem}
\newtheorem*{problem*}{Problem} 

\newtheorem*{fact*}{Fact}

\newtheorem{claim}[theorem]{Claim}

\newcounter{example}[section]

\title{Competitive Equilibrium with Chores:\\ Combinatorial Algorithm and Hardness}
	\author{Bhaskar Ray Chaudhury\thanks{University of Illinois at Urbana-Champaign}\\ \texttt{\small braycha@illinois.edu} \and Jugal Garg\thanks{University of Illinois at Urbana-Champaign. Supported by NSF Grant CCF-1942321}\\ \texttt{\small jugal@illinois.edu} \and Peter  McGlaughlin \thanks{University of Illinois at Urbana-Champaign}\\ \texttt{\small mcglghl2@illinois.edu} \and Ruta Mehta\thanks{University of Illinois at Urbana-Champaign. Supported by NSF Grant CCF-1750436}\\ \texttt{\small rutameht@illinois.edu}}

\usepackage{booktabs} % For formal tables
\newcommand{\emax}{e_{\mathit{max}}}
\newcommand{\emin}{e_{\mathit{min}}}
\newcommand{\temax}{\tilde{e}_{\mathit{max}}}
\newcommand{\temin}{\tilde{e}_{\mathit{min}}}
\newcommand{\iter}{\mathit{iter}}

\newcommand\addtag{\refstepcounter{equation}\tag{\theequation}}

\newcommand{\MPB}{\mathit{MPB}}
\newcommand{\MBB}{\mathit{MBB}}

\newcommand{\M}{\mathbf{M}}
\newcommand{\A}{\overline{a}}

\newcommand{\NSW}{{\mathit{NSW}}}

\newcommand{\of}{\mathit{outflow}} 
\newcommand{\eA}{{$(1-\epsilon)$-approximate }}
\newcommand{\DG}{{{\mathcal D}}}
\begin{document}
	
\maketitle 
 \thispagestyle{empty}
% Abstract. Note that this must come before \maketitle.
\begin{abstract}
We study the computational complexity of finding a competitive equilibrium (CE) with chores when agents have linear preferences. CE is one of the most preferred mechanisms for allocating a set of items among agents. 
%An allocation through CE is a canonical way of allocating a set of items among agents. 
CE with equal incomes (CEEI), Fisher, and Arrow-Debreu (exchange) are the fundamental economic models to study allocation problems, where CEEI is a special case of Fisher and Fisher is a special case of exchange. When the items are goods (giving utility), the CE set is convex even in the exchange model, facilitating several combinatorial polynomial-time algorithms {(starting with the seminal work of Devanur, Papadimitriou, Saberi and Vazirani~\cite{DevanurPSV08})} for all of these models. In sharp contrast, when the items are chores (giving disutility), the CE set is known to be non-convex and disconnected even in the CEEI model. Further, no combinatorial algorithms or hardness results are known for these models.

In this paper, we give two main results for CE with chores: 
\begin{itemize}
\item A combinatorial algorithm to compute a $(1-\varepsilon)$-approximate CEEI in time $\tilde{\mathcal{O}}(n^4m^2 / \varepsilon^2)$, where $n$ is the number of agents and $m$ is the number of chores. %This implies that the problem of finding an exact CE in the Fisher model lies in CLS (PPAD $\cap$ PLS). 
\item PPAD-hardness of finding a $(1-1/\textup{poly}(n))$-approximate CE in the exchange model under a sufficient condition. 
\end{itemize}
%These results show a separation between CEEI and exchange models with chores assuming PPAD $\neq $ P. 
To the best of our knowledge, these results show the first separation between the CEEI and exchange models when agents have linear preferences, assuming PPAD $\neq $ P. Furthermore, this is also the first separation between the two economic models when the CE set is non-convex in both cases.

{Finally, we show that our new insight implies a straightforward proof of the existence of an allocation that is both envy-free up to one chore (EF1) and Pareto optimal (PO) in the discrete setting when agents have factored bivalued preferences. This result is recently obtained in~\cite{GargMQ22,EbadianPS22} using an involved analysis.}
\end{abstract}

% Paper body
\section{Introduction}
\emph{Competitive equilibrium (CE)} theory has been one of the most fundamental concepts in mathematical economics for more than a century. Problems in this domain, study the pricing and allocation of resources to agents based on the interaction of demand and supply. %This theory finds significant applications in  market equilibrium and fair division.  

The existence and computation of CE has been  extensively studied in several economic models. The two most fundamental economic models are the \emph{Arrow-Debreu (exchange)} model and the \emph{Fisher model}. The exchange model is like a barter system, where each agent comes with some initial endowment of items and exchanges them with other agents to maximize their utility. The goal is to determine the prices for the items such that (i) each agent gets her most preferred  affordable\footnote{Affordable in exchange of her initial endowment of items.} bundle of items and (ii) all items are completely allocated. The Fisher model is a special case of the exchange model where every agent owns a fixed amount of each good. This has been shown to be equivalent to the setting where each agent is endowed with some fixed amount of money (instead of items). A prominent special case of the Fisher model, which is also of high interest to the fair division community, is \emph{Competitive Equilibrium with Equal Income (CEEI)}, where each agent is endowed with equal amount of money.  %It is a folklore result that the Fisher model is a special case of the exchange model. %Both the models are formally defined in upcoming Subsection~\ref{model}. 

The items to be divided can be either goods (giving utility) or chores (giving disutility). We assume that agents have \emph{linear} utility functions, which are also very commonly used in applications.\footnote{Utility from a bundle $x=\langle x_1, x_2, \dots, x_m \rangle$ of items is defined as: $u(x) = \sum_j u_{j}x_{j}$, where $u_j$ is the utility from one unit of item $j$. Spliddit~(\url{www.spliddit.org}) is a user friendly online platform for computing fair allocation in a variety of problems, which has drawn tens of thousands of visitors in the last few years~\cite{GoldmanP14}. Spliddit uses additive preferences which are the parallel to linear utilities in the context of dividing indivisible goods.} %For simplicity, our algorithmic results will focus on CEEI as they can be extended to the Fisher setting.

\paragraph{CE with goods.} CE with goods has been well-studied in all models since the 1950s: Earlier work settled the existence of equilibrium~\cite{Gale76} and derived various convex programming and linear complementarity problem (LCP) formulations that capture the set of CE~\cite{EisenbergG59,nenakov83,cornet89,Eaves76}. As a result, there have been continuous optimization based algorithms (interior point~\cite{Jain07}, ellipsoid~\cite{Ye08}) for determining CE. Later on, faster algorithms have been obtained using combinatorial methods~\cite{DevanurPSV08, Orlin10, DuanM15, DuanGM16, GargV19}. In contrast to most continuous optimization based methods, the combinatorial algorithms simulate intuitive dynamics: they start with arbitrary prices of the goods and gradually adjust them according to demand and supply. The challenge lies in showing fast (polynomial time) convergence. Moreover, most of the combinatorial algorithms are flow-based algorithms and therefore provide more room for faster algorithms, given the deep understanding of flow based subroutines since the last four decades. Indeed, the fastest weakly polynomial time  and strongly polynomial time algorithms for finding CE in all models are combinatorial algorithms~\cite{Orlin10, DuanGM16, ChaudhuryM18, GargV19}.
%\footnote{It is a well known fact that if there is a $T(n)$-time algorithm to find a CE in a specific setting of the exchange model, where each agents owns only one unit of exactly one good, then one can obtain a $T(nm)$-time algorithm for the general setting of the exchange model where each agent can own arbitrary amounts of all goods. The algorithm in~\cite{} is faster than the combinatorial algorithm~\cite{DuanGM16} in the special setting of the exchange model. However,~\cite{ChaudhuryM18} show how to generalize the ideas of the combinatorial algorithm in~\cite{DuanGM16} to solve the general setting in $\mathcal{\tilde{O}} ((n+m)^7)$ (saving a factor of $(n+m)^7$). The algorithm in~\cite{} crucially relies on the rational convex program for CE in the exchange setting described in~\cite{DevanurGV16}, which only works in the special case of the exchange model.  This is indeed faster than the algorithm obtained by generalizing the ideas in~\cite{}. } 

\paragraph{CE with chores.} CE with chores turns out to be significantly more challenging. In contrast to the case with goods, where the CE set is convex even in the (most general) exchange model, the set of CE  with chores can be non-convex containing many disconnected sets in the CEEI model~\cite{BogomolnaiaMSY17}. This ``disconnectedness" may bring up several computational bottlenecks. In fact, ~\cite{BogomolnaiaMSY17} mention: \emph{``we expect computational difficulties in problems with many agents and/or items.''}
However, quite recently, ~\cite{BoodaghiansCM22} presents a $\tilde{\mathcal{O}}(n^{6}m^{3}/ \varepsilon^2)$ time\footnote{$\tilde{\mathcal{O}}(\cdot )$ hides poly-logarithmic factors. The algorithm in~\cite{BoodaghiansCM22} runs in $\mathcal{O}(n^3/\varepsilon^2)$ iterations (Lemma 13) and each iteration solves a convex quadratic program in $nm$ variables, which takes $\tilde{\mathcal{O}}(n^3m^3)$ time~\cite{GoldfarbL91}.}  exterior point algorithm for determining an $(1-\varepsilon)$-approximate CE with chores in the CEEI model where $n$ and $m$ denote the number of agents and chores, respectively. We remark that there are instances where the set of $(1-\varepsilon)$-approximate CE with chores in the CEEI model are also disconnected. However, the algorithm in~\cite{BoodaghiansCM22} does not provide any insights into the dynamics of pricing and allocation under the influence of demand and supply. Therefore, a natural question that arises is whether there are combinatorial algorithms (which are also hopefully faster) to find an approximate CE with chores. Unfortunately, the combinatorial algorithms for CE with goods and their corresponding convergence analysis do not generalize to the chores setting as all of them exhibit the fundamental bottleneck of \emph{non-monotone surpluses} (elaborated in Section~\ref{techoverview}).

This brings us to the first main contribution of the paper. We overcome the barrier of \emph{non-monotone surpluses} and give the first combinatorial algorithm for determining a $(1-\varepsilon)$-approximate CE with chores in the CEEI model. Our algorithm is significantly faster than the algorithm in~\cite{BoodaghiansCM22} (by a factor of $\tilde{\Omega}(n^{2}m)$). Furthermore, when the disutility values are $\alpha$-\emph{rounded}, %~\cite{BarmanBKS20}, 
i.e., the disutility incurred by any agent from consuming one unit of any chore is a power of $(1+ \alpha)$ for an $\alpha > 0$, we can compute an exact CE in $\tilde{\mathcal{O}}(n^2m^2/ \alpha^2)$ time. We remark that algorithm in~\cite{BoodaghiansCM22} does not obtain an exact CE for rounded disutilities in polynomial time! 
%In fact, our algorithm extends to the setting with a \emph{mixed manna}, where the set of items contains both goods and chores. This provides a solution to~\cite{BogomolnaiaMSY17}'s concern (in abstract): \emph{``the implementation of competitive fairness under linear preferences in interactive platforms like SPLIDDIT will be more difficult when the manna contains bads that overwhelm the goods.''} 

\begin{theorem}
	\label{mainthm1intro}
	There exist combinatorial algorithms that 
	 \begin{itemize}
	 	\item compute  a $(1-\varepsilon)$-approximate CEEI with chores in time $\tilde{\mathcal{O}}(n^4m^2 / \varepsilon^2)$, and 
	 	\item compute an exact CEEI with chores in time $\tilde{\mathcal{O}}(n^2m^2 / \alpha^2)$ when the disutility values are $\alpha$-rounded.
	 \end{itemize} 
\end{theorem}

We mention some additional insights that our algorithm provides. Firstly, the convergence analysis of our combinatorial algorithm also works for the setting with goods, which gives a new form of convergence analysis in that case as well. Secondly, our algorithm also shows that finding an exact CE in the Fisher model is in PLS. Since ~\cite{ChaudhuryGMM21} shows that this problem is in PPAD, which implies that the problem lies at PPAD $\cap$ PLS $= $ CLS~\cite{FearnleyGHS21}. 

We now address the computational complexity of finding a CE in the exchange model. Similar to the case with goods~\cite{Gale76}, a CE may not always exist in the exchange model for arbitrary instances. ~\cite{ChaudhuryGMM22} gives a polynomial-time verifiable sufficient condition\footnote{The sufficient condition is automatically satisfied by the Fisher model.} under which a CE is guaranteed to exist. In our second main result, we show that given an instance that satisfies the sufficient condition, determining an approximate-CE in the exchange model is PPAD-hard.

\begin{theorem}
	\label{mainthm2intro}
	Finding a $(1-1/\textup{poly}(n))$-approximate CE in the exchange model under the sufficient conditions is PPAD-hard.  
\end{theorem}

Theorems~\ref{mainthm1intro} and~\ref{mainthm2intro} show the first separation between the CEEI model and the exchange model when agents have linear preferences, assuming PPAD $\neq $P. To the best of our knowledge, these results also give the first separation between the two economic models where the CE set is non-convex in both cases!~\footnote{Such a separation is also known for constant elasticity of substitution (CES) and Leontief preferences, where equilibrium computation in the CEEI model is in P using convex formulation~\cite{Eisenberg61}, and it is PPAD-hard in the exchange model~\cite{CodenottiSVY06,chen2017complexity,GargMVY17}. However, the set of CE is convex in the CEEI model and non-convex in the exchange model.} 
\medskip

\noindent {{\bf Organization of the rest of the paper.} Section~\ref{sec:model} defines the problem with all studied models. Sections~\ref{sec:alg} and ~\ref{mainres3} present a sketch of main techniques and ideas used in designing combinatorial algorithm and showing PPAD-hardness, respectively.  Section~\ref{sec:bi} presents a sketch of how our new insight implies the existence  of an allocation that is both envy-free up to one chore (EF1) and Pareto optimal (PO) in the discrete setting under bivalued preferences. Section~\ref{sec:rw} presents applications and further related work. Appendices~\ref{algorithm} and ~\ref{ppadhardness} contain the full details of our two main results.} 

\section{Model and Technical Overview}\label{sec:model}
%\noindent{\bf Notations.} For an integer $k>0$, we use $[k]$ to denote set $\{1,\dots,k\}$. For two quantities $x$ and $y$, by $x=(1\pm \epsilon) y$ we mean $(1-\epsilon) y \le x \le (1+\epsilon) y$.
%\subsection{Model}
%\label{model}
A chore division problem consists of a set $[m]$ of divisible chores (bads), and a set $[n]$ of agents, where by $[a]$ we mean set $\{1,\dots,a\}$. We assume that each agent $i$ has a disutility of $d_{ij} \in (0,\infty]$ for one unit of chore $j$.\footnote{If $d_{ij}=0$, we can simply assign chore $j$ entirely to agent $i$ and remove it from the instance.} $d_{ij} = \infty$ implies that chore $j$ cannot be allocated to agent $i$ as she may not have the required skills. If agent $i$ is assigned a bundle $x_i = \langle x_{i1}, x_{i2} , \dots , x_{im} \rangle$, then her total disutility $D_i(x_i) = \sum_{j \in [m]} d_{ij} x_{ij}$. We first define CE in the exchange model as it is the more general of the two models under consideration in this paper.
\medskip

\noindent{\bf Exchange model.} In the exchange model, each agent $i$ brings $w_{ij}$ units of chore $j$ to be done (by herself or other agents). Given prices $p = \langle p_1, p_2, \dots , p_m \rangle \in \mathbb{R}^m_{\geq 0}$ for chores, where $p_j$ denotes the payment for doing unit amount of chore $j$, agent $i$ needs to earn $\sum_{j \in [m]} w_{ij} p_j$ in order to pay to get her own chores done. %In this light, we define the feasible set of bundles $F_i(p)$ as those bundles with which an agent can earn her required money, i.e., $F_i(p) = \left\{ x_i \in \mathbb{R}^m_{\geq 0} \mid \sum_{j \in [m]} x_{ij} p_j \ge \sum_{j \in [m]} w_{ij}p_j  \right\} $. 
Clearly, $i$ would like to choose a {\em feasible bundle} that minimizes her disutility -- this defines her {\em optimal bundle} (or optimal chore set). 

\begin{equation}\label{eq:ob}
\mathit{OB}_i(p)=\mathit{argmin}_{x_i \in \mathbb{R}^m_{\geq 0}: \langle x_i,p\rangle \ge \langle w_i,p\rangle}{~~D_i(x_i)}.
\end{equation}

Due to linearity of the disutility functions, in an optimal bundle,  each agent $i$ is assigned only those chores that minimize her disutility per dollar earned and agent $i$ earns money exactly equal to the total price of her endowments. Let $\MPB_i$ denote the minimum pain per buck of agent $i$ at prices $p$, i.e., $\MPB_i = \mathit{min}_{j \in [m]} d_{ij}/p_j$. Formally, if $x_i \in \mathit{OB}_i(p)$, then, 
\[
\forall j\in [m],\  \ x_{ij}>0 \ \  \Rightarrow \  \ \frac{d_{ij}}{p_j} =\MPB_i,
\ \ 
\mbox{ and }\ \ 
\sum_{j \in [m]} x_{ij} \cdot p_j = \sum_{j \in [m]} w_{ij} \cdot p_j.
\]
In the above ratios, to deal with zero prices and infinite disutilities we assume that $\infty/a > b/0$ for any $a,b\ge 0$. Clearly, an optimal bundle of an agent contains only those chores for which she has finite disutility.  The price vector $p$ is said to be at a CE if all chores are completely assigned when every agent gets one of her optimal bundles, {\em i.e.,} $x_i \in \mathit{OB}_i(p)$ ~$\forall i \in [n]$, and $\sum_{i\in[n]} x_{ij} =\sum_{i\in[n]} w_{ij},\ \forall j\in [m]$. It is without loss of generality to assume that each chore is available in one unit total, i.e. $ \sum_{i \in [n]} w_{ij}=1$ (through appropriate scaling of the disutility values). 

At approximate competitive equilibrium, we require that agents approximately earn enough to pay for their chores. 
Next we define $(1-\epsilon)$-approximate CE (as in \cite{BoodaghiansCM22}), where setting $\epsilon=0$ gives the exact CE.  Now on, by $x=(1\pm \epsilon) y$ we mean $(1-\epsilon) y \le x \le (1+\epsilon) y$. %defined next. %We now formally describe the competitive equilibrium.

%\begin{problem}[Chore Division in the Exchange Model]
%	\label{exchange}
%	Given a set of agents $[n]$, divisble chores $[m]$, disutilities $\cup_{i\in [n], j \in [m]}d_{ij}$ and endowments $\cup_{i \in [n], j \in [m]}w_{ij}$, our goal is to find
%	a price vector $p = \langle p_1, p_2, \dots , p_m \rangle \in \mathbb{R}^m_{\geq 0}$ and allocation $x = \langle x_1, x_2, \dots, x_n \rangle$, such that 
%	\begin{itemize}
%		\item every agent gets their optimal bundle, i.e.,  $x_i \in \mathit{OB}_i(p)$, and 
%		\item all chores are completely allocated, i.e., $\sum_{i \in [n]} x_{ij} = \sum_{i \in [n]} w_{ij} = 1$, for all $j \in [m]$.
%	\end{itemize}
%\end{problem}
\begin{definition}\label{def:ce}\cite{BoodaghiansCM22}
We say that price vector $p = \langle p_1, p_2, \dots , p_m \rangle \in \mathbb{R}^m_{\geq 0}$ and allocation $x = \langle x_1, x_2, \dots, x_n \rangle$ are at $(1-\epsilon)$-approximate competitive equilibrium (CE), if
\begin{enumerate}
\setlength\itemsep{0em}
\item {\bf Complete allocation.} For each item/chore $j\in [m]$,  $\sum_{i \in [n]} x_{ij} = 1 (=\sum_{i\in[n]} w_{ij})$.
\item {\bf Consume best chores.} For each agent $i\in [n]$, $x_{ij}> 0 \Rightarrow \frac{d_{ij}}{p_j} =\MPB_i$. % for all $j\in [m]$. 
\item {\bf Approximate earning.} For each agent $i\in [n]$, $\sum_{j\in [m]} x_{ij} p_j = (1\pm \epsilon) \cdot \sum_{j \in [m]} w_{ij} p_j$. %in the {\em exchange model}.
\end{enumerate}
%	\begin{enumerate}
%		\item {\bf (optimal bundle)} every agent gets their optimal bundle, i.e.,  $x_i \in \mathit{OB}_i(p),\ \forall i\in[n]$, and 
%		\item {\bf (complete allocation)} all items are completely allocated, i.e., $\sum_{i \in [n]} x_{ij} =1,\ \forall j\in[m].$ % = \sum_{i \in [n]} w_{ij} = 1,\ \forall j\in[m]$.
%	\end{enumerate}
\end{definition}

In general, a CE may not always exist in the exchange model (with goods or bads)~\cite{ChaudhuryGMM22}. In the setting with goods, there exists a polynomial time verifiable necessary and sufficient condition, under which an instance will always admit a CE. However, in the setting with chores, determining whether an arbitrary instance admits a CE is NP-complete~\cite{ChaudhuryGMM22}, even in the CEEI model. This rules out the possibility of obtaining a polynomial time verifiable necessary and sufficient condition unless P = NP. Therefore, the best one can hope for is a polynomial time verifiable sufficient condition that captures interesting instances. Chaudhury et al.~\cite{ChaudhuryGMM22} show such a sufficient condition. We now elaborate their sufficient condition. To this end, we first define the \emph{economy graph} and the \emph{disutility graph} of an instance $I$. The {economy graph} $E(I)$ is a directed graph with vertices corresponding to the agents $[n]$ and there is an edge from $i$ to $i'$ in $E(I)$ if there exists a chore $j$ such that $w_{ij} > 0$ and $d_{i'j} \neq \infty$. The {disutility graph} $G(I)$ is a bipartite graph with the set of agents $[n]$ and the set of chores $[m]$ as the two parts and there is an edge from $i \in [n]$ to $j \in [m]$ if $d_{ij} \neq \infty$.~\cite{ChaudhuryGMM22} show the following theorem.

\begin{theorem}
	\label{sufficientcondition}
	An instance $I$ of chore division in the exchange model, admits a CE, if it satisfies the following two conditions:
	\begin{itemize}
	\setlength\itemsep{0em}
		\item[$SC_1.$] $E(I)$ is strongly connected, and 
		\item[$SC_2.$] $G(I)$ is a disjoint union of complete bipartite-graphs (bicliques).
	\end{itemize}
\end{theorem}

\paragraph{CEEI model.} The CEEI model is a special case of the exchange model, where instead of the initial endowment of chores, all the agents have a fixed and the same earning requirement, say of one unit (dollar). In particular, the only change is in the definition of the feasible set of chores that can be allocated to an agent at prices $p$, namely $\{ x_i \in \mathbb{R}^m_{\geq 0} \mid \sum_{j \in [m]} x_{ij}p_j \geq 1 \}$. Under this feasibility set, $\mathit{OB}_i(p)$ is well-defined. Accordingly, in the Definition \ref{def:ce} of $(1-\epsilon)$-approximate competitive equilibrium, the only change is in the last condition, % can be as in Definition \ref{def:ce}.
\[
\text{{\bf Approximate earning.} For each agent }i\in [n], \sum_{j\in [m]} x_{ij} p_j = (1\pm \epsilon)
\]

%\begin{problem}[Chore Division in the Fisher Model]
%	\label{fisher}
%	Given a set of agents $[n]$, divisble chores $[m]$, disutilities $\cup_{i\in [n], j \in [m]}d_{ij}$, % and earning requirements $\cup_{i \in [n]}\eta_i$, 
%	our goal is to find
%	a price vector $p = \langle p_1, p_2, \dots , p_m \rangle \in \mathbb{R}^m_{\geq 0}$ and allocation $x = \langle x_1, x_2, \dots, x_n \rangle$, such that 
%	\begin{itemize}
%		\item every agent gets their optimal bundle, i.e.,  $x_i \in \mathit{OB}_i(p)$, and 
%		\item all chores are completely allocated, i.e., $\sum_{i \in [n]} x_{ij} = 1$, for all $j \in [m]$.
%	\end{itemize}
%\end{problem}

To understand why CEEI is a special case of exchange model, note that we can convert an instance of CEEI to an instance of exchange by setting all $w_{ij} = 1/n$ for each $i \in [n]$ and $j \in [m]$ and keeping the disutility values as is. Since the CE prices are scale invariant in the exchange model, wlog we can assume that they sum to $n$. This translates to earning requirement of $1$ for each agent under this reduction.

Recall that determining whether an arbitrary instance of the CEEI model admits a CE is NP-complete~\cite{ChaudhuryGMM22}. It is straightforward to see that if we model an instance in the CEEI model as an instance in the exchange model, then the economy graph $E(I)$ is a clique and thus strongly connected. Therefore, the only sufficient condition that the instance needs to satisfy to admit a CE is that the disutility graph is a disjoint union of bicliques. 

\begin{corollary}
	\label{fisherexistence}
    An instance $I$ of chore division in the CEEI model, admits a CE, if $G(I)$ is a disjoint union of bicliques.
\end{corollary}

Note that Corollary~\ref{fisherexistence} contains all instances where all disutility values are finite. Furthermore, given any arbitrary instance $I$ that satisfies the condition in Corollary~\ref{fisherexistence}, we can decompose the instances into sub-instances, where each sub-instance contains the agents and chores in a component in the disutility graph and solve each sub-instance separately. Thus, if there exists any polynomial time (approximation) algorithm for instances where all disutility values are finite, we get a polynomial time (approximation) algorithm for all instances that satisfy the sufficient condition in Corollary~\ref{fisherexistence}.	
%
%\paragraph{Approximation.} Next we define approximate competitive equilibrium, where approximation appears only in agent's earning requirements. 
%\begin{definition}
%For an $\epsilon\ge 0$, we say that a price vector $p = \langle p_1, p_2, \dots , p_m \rangle \in \mathbb{R}^m_{\geq 0}$ and allocation $x = \langle x_1, x_2, \dots, x_n \rangle$ are at \eCE if 
%\begin{enumerate}
%\item {\bf Complete allocation.} $\forall j\in [m]$, $\sum_{i \in [n]} x_{ij} = 1$.
%\item {\bf Consume best chores.} $\forall i\in [n]$, $x_{ij}> 0 \Rightarrow \frac{d_{ij}}{p_j} =\MPB_i = \mathit{min}_{j \in [m]} d_{ij}/p_j$ for all $j\in [m]$. 
%\item {\bf Approximate earning.} For each agent $i\in [n]$,
%\begin{enumerate}
%\item[(Exchange model.)] $(1-\epsilon) \cdot \sum_{j \in [m]} w_{ij} p_j \le \sum_{j\in [m]} x_{ij} p_j \le (1+\epsilon) \cdot \sum_{j \in [m]} w_{ij} p_j$. %in the {\em exchange model}.
%\item[(CEEI model.)]\ \ \ \   $(1-\epsilon) \le \sum_{j\in [m]} x_{ij} p_j \le (1+\epsilon)$. %in the {\em CEEI model}.
%\end{enumerate}
%\end{enumerate}
%\end{definition}
%
\label{techoverview}
Next, for each of our two main results, we elaborate the fundamental roadblocks and the main technical ideas to overcome them.

\subsection{Combinatorial Algorithm for the CEEI Model}\label{sec:alg} 
In this section, we highlight the issue of \emph{non-monotone surpluses} that arises when we try to generalize combinatorial algorithms from the setting with goods to the setting with chores. Thereafter, we discuss a possible fix which will give us the crucial ideas needed for our combinatorial algorithm.
%We first explain the primary idea behind almost all known combinatorial algorithms that find a CEEI with goods. Thereafter, we point out the issue of \emph{non-monotone surpluses} that arises when one tries to generalize these algorithms and their convergence analysis to chores. Finally, we show how to fix this issue and this will give us the crucial ideas of our combinatorial algorithm.

\paragraph{Combinatorial algorithms for goods.} For the case of goods, think of $d_{ij}$ as the utility (happiness) agent $i$ gets per unit of {\em good (item)} $j$, and has \$1 of budget {\em to spend}. Clearly, at CE the agents would want to consume only the {\em goods} that give {\em maximum utility per dollar spent}, a.k.a. Max-Bang-Per-Buck (MBB). Based on this, in most of the known combinatorial algorithms, finding a CE is formulated as a flow problem as follows: Given prices of {\em goods} $p = \langle p_1, p_2, \dots , p_m \rangle$ such that $\sum_{j \in [m]} p_j = n=$ total-budget, construct an MBB network $N_p = \{ \{s\} \cup [n] \cup [m] \cup \{t\} , E \}$. It is a flow network comprising of a source vertex $s$, the set of agents $[n]$, the set of goods $[m]$ and a sink $t$. For each agent $i$, let $\MBB_i = \mathit{max}_{j \in [m]} d_{ij}/p_j$. The edge set $E$ of $N_p$ comprises of 
\begin{itemize}
\setlength\itemsep{0em}
	\item edge $(s,i)$ from the source $s$ to agent $i$ of capacity $1$ (budget), for all $i \in [n]$,
	\item edge $(i,j)$ with $\infty$ capacity, $\forall i\in [n], \forall j\in [m],$ where $d_{ij}/p_j = \MBB_i$, and %for all $i \in [n]$ and $j \in [m]$, and 
	\item edge $(j,t)$ from good $j$ to sink $t$ of capacity $p_j$, for all $j \in [m]$. 
\end{itemize} 

%If $N_p$ has a max-flow of value $n=\sum_{j \in [m]} p_j$, then a CE can be constructed as follows: 
Observe that, determining a CE reduces to finding a set of prices $p$ such that the maximum flow in $N_p$ is $\sum_{j \in [m]} p_j=n$: 
Let $f$ be a max-flow. Set the allocation $x$ as $x_{ij} = f_{ij} / p_j$ for all $i \in [n]$ and $j \in [m]$.  Note that $f$ saturates all the edges from $[m]$ to $t$ implying {\em complete allocation}. Each agent $i$ is only allocated her MBB goods at prices $p$, and since $f$ saturates $(s,i)$ for all $i$, each agent $i$ sends all of her budget, in turn implying that she gets her {\em optimal bundle}. These conditions imply a CE (see Definition \ref{def:ce} with $\epsilon=0$). %edges from Note that each agent is allocated goods only along MBB edges as all the edges in $N_p$ between agents and goods are  MBB edges. Note that $f$ saturates all the edges from $s$ to $[n]$ and the edges from $[m]$ to $t$. Therefore, each agent spends her budget of $1$ completely on goods that maximize her utility and all the goods are completely sold, implying that we have a CE.
%So now the goal is to discover a price vector $p$ such that graph $N_p$ has max-flow of value $\sum_j p_j = n$.  

Given any flow $f$ in $N_p$, let \emph{surplus} $r_f(i)$ of an agent $i$ be the residual capacity remaining on the edge $(s,i)$, i.e.,  $r_f(i) = 1 - f_{si} = 1 -\sum_{j \in [m]} f_{ij}$ . If the prices are at a CE, the surplus of all agents will be zero. At a high-level, all combinatorial algorithms start with arbitrary prices and allocation and then adjust the prices and allocation such that the $\ell_2$-norm of the surpluses of the agents decreases. Faster algorithms are obtained by achieving faster convergence rates with more sophisticated adjustment rules. We briefly sketch the intuitive price and flow adjustment:
\medskip

	\noindent \emph{Flow adjustment:} During this phase, we make no changes to the prices. At a given price vector $p$, the flow is set to the one that minimizes the $\ell_2$-norm of the surpluses of the agents.
\medskip
	
	\noindent \emph{Price adjustment:} During this phase, we make no changes in the allocation. The price adjustment is done in a very intuitive way: \emph{increase the prices of the goods {\em in-demand} and decrease the prices of the rest}. That is, let $S$ be the set of agents with highest surplus, then the goods adjacent to them, say $\Gamma(S)$ in the MBB graph, are completely sold, as otherwise one can reduce surpluses of some of the agents in $S$ and thereby reduce the $\ell_2$-norm. %  of some of the agents in $S$ by increasing their consumption of the underallocated good in $\Gamma(S)$ and thereby reducing the $\ell_2$-norm. 
	Thus, the goods in $\Gamma(S)$ are {\em in-demand} as the demand exceeds supply for these goods. 
	
	We increase the prices of the goods in $\Gamma(S)$ and decrease the prices of the goods outside $\Gamma(S)$ continuously such that the sum of prices of the goods equals the sum of budgets of the agents, i.e., we scale the prices of goods in $\Gamma(S)$ by $x>1$ and  the prices of the goods in $[m] \setminus \Gamma(S)$ by $y <1 $ such that $x \cdot \sum_{j \in \Gamma(S)}p_j + y \cdot \sum_{j \in [m] \setminus \Gamma(S)} p_j = n$. See Figure~\ref{price-update} for an illustration. %\sum_{j \in [m]}p_j = \sum_{i \in [n]} \eta_i$. 
    Note that, during this price change a feasible flow can be easily maintained by scaling the edge-flows appropriately. Note that the allocation still remains the same as $f_{ij}/p_j$ remains the same.
%The price and flow update can be sumarized as follows,
	
	\[
	 p'_{j} =  \begin{cases} 
					xp_{j} & j \in \Gamma(S) \\
					yp_{j} & j \notin \Gamma(S) 
				\end{cases}   \addtag \label{priceeq}
	\]

	\[
	\begin{array}{ccc}
	f'_{ij} =  \begin{cases} 
					xf_{ij}, & \ \mbox{if }j \in \Gamma(S) \\
					yf_{ij}, & \ \mbox{if } j \notin \Gamma(S) 
				 \end{cases} %\addtag \label{floweq1}
	\quad & \quad
	
	f'_{jt} =  \begin{cases} 
						xf_{jt}, & \ \mbox{if } j \in \Gamma(S) \\
						yf_{jt}, & \ \mbox{if } j \notin \Gamma(S) 
					\end{cases} %\addtag \label{floweq2}
	
	\quad &\quad
	f'_{si} =  \begin{cases} 
					xf_{si}, & \ \mbox{if } i \in S \\
					yf_{si}, & \ \mbox{if } i \notin S 
				\end{cases} \addtag  \label{floweq3}
\end{array}
\]

Note that, in this process the surplus of agents in $S$ is decreasing. %while of those outside $S$ is increasing. Thus $\ell_1$-norm of the surpluses is decreasing.
%Since we want to maintain a feasible flow during the price change, we update the flow $f$ to $f'$ by scaling every incoming flow to $\Gamma(S)$ by $x$ and scaling every incoming flow to $[m] \setminus \Gamma(S)$ by $y$. We appropriately update the flow from $s$ to $[n]$ and from $[m]$ to $t$. 
We keep updating the prices continuously until either a new MBB edge appears in the MBB graph (from some agent in $S$ to a good outside $\Gamma(S)$) or the surplus of one of the agents in $S$ becomes equal to the surplus of an agent outside $S$.  
\medskip

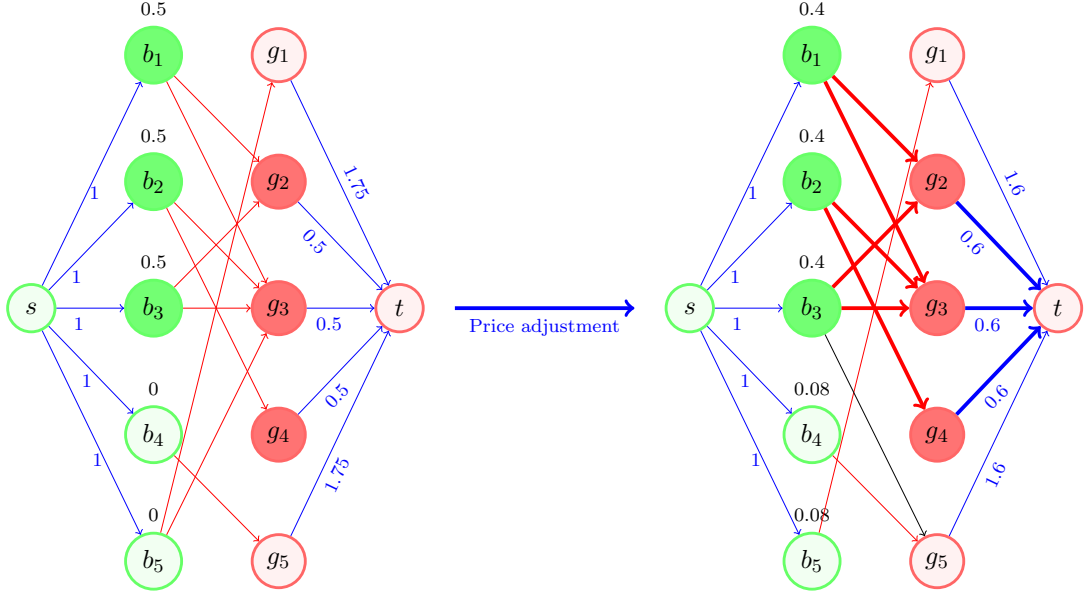
\begin{figure}[htb]
	\centering 
	\label{price-update}
	\scalebox{0.90}{\input{price-update.tex}}
	\caption{\small{Illustration of the price adjustment phase. The MBB flow network $N_p$ before the price adjustment phase is shown on the left. The numbers below the edges from $s$ to $[n]$ and $[m]$ to $t$ are the capacities of the edges (earning limit and prices respectively). The numbers on top of the agents show their respective surpluses w.r.t. a flow in $N_p$ that minimizes the $\ell_2$-norm of the surpluses. The set $S= \{b_1, b_2, b_3 \}$ and $\Gamma(S) = \{g_2, g_3, g_4 \}$. The price adjustment phase increases the prices of the goods in $\Gamma(S)$ and the inflow to $\Gamma(S)$ until either a new MBB edge appears in $N_p$ or one of the surpluses of an agent in $S$ equals that of some agent outside $S$. The figure on the right shows the MBB network after the price adjustment phase -- a new MBB edge from $b_3$ to $g_5$ appears and the $\ell_1$-norm of the surpluses decreases.}}
\end{figure}

%See Figure~\ref{price-update} for an illustration. 
\noindent{\em Correctness.}
We need to first argue that $f'$ is a {\em valid-flow} in $N_{p'}$. Note that if $f'$ augments flow on valid edges in $N_{p'}$, it is easy to see that it satisfies flow conservation and capacity constraints. So we focus only on showing that $f'_{ij} > 0$ if and only if $(i,j) \in N_{p'}$. Note that, since flow on edges are scaled multiplicatively, $f'_{ij}>0$ if and only if $f_{ij}>0$. Therefore, it suffices to prove that no edge with positive flow under $f$ can disappear during the price update. Since the prices of goods in $\Gamma(S)$ are increased by factor $x$ and those in $[m]\setminus \Gamma(S)$ are decreased by factor $y$, only the edges to $\Gamma(S)$ can disappear. However, all agents in $S$ have MBB edges only to goods in $\Gamma(S)$ and since the price of all these goods increases by the same factor, none of the edges from $S$ to $\Gamma(S)$ disappear. The only edges that may disappear are the ones from agents outside $S$ to goods in $\Gamma(S)$. We next argue that these edges must have zero flow. %And those must be from agents outside of $S$. We next argue that these edges must have zero flow. %Furthermore, such an edge has to be from an agent outside of $S$, since agents of $S$ only consume goods from $\Gamma(S)$ and therefor

%It is easy to see that if $f$ was a max-flow wrt prices $p$ before the price update then $f'$ is a max-flow wrt $p'$. First, we argue that $f'$ is indeed a valid flow in $N_{p'}$. For this, one must show that $f'$ augments flow on valid edges in $N_{p'}$, i.e., $f'_{ij} > 0$ only if $(i,j) \in N_{p'}$. After that, it is easy to see that $f'$ satisfies all capacity and flow conservation constraints. Since $f'_{ij} > 0$ only if $f_{ij} > 0$ (by our update rules), it suffices to argue that no positive flow carrying edge disappears during the price adjustment, i.e., if $f_{ij} > 0$, then $(i,j) \in N_{p'}$. To prove this claim, we first look at the set of edges that may disappear from the MBB graph during price update - Since the prices of the goods in $\Gamma(S)$ increases, only the edges incident to $\Gamma(S)$ may cease to be an MBB edge for an agent. However, all agents in $S$ have MBB edges only to goods in $\Gamma(S)$ and since the price of all these goods increases by the same factor, none of the edges from $S$ to $\Gamma(S)$ disappear. The only edges that may disappear are the ones from agents outside $S$ to goods in $\Gamma(S)$. We need to show that these edges were carrying zero flow. 

To this end, note that $f$ was the flow that minimized the $\ell_2$-norm of the surpluses of the agents in $N_p$. Suppose $i' \notin S$ has a positive flow to some good $j \in \Gamma(S)$, i.e., $f_{i'j} \geq \varepsilon > 0$. Since $j\in \Gamma (S)$, there must be an $i\in S$ with an MBB edge to $j$. Since $S$ is the set of agents with highest surplus, we can conclude that $r_f(i) > r_f(i')$. Then, we can push more flow from $i$ to $j$ and push back the same amount of flow from $j$ to $i'$ continuously. Note that in this process the $\ell_1$-norm of the surplus remains the same while the $\ell_2$-norm decreases (as two surpluses $r_f(i)$ and $r_f(i')$ are moving closer to each other), which is a contradiction. Thus $f_{i'j} = 0$.

Now, we are ready to argue why the above procedure should converge. The price and flow update influences the surpluses of the agents. We show that the $\ell_1$-norm of the surplus monotonically decreases. Note that, after a price adjustment phase,

{\small
\begin{align*}
 \sum_{i \in [n]} r_{f'}(i) &= \sum_{i \in [n]} (1 - \sum_{j \in [m]} f'_{ij})= \sum_{ j \in [m]} (p'_j - \sum_{i \in [n]} f'_{ij} ) \quad \quad \quad \quad \quad \quad \quad \quad \quad \quad \quad \quad \quad \quad (\sum_{j \in [m]} p'_j = n)\\
                            &= \sum_{ j \in \Gamma(S)} (p'_j - \sum_{i \in [n]} f'_{ij} ) +  \sum_{ j \notin \Gamma(S)} (p'_j - \sum_{i \in [n]} f'_{ij} ) = x\sum_{ j \in \Gamma(S)} (p_j - \sum_{i \in [n]} f_{ij} ) +  y\sum_{ j \notin \Gamma(S)} (p_j - \sum_{i \in [n]} f_{ij} )
\end{align*}
\begin{align*}
                            &= y\sum_{ j \notin \Gamma(S)} (p_j - \sum_{i \in [n]} f_{ij} ) \quad \quad \quad \quad \quad \quad\quad \quad \quad \quad \quad \quad \quad \quad \quad \text{(goods in $\Gamma(S)$ are fully sold)}\\
                            &< \sum_{ j \in [m]} (p_j - \sum_{i \in [n]} f_{ij} ) = \sum_{i \in [n]} (1 - \sum_{j \in [m]} f_{ij})=\sum_{i \in [n]} r_f(i). \quad \quad\quad\quad\quad\quad (y < 1, \sum_{j \in [m]} p_j = n)
                            %&=\sum_{i \in [n]} (1 - \sum_{j \in [m]} f_{ij}) &(\sum_{j \in [m]} p_j = n)\\
                            %&=\sum_{i \in [n]} r_f(i).
\end{align*}
}

Furthermore, every flow that minimizes the $\ell_2$-norm of the surpluses is also a maximum flow~\cite{DevanurPSV08}. Thus, after the flow update phase,  the $\ell_1$-norm of the surpluses does not increase, implying that after each iteration the $\ell_1$-norm of the surpluses of the agents decreases.

To get faster convergence, many algorithms show that the $\ell_2$-norm of the surpluses decreases significantly. Intuitively it is equivalent to arguing that surpluses of the agents are getting more balanced, and this is straightforward to see. Observe that
\begin{itemize}
\setlength\itemsep{0em}
	\item the $\ell_1$-norm decreases,
	\item the surplus of the agents in $S$ (the high surplus agents) decreases, and % as their outflow increases and 
	\item the surplus of the agents in $[n] \setminus S$ (the low surplus agents) increases. % as their outflow decreases.
\end{itemize} 
%Also recall that even after price adjustments the surplus of the agents is $S$ is not less than the ones outside $S$. Thus, the high surpluses move closer to the low surpluses and the sum of the surpluses decreases. 
A proper analysis can show that the $\ell_2$-norm decreases as well. The decrease in the $\ell_2$-norm is used to prove polynomial convergence in the exchange model too ~\cite{DuanM15, DuanGM16, ChaudhuryM18}.

\paragraph{Non-monotone surpluses for CE with chores.} We now elaborate the fundamental problem in generalizing % that arises when one tries to generalize 
the overall structure of the above approach to the setting with {\em chores}. %combinatorial algorithm for finding CE with goods to the setting with chores. 

At CE, since every agent consumes only her Min-Pain-per-Buck (MPB) chores, we can define the MPB-network analogously to the MBB-network for goods. That is, network $N_p$ at prices $p$ where only difference is that edge $(i,j)$ is present if $\tfrac{d_{ij}}{p_j} =min_k \tfrac{d_{ik}}{p_k}=\MPB_i$.
%With chores, clearly we define the Minimum Pain per Buck network (MPB network) analogously to a MBB network.  
Like earlier, then the goal is to find a set of prices $p$ at which the maximum flow in $N_p$ is $n = \sum_{j \in [m]} p_j$. Given a flow $f$ in $N_p$, we define the surplus of an agent $r_f(i) = 1 - f_{si} = 1 - \sum_{j \in [m]}f_{ij}$. Now, we wish to come up with an allocation and price adjustment rules that help decrease the $\ell_1$ or $\ell_2$-norm of the surpluses. At a fixed price, we always set the allocation to the one that minimizes the $\ell_2$-norm of the surpluses. %This helps us reduce the $\ell_2$-norm and also makes the MPB graph less sensitive to price adjustment. 
Now, we show the crucial problem during a price adjustment phase. Let $S$ be the set of agents with highest surplus and $\Gamma(S)$ be the set of chores that are their neighbors in $N_p$. Note that, in case of chores, agents prefer {\em higher prices} since they fetch better earning. 
\begin{itemize}
	\item \textbf{Increasing prices of chores in $\Gamma(S)$ and decreasing prices of chores outside $\Gamma(S)$:} %Note that, in case of chores, agents prefer higher priced chores since they fetch better earning. Therefore, 
%	This is counter-intuitive, since it will further increase the demand of chores in $\Gamma(S)$ that are already {\em in-demand}. 
	Increasing prices of chores $\Gamma(S)$ ({\em in-demand} chores) is counter-intuitive, since it will further increase their demand. 
	Formally, after the flow-adjustment step, each agent $i\notin S$ has flow only to the chores outside of $\Gamma(S)$, but may have an MPB edge to a $j\in \Gamma(S)$. Now if prices of chores in $\Gamma(S)$ are increased, and of those outside $\Gamma(S)$ are decreased, then $i$ would strictly prefer $j$ over any chore outside of $\Gamma(S)$. Hence all of her flow-carrying edges will disappear, and to maintain a valid flow we will have to reduce flow on $(s,i)$ edge accordingly which will significantly increase the surplus of agent $i$ and the total surplus. 
%	there could be an MPB edge $(i,j)$ from $i\notin S$ to $j\in \Gamma(S)$ (carrying zero flow). And because of this, increasing prices of chores in $\Gamma(S)$ and decreasing prices of those not in $\Gamma(S)$, all outgoing edges from agent $i$ to chores outside $\Gamma(S)$ will disappear and this may significantly increase the surplus of agent $i$.
	 
	\item \textbf{Decreasing prices of chores in $\Gamma(S)$ and increasing prices of chores outside $\Gamma(S)$:} While this definitely seems more intuitive, the $\ell_1$ and the $\ell_2$-norm of the surpluses can still increase. So we update the prices and flows as in equations~\eqref{priceeq} to~\eqref{floweq3} with $x < 1$ and $y > 1$ this time. Then using the same derivation as in the goods case and the fact that $y>1$ we get that $\sum_{i \in [n]} r_{f'}(i) > \sum_{i \in [n]} r_f(i)$. That is, the $\ell_1$-norm of the surpluses increases (diverges).
%	\begin{align*}
%	   \sum_{i \in [n]} r_{f'}(i) &= \sum_{i \in [n]} (\eta_i - \sum_{j \in [m]} f'_{ij})\\
%	   &= \sum_{ j \in [m]} (p'_j - \sum_{i \in [n]} f'_{ij} )  &(\sum_{j \in [m]} p'_j = \sum_{i \in [n]} \eta_i)\\
%	   &= \sum_{ j \in \Gamma(S)} (p'_j - \sum_{i \in [n]} f'_{ij} ) +  \sum_{ j \notin \Gamma(S)} (p'_j - \sum_{i \in [n]} f'_{ij} )\\
%	   &= x\sum_{ j \in \Gamma(S)} (p_j - \sum_{i \in [n]} f_{ij} ) +  y\sum_{ j \notin \Gamma(S)} (p_j - \sum_{i \in [n]} f_{ij} )\\
%	   &= \sum_{ j \in \Gamma(S)} (p_j - \sum_{i \in [n]} f_{ij} ) + y\sum_{ j \notin \Gamma(S)} (p_j - \sum_{i \in [n]} f_{ij} ) &\text{(goods in $\Gamma(S)$ are fully sold)}\\
%	   &> \sum_{ j \in [m]} (p_j - \sum_{i \in [n]} f_{ij} )\\
%	   &=\sum_{i \in [n]} (\eta_i - \sum_{j \in [m]} f_{ij}) &(\sum_{j \in [m]} p_j = \sum_{i \in [n]} \eta_i)\\
%	   &=\sum_{i \in [n]} r_f(i).
%	\end{align*}
\end{itemize}

	A possible workaround for the second bullet above is to allow negative surpluses (or equivalently relaxing the capacity constraints on the flow on the edges from $s$ to $[n]$). This will always ensure that $f$ will saturate all edges from $t$ to $[m]$ or equivalently $f$ has total flow-value of $\sum_{j \in [m]}p_j$. This way, we have sum of the surpluses is always zero. However, our goal now becomes to get the $\ell_2$-norm of the surpluses to be zero (as this will ensure that each agent $i$ earns exactly $1$ through her MPB chores). But observe that the $\ell_2$-norm may increase: Since we are decreasing the outflow from agents in $S$ (who are the high-surplus agents), the surplus of each agent in $S$ increases. Similarly, since we are increasing the outflow from each agent outside $S$ (the low surplus agents), the surplus of the agents outside $S$ decreases. As such, the gap between the high surplus and the low surplus increases, and in turn the $\ell_2$-norm increases. % as the $\ell_1$-norm remains constant. This shows that the $\ell_2$-norm increases. 
	In fact, this is the main bottleneck in using any potential that ``balances" the surpluses.
	\medskip

\paragraph{Overcoming non-monotone surpluses: Nash-welfare.} We show how to overcome the problem of non-monotone surpluses. Decreasing prices of the high-demand chores and increasing that of the low-demand chores is a natural price adjustment scheme. Our solution lies in identifying the correct convergence analysis.

Recall that when we decrease the prices of the high-demand chores and increase the price of the low demand chores, the $\ell_2$-norm of the surplus increases as this process increases the gap between the surpluses of the agents in $S$ and the ones outside $S$ while keeping sum of surpluses the same. However, such a price adjustment can also bring new MPB edges (say $(i,j)$) in the graph from agents in $S$ to chores outside of $\Gamma(S)$. Let $i'$ be an agent outside $S$ that has positive outflow to $j$. Note that at this point we can definitely balance the surplus by pushing some flow from $i$ to $j$ and pushing back some flow from $j$ to $i'$. What is unclear is whether we can balance the flow sufficient enough to compensate for the increase in $\ell_2$-norm prior to the appearance of the MPB edge! This observation leads us to our new potential function. 

\begin{center}
	\emph{We want a potential function which does not change during the price adjustment, but improves during the allocation update with the help of the new MPB edge.}
\end{center}

We define our potential function as the product of the disutilities: $\prod_{i \in [n]} D_i(x_i)$. Crucially observe that the allocation $x_{ij} = f_{ij}/p_j$ does not change during the price adjustment as $f_{ij}$ and $p_j$ are scaled by the same scalar. Therefore, our potential does not change during price update. We now make subtle changes in the price and the allocation adjustment process and then argue about the convergence of our algorithm.

Given a set of prices $p$, let $N_p$ be the MPB flow network. We make a subtle change in the capacities of the edges. We make the capacity of all edges from $s$ to $[n]$ infinite, implying that any max-flow will push a total flow of $\sum_{j \in [m]} p_j$ saturating all edges from $[m]$ to $t$. Given a flow $f$ in $N_p$, let $\of_f(i)$ denote the total outflow from agent $i$, i.e., $\of_f(i) = \sum_{j \in [m]} f_{ij}$.  Our goal now is to determine prices $p$ and a max flow such that $\of_f(i) = 1$ for all $i\in [n]$ and $\sum_{j \in [m]} p_j =n$. Note that we can write $D_i(x_i) = \MPB_i \cdot \of_f(i)$. Therefore, we can re-write our potential $\prod_{i \in [n]} D_i(x_i) = \prod_{i \in [n]} (\MPB_i \cdot \of_f(i)) = \prod_{i \in [n]} \MPB_i \cdot \prod_{i \in [n]} \of_f(i)$. 

Let $S$ be the set of agents with smallest value of $\of_f(i)$ and let $\Gamma(S)$ be the set of chores adjacent to $S$ in $N_p$. We adjust the allocation and prices as follows,
	\medskip
%\begin{itemize}

	\noindent\emph{Flow adjustment.} We do not change prices and set flow to the one that maximizes $\prod_{i \in [n]} \of_f(i)$.  
	\smallskip

	\noindent\emph{Price adjustment.} We decrease the prices of the chores in $\Gamma(S)$ by factor of $x<1$ until a new MPB edge appears from some agent in $S$ to a chore outside of $\Gamma(S)$. And we keep a valid flow. 
	
	\[p'_{j} =  \begin{cases} 
	              xp_{j} & j \in \Gamma(S) \\
	              p_{j} & j \notin \Gamma(S) 
	\end{cases} \addtag \label{cpriceeq}
	\]

	\[
	\begin{array}{ccc}
	f'_{ij} =  \begin{cases} 
				   xf_{ij} & j \in \Gamma(S) \\
				   f_{ij} & j \notin \Gamma(S) 
				\end{cases} 
& 	
	f'_{jt} =  \begin{cases} 
					xf_{jt} & j \in \Gamma(S) \\
					f_{jt} & j \notin \Gamma(S) 
				\end{cases} 		
&	
	
	f'_{si} =  \begin{cases} 
					xf_{si} & i \in S \\
					f_{si} & i \notin S 
				\end{cases} \addtag  \label{cfloweq3}
	\end{array}
	\]

\noindent\emph{Correctness.} Crucially, we can still argue that $f'$ is a valid flow in $N_{p'}$ like earlier.\footnote{The most crucial argument is to show that all positive flow carrying edges do not disappear during the price adjustment, i.e., all $(i,j)$ such that $f_{ij} > 0$ are present in $N_{p'}$ (equations of \eqref{cfloweq3} would then automatically imply that $f'_{ij} > 0$ only if $(i,j) \in N_{p'}$). Since we only decrease the prices of the chores in $\Gamma(S)$, the only edges that can disappear are the ones that are incoming to $\Gamma(S)$ from agents outside $S$. Post flow-adjustment step, these edges will carry zero flow. Otherwise we have a chore $j \in \Gamma(S)$, an agent $i \in S$ and an agent $i' \notin S$ such that $f_{ij} > \varepsilon >0$ and $f_{i'j} > \varepsilon >0$. Note that we can push more flow from $i$ to $j$ and push back some flow from $j$ to $i'$ such that $\of_f(i)$ and $\of_f(i')$ move closer to each other while keeping the sum of outflows the same, i.e., $\sum_{i \in [n]} \of_f(i)$ remains same. This improves $\prod_{i \in [n]} \of_f(i)$, which is a contradiction. Thereafter, it is straightforward to see that the $f'$ satisfies flow conservation and capacity constraints in $N_{p'}$.} To show convergence, we show that our potential improves after every price and flow adjustment. First note that the potential does not change during the price adjustment: Let $x$ and $x'$ be the allocations corresponding to the flows $f$ and $f'$ respectively. Note that $x_{ij} = f_{ij}/p_j$ and $x'_{ij} = f'_{ij}/p'_j$ for all $i \in [n]$, $j \in [m]$. From equations~\eqref{cpriceeq} and~\eqref{cfloweq3}, we have $f'_{ij}/p'_j = f_{ij}/p_j$ for all $i \in [n]$, $j \in [m]$, implying that $x = x'$. Therefore $D_i(x'_i) = D_i(x_i)$ and the potential does not change during price adjustment.

Now we show that the potential improves during flow adjustment. Note that, since our potential is $\prod_{i \in [n]} \MPB_i \cdot \prod_{i \in [n]} \of_{f'}(i)$, where  $\prod_{i \in [n]} \MPB_i$ does not change during the flow adjustment, it suffices to show that there is a flow $f''$ in $N_{p'}$ such that $\prod_{i \in [n]} \of_{f''}(i) > \prod_{i \in [n]} \of_{f'}(i)$.  Let $(i,j)$ with $i \in S$ and $j \notin \Gamma(S)$ be the new MPB edge that appears in $N_{p'}$ when we update the flow from $f$ to $f'$. Let $i' \notin S$ be an agent that had positive outflow to $j$. Note that $\of_f(i) < \of_f(i')$ and the flow and price adjustments outlined in equations~\eqref{cpriceeq} and~\eqref{cfloweq3} show that the outflow of agent $i$ decreases and that of $i'$ remains the same. This implies that $\of_{f'} (i)<  \of_f(i) < \of_f(i') = \of_{f'}(i')$. Also, $f'_{i'j} > 0$ as $f_{i'j} > 0$. Let $f''$ be the flow obtained from $f'$ by pushing more flow from $i$ to $j$ and  pushing back the same amount of flow from $j$ to $i'$, while still maintaining $\of_{f''}(i) \leq \of_{f''}(i')$. Note that $\of_{f'}(i)< \of_{f''} (i)\leq \of_{f''} (i') < \of_{f'} (i')$ and $ \of_{f''} (i) + \of_{f''} (i') = \of_{f'}(i) + \of_{f'} (i')$. Therefore, we have $\prod_{i \in [n]} \of_{f''}(i) > \prod_{i \in [n]} \of_{f'}(i)$ (as the outflows are getting more balanced).

We refer the reader to Section~\ref{algorithm} for a more detailed and rigorous explanation of the combinatorial algorithm. The remaining bulk of the effort is put into proving fast (polynomial time) convergence when either $(i)$ the disutility values are $\alpha$-rounded or $(ii)$ we are only aiming for $(1-\varepsilon)$-CEEI. Lastly, note that our algorithm moves from a Nash welfare maximizing allocation at a particular MPB configuration to a Nash welfare maximizing allocation at another MPB configuration. In this process, it also strictly improves the Nash welfare. This shows that finding a CEEI (exactly) is in PLS. Since~\cite{ChaudhuryGMM21} show that this problem is in PPAD, we can conclude that this problem lies in PPAD $\cap$ PLS $=$ CLS.

\subsection{Discrete Setting with Bivalued Preferences}\label{sec:bi}
In the discrete chore division problem, a set of \emph{indivisible} chores needs to be divided among agents in a fair and efficient manner. In this setting, envy-freeness up to one chore (EF1) and Pareto optimality (PO) are popular fairness and efficiency notions, respectively. 
\medskip

\noindent{\bf EF1 and PO.} An allocation $x$ is said to be EF1 if for every pair of agents $(i, i')$, there exists a chore $c(i,i')\in x_i$ such that $D_i(x_i \setminus c(i,i')) \le D_i(x_{i'})$, i.e., the disutility of agent $i$ from $x_i$ after the removal of a chore $c(i,i')$ is at most the disutility of $i$ for $x_{i'}$.\footnote{In the discrete setting, we assume $x$ to be integral, i.e., $x_{ij} \in \{0, 1\}, \forall (i,j)$.}  

An allocation $x$ is said to Pareto dominates another allocation $x'$ if $D_i(x_i) \le D_i(x_i'), \forall i$ with a strict inequality for at least one agent. We say that $x$ is PO if no other allocation $x'$ dominates it. 
\medskip

A major open question in the discrete setting is whether there always exists an allocation that is both EF1 and PO. Very recently,~\cite{GargMQ22,EbadianPS22} answered this question affirmatively when agents have bivalued preferences, i.e., $d_{ij} \in \{1, \beta\}$, for some $\beta >1$. However, their analysis is quite involved, especially due to lack of a potential function for the convergence. Next, we briefly sketch how our new insight on the potential function of the product of disutility values implies a straightforward argument for the convergence. 

First, we briefly sketch the approaches of~\cite{GargMQ22,EbadianPS22}. They both utilize the notions of \emph{integral} competitive equilibrium (CE) and \emph{price envy-freeness up to one chore} (pEF1)~\cite{BarmanKV18}. Since pEF1 implies both EF1 and PO, we only need to find an allocation that is pEF1. For that, they start with a trivial CE $(x,p)$ and achieve pEF1 by \emph{directly} transferring chores (on MPB edges) from a \emph{big earner} to a \emph{least earner}.\footnote{We note that this direct transfer is always possible in the case of bivalued preferences, which may not be the case for general additive preferences, and hence the problem for general additive preferences is still open.}  The main challenge in both the papers is to show that this process converges. Let $p(x_i) = \sum_{j\in x_i} p_j$ is the total payment of doing chores in $x_i$. Big earner, say $b$, is an agent who earns the maximum after the removal of a chore, i.e., $b = \arg\max_{i} \min_{c\in x_i} p(x_i\setminus c)$. Similarly, least earner, say $l$, is an agent who earns the least, i.e., $l = \arg\min_i p(x_i)$. If $p(x_b) - p_c \le p(x_l)$, where $c$ is a chore in $x_b$, then $x$ is a pEF1 allocation. 

Using our new insight on the potential function of the product of disutilities, we next show that it strictly increases in the above setting as well, thereby giving a straightforward proof of the convergence. Since any form of price-update doesn't change the potential function (because the allocation remains the same), we only need to show that it strictly increases during an allocation update. Further, the above discussion implies that we only need to show that if $p(x_b) - p_c > p(x_l)$ then transferring $c$ from $b$ to $l$ results in an increased potential function. Let $x'$ be the new allocation where $x'_b = x_b \setminus c$ and $x'_l = x_l \cup c$. Recall from Section~\ref{sec:alg} that we can write $D_i(x_i) = \prod_i MPB_i \cdot \prod_i p(x_i)$. Since $MPB_i$'s doesn't change during allocation change, we need to show that 
\[ \prod_i p(x_i) < \prod_i p(x'_i) ,\] 
which is
\[ p(x_b)p(x_l) < p(x_b')p(x_l').\] 
Expanding it, we get
\[ p(x_b) p(x_l) < (p(x_b) - p_c) (p(x_l) + p_c).\]
Simplifying it, we get
\[ p_c < p(x_b) - p(x_l), \]
which is clearly true if $x$ is not an pEF1 allocation. Since bivalued preferences are already $(\beta-1)$-rounded utilities, our running time analysis in Appendix~\ref{algorithm} can be adapted to show convergence polynomial in $n, m$ and $1/ (\beta -1)$.

\subsection{PPAD-Hardness for the Exchange Model}\label{mainres3}
In this section, we sketch the proof of Theorem~\ref{mainthm2intro}, i.e., computing a $(1-1/poly(n))$-approximate CE of an exchange instance is PPAD-hard, even if the instance satisfies sufficiency conditions ($SC_1$ and $SC_2$) of Theorem~\ref{sufficientcondition} (see Section \ref{ppadhardness} for the formal and detailed proof). We give a polynomial time reduction from the problem of computing a {\em Nash equilibrium} of a {\em normalized polymatrix game}, a known PPAD-hard problem~\cite{chen2017complexity}.

A polymatrix game is represented by a game graph where each node is a player who plays a two-player game with each of her neighbors. She has to play the same strategy with each of her neighbors and her payoff is the sum of the payoffs on each of her incident edges. If there are $n$ players and each of them has exactly two strategies to choose from, then such a game can be represented by $2n \times 2n$ matrix. When thought of it as $n \times n$ block matrix, where each block is $2 \times 2$ matrix, then $(i,j)^{\mathit{th}}$ block is the payoff matrix of player $j$ for the game on edge $(i,j)$. 
\begin{problem}\textbf{(Normalized Polymatrix Game)~\cite{chen2017complexity}}\\
	\label{prob:ne}
	\textbf{Given}: A $2n \times 2n$ rational matrix $\M$ with every entry in $[0,1]$ and $\M_{i,2j-1} + \M_{i,2j} = 1$ for all $i \in [2n]$ and $j \in [n]$ .\\
	\textbf{Find}: $1/n$-approximate Nash equilibrium: Strategy vector $x \in \mathbb{R}^{2n}_{\geq 0}$ such that $x_{2i-1} + x_{2i} = 1$ and 
	\begin{align*}
	& x^T \cdot \M_{*,{2i-1}} > x^T \cdot \M_{*,2i} + \tfrac{1}{n} \implies x_{2i} = 0.\\
	& x^T \cdot \M_{*,{2i}} > x^T \cdot \M_{*,2i-1} + \tfrac{1}{n} \implies x_{2i-1} = 0.
	\end{align*}
	where $M_{*,k}$ represents the $k^{\mathit{th}}$ column of the matrix $\M$.	
\end{problem}

Given an instance $I = \langle \M \rangle $ of the polymatrix game, we create an instance of chore division $E(I) = \langle A \cup  B, d( \cdot, \cdot) , w(\cdot , \cdot ) \rangle$, such that given any  CE in $E(I)$, we can determine in polynomial time an equilibrium strategy vector $x$ for $I$. Since we will have to create many different {\em types} of agents and chores, for ease of notation, we use $d(i,j), w(i,j)$ and $p(j)$ instead of $d_{ij}, w_{ij}$ and $p_j$ respectively in this section and in Section \ref{ppadhardness}. %$w(i,j)$ instead of $w_{ij}$ to represent endowments, and $p(j)$ instead of $p_j$ to denote prices. 
%We note that our reduction works even in the setting of $(1 - \tfrac{1}{poly(n)})$-approximate competitive equilibria (see Section \ref{ppadhardness}). However, for simplicity and to convey the main ideas, we stick to the exact competitive equilibria in this section.% \rh{We next give a brief sketch of the proof, while the formal proof with all the details are in Section \ref{ppadhardness}.}

\paragraph{Sketch of the Reduction.} \label{reduction-sketch} The entries of matrix $M$ are encoded into endowments $w(\cdot, \cdot)$ of agents, and the equilibrium strategy vector $x$ of the polymatrix game $I$ can be {\em extracted} from the price vector at a  CE of $E(I).$ The key properties that our hard instance $E(I)$ exhibits are  \emph{pairwise equal endowments}, \emph{fixed earning}, \emph{price equality}, \emph{price regulation} and \emph{reverse ratio amplification} (we will give a precise definition of these properties shortly). These techniques (constructing hard instances exhibiting these properties) have been used earlier to prove PPAD-hardness for the exchange model with only goods when agents have constant elasticity of substitution (CES) utilities~\cite{chen2017complexity} and even for the Fisher model when agents have separable piecewise linear concave (SPLC) utilities~\cite{ChenT09}. However, the challenge is to construct these gadgets, in particular the \emph{reverse ratio amplification gadget} (a brief description will be given towards the end of this section) and make them work together only using linear disutility functions; as clearly this is not possible in case of goods, when agents have linear utility functions.

We now describe our instance $E(I)$. To encode the entries of the matrix $\M$ of the polymatrix game, we introduce two sets of $2n$ chores in $E(I)$, namely $B = \left\{ b_1,b_2, \dots ,b_{2n} \right\}$ and $B' = \left\{ b'_1,b'_2, \dots ,b'_{2n} \right\}$. %Note that in $\M$
%rows/columns $(2i-1)$ and $(2i)$ belong to player $i$. We will use chores $(2i-1)$ and $2i$ from both the sets to capture these, and therefore they will always appear in pairs. 
For each $i,j\in [2n]$, we add an agent $a_{ij}$ who brings $\M_{i,j}$ units of chore $b'_i$. The disutility values of these agents are as follows:
\begin{align*}
\forall i\in[2n], j\in [n],\ 
d(a_{i,2j-1},b_{2j-1}) &= (1- \alpha) &\text{and}&    & d(a_{i,2j-1},b_{2j}) &= (1+ \alpha)\\
d(a_{i,2j},b_{2j-1}) &= (1+ \alpha)   &\text{and}&    & d(a_{i,2j},b_{2j}) &= (1 - \alpha),
\end{align*}
for some infinitesimally small $\alpha > 0$. The disutility values that are not specified are all infinity. Observe that the agents $a_{i,j}$s own chores of $B'$, but can only do chores of $B$. Now, to complete the loop, we introduce another set of $2n$ agents, namely $a'_1,a'_2, \dots a'_{2n}$, where for each  $i \in [2n]$ agent $a'_i$ brings $n$ units of chore $b_i$, and can do only the chores in set $B'$. 
\begin{align*}
\forall i\in[n],\ 
d(a'_{2i-1},b'_{2i-1}) &= (1- \alpha') &\text{and}&    & d(a'_{2i-1},b'_{2i}) &= (1+ \alpha')\\
d(a'_{2i},b'_{2i-1}) &= (1+ \alpha')   &\text{and}&    & d(a'_{2i},b'_{2i}) &= (1 - \alpha'),
\end{align*}
for some small $\alpha'\gg {n^2 \cdot \alpha}$. Next, we add a set of $2n$ agents for whom we will ensure certain {\em fixed earnings}, namely $a^f_1,\dots, a^f_{2n}$, and they can do only chores in set $B$. 
\begin{align*}
d(a^f_{2i-1},b_{2i-1}) &= (1- \alpha) &\text{and}&    & d(a^f_{2i-1},b_{2i}) &= (1+ \alpha)\\
d(a^f_{2i},b_{2i-1}) &= (1+ \alpha)   &\text{and}&    & d(a^f_{2i},b_{2i}) &= (1 - \alpha),
\end{align*}

The instance $E(I)$ has additional agents and chores, and none of them has finite disutility towards the chores in set $B$. At any  CE, our instance satisfies the following five key properties. 

\begin{itemize}
	\item \emph{Pairwise equal endowments:} In $E(I)$, the total endowment of chore $b_{2i-1}$ equals the total endowment of the chore $b_{2i}$, and similarly total endowment of $b'_{2i-1}$ equals total endowment of $b'_{2i}$. Also the total endowments of each of these chores is $\mathcal{O}(n)$; from the above construction it is easy to see that this holds for chores in set $B$ and $B'$. 
	\item \emph{Fixed earning:} In any  CE, for each $i \in [2n]$, we have the total earning of agent $a^f_i$ is $(1- \alpha') \cdot (2n - \sum_{j \in [2n]} \M_{j,i})$.  
	
	\item \emph{Price equality:} In any  CE, the sum of prices of the chores $b_{2i-1}$ and $b_{2i}$ is the same for all $i \in [n]$, and it is equal to the sum of prices of the chores $b'_{2i-1}$ and $b'_{2i}$ for all $i \in [n]$. Since the prices of the chores at a  CE are scale invariant, we can assume, without loss of generality, that for all $i \in [n]$, we have $p(b_{2i-1}) + p(b_{2i}) = p(b'_{2i-1}) + p(b'_{2i}) = 2$.  
	\item \emph{Price regulation: } In any  CE, for all $i \in [2n]$ we have 
	\begin{align*} 
	\frac{1 - \alpha}{1 + \alpha} \leq \frac{p(b_{2i-1})}{p(b_{2i})} \leq \frac{1 + \alpha}{1 - \alpha}& &\text{and}& &\frac{1 - \alpha'}{1 + \alpha'} \leq \frac{p(b'_{2i-1})}{p(b'_{2i})} \leq \frac{1 + \alpha'}{1 - \alpha'}.
	\end{align*}
	%	for some infinitesimally small $\beta \gg {n^2 \cdot \alpha}$.
	\item \emph{Reverse ratio amplification:} In any  CE, for all $i \in [2n]$, if $\tfrac{p(b_{2i-1})}{p(b_{2i})} = \tfrac{1 + \alpha}{1 - \alpha}$, then we have $\tfrac{p(b'_{2i-1})}{p(b'_{2i})} = \tfrac{1 - \alpha'}{1 + \alpha'}$ and similarly when  $\tfrac{p(b_{2i-1})}{p(b_{2i})} = \tfrac{1 - \alpha}{1 + \alpha}$,  we have $\tfrac{p(b'_{2i-1})}{p(b'_{2i})} = \tfrac{1 + \alpha'}{1 - \alpha'}$.									
\end{itemize}

We show that our instance $E(I)$ satisfies both $SC_1$ and $SC_2$ conditions of Theorem~\ref{sufficientcondition}, and thereby it admits a CE. Using the above five properties satisfied by a CE, we next describe how the prices at CE give the equilibrium strategy vector of the polymatrix game $I$. Given a  CE price vector $p$ of $E(I)$, construct a vector $x$ as follows. 
\begin{equation}\label{above-eq}
%\begin{align*}
\forall i\in [2n],\ \ \ x_i = \frac{p(b'_{i}) - (1-\alpha')}{2 \cdot \alpha'}
%\end{align*}
\end{equation}

We will show that $x$ satisfies the Nash equilibrium conditions described in Problem \ref{prob:ne} for instance $I$, thereby completing the reduction. First, let us argue why $x$ is a valid strategy profile for the game, i.e., for each $i \in [n]$, we have $x_{2i-1} \geq 0$, $x_{2i} \geq 0$,  and $x_{2i-1}+x_{2i}=1$. Using the {\em price equality} and {\em price regulation} properties, it follows that for each $i\in [2n]$, $(1-\alpha') \le p(b'_i) \le (1+\alpha')$. This immediately implies $x_i\ge 0,\ \forall i \in [2n]$. Furthermore, for each $i\in [n]$, we have $x_{2i-1}+x_{2i} = \tfrac{p(b'_{2i-1})  + p(b'_{2i}) - 2(1-\alpha')}{2 \cdot \alpha'} = \tfrac{2\alpha'}{2\alpha'} = 1$.

Next, we will show that $x$ satisfies the equilibrium conditions, namely, if $x^T \cdot \M_{*,{2i}} > x^T \cdot \M_{*,2i-1} + \tfrac{1}{n}$, then $ x_{2i-1} = 0$; the proof for the symmetric condition is similar. To show that $x_{2i-1} = 0$, it suffices to show that $p(b'_{2i-1})=(1-\alpha')$ (by \eqref{above-eq}). So let us assume that $x^T \cdot \M_{*,{2i}} > x^T \cdot \M_{*,2i-1} + \tfrac{1}{n}$.  The positive correlation between $x$ and $p$ imply that the total endowment money of the agents $\{a_{j,2i}\ |\ j \in [2n]\}$ is larger than that of the agents $\{a_{j,2i-1}\ |\ j \in [2n]\}$. Since the agents $\{a_{j,2i}\ |\ j \in [2n]\}$ prefer chore $b_{2i}$ to $b_{2i-1}$ and the agents $\{a_{j,2i-1}\ |\ j \in [2n]\}$ prefer chore $b_{2i-1}$ to $b_{2i}$, the difference in their endowment money would further enforce that chore $b_{2i}$ is higher priced than chore $b_{2i-1}$. Next we formalize this intuition, and how that leads to $p(b'_{2i-1})=(1-\alpha')$. 

%further implying a higher endowment money of the agents that prefer chore $b_{2i}$ to $b_{2i-1}$ as the agents $\cup_{j \in [2n]}a_{j,2i}$ have a disutility of $1-\alpha$ towards $b_{2i}$ and $1+\alpha$ towards $b_{2i-1}$ and similarly the agents $\cup_{j \in [2n]}a_{j,2i-1}$ have a disutility of $1-\alpha$ towards $b_{2i-1}$ and $1+\alpha$ towards $b_{2i}$. We formalize this intuition.  

The only agents who can consume chores $b_{2i-1}$ and $b_{2i}$ are the ones in $\{a_{j,2i}, a_{j, 2i-i}\ |\ j \in [2n]\} \cup \{a^f_i\ |\ i \in [2n]\}$. Out of all these agents, the agents $A(b_{2i}) = \{a_{j,2i}\ |\ j \in [2n]\}\cup \{a^f_{2i}\ |\ i \in [n]\}$ are the agents that have $(1-\alpha)$ disutility towards $b_{2i}$ and $(1+\alpha)$ disutility towards $b_{2i-1}$. Symmetrically define set $A(b_{2i-1})$. The total endowment money of the agents in $A(b_{2i})$ (say $P(b_{2i})$) is,
\begin{align*}
&=\sum_{j \in [2n]} \M_{j,2i} \cdot p(b'_j) + (1- \alpha') \cdot (2n - \sum_{j \in [2n]} \M_{j,2i}) &\text{(fixed earning of $a^f_{2i}$)}\\
&=\sum_{j \in [2n]} \M_{j,2i} \cdot (2\alpha'\cdot x_{j} + (1-\alpha')) +  (1- \alpha') \cdot (2n - \sum_{j \in [2n]} \M_{j,2i}) &\text{(by equation~\eqref{above-eq})}\\ 
&=\sum_{j \in [2n]} 2\alpha'\cdot x_{j} \cdot \M_{j,2i} + (1-\alpha') \cdot\sum_{j \in [2n]}  \M_{j,2i} +  (1- \alpha') \cdot (2n - \sum_{j \in [2n]} \M_{j,2i})\\          
&= 2\alpha'x^T \cdot \M_{*,2i} +  2n \cdot (1-\alpha').               
\end{align*}

Similarly, the total endowment price of the agents in $A(b_{2i-1})$ (say $P(b_{2i-1})$) is $2\alpha'x^T \cdot \M_{*,2i-1} +  2n \cdot (1-\alpha')$ which is less than $P(b_{2i})$. Recall again that, agents $A(b_{2i}) \cup A(b_{2i-1})$ can only consume chores $b_{2i}$ and $b_{2i-1}$, and no other agents can consume these two chores. We now prove that if $A(b_{2i})$ earn all their endowment money of $P(b_{2i})$ by only consuming (doing) chore $b_{2i}$, then the price of $b_{2i}$ ($p(b_{2i})$) has to be significantly high, leading to the violation of the \emph{price-regulation} property: By \emph{pairwise equal endowments} property, we have equal linear total endowments of both $b_{2i}$ and $b_{2i-1}$, say $cn$. Therefore, if the agents in $A(b_{2i})$ earn all their endowment money of $P(b_{2i})$ only from chore $b_{2i}$, we have $p(b_{2i}) \geq P(b_{2i}) / (cn)$, and $p(b_{2i-1}) \leq P(b_{2i-1}) / (cn)$, implying that
{\small 
\[
%\begin{align*}
p(b_{2i}) - p(b_{2i-1}) \geq \frac{P(b_{2i}) - P(b_{2i-1})}{cn} =  \frac{2\alpha' \cdot (x^T \cdot \M_{*,2i} - x^T \cdot \M_{*,2i})}{cn}\geq \frac{2 \alpha' \cdot 1/n}{cn} \gg 2\alpha \quad \text{ (as $\frac{\alpha'}{n^2}  \gg \alpha$)}.  
%\end{align*} 
\]
}

Since $p(b_{2i}) + p(b_{2i-1}) = 2$ (by the \emph{price equality} property), this immediately implies that $p(b_{2i})/ p(b_{2i-1})$ $>  (1+\alpha)/ (1-\alpha)$ (violation of the \emph{price regulation} property). This implies that the agents in $A(b_{2i})$ earn some of their endowment money of $P(b_{2i})$ from the chore $b_{2i-1}$ as well. Since all agents in $A(b_{2i})$ have a disutility of $1-\alpha$ for the chore $b_{2i}$ and $1+\alpha$ for the chore $b_{2i-1}$, agents in $A(b_{2i})$ can do some amount of $b_{2i-1}$ only if  $\tfrac{p(b_{2i-1})}{p(b_{2i})} \geq \tfrac{1+\alpha}{1-\alpha}$. Then, the \emph{price regulation} property dictates that $\tfrac{p(b_{2i-1})}{p(b_{2i})} = \tfrac{1+\alpha}{1-\alpha}$.  Finally, by \emph{reverse ratio amplification} property, we have $\tfrac{p(b'_{2i-1})}{p(b'_{2i})} = \tfrac{1-\alpha'}{1+\alpha'}$. This implies $p(b'_{2i-1})=(1-\alpha')$ and thereby $x_{2i-1} = 0$, when $x^T \cdot \M_{*,{2i}} > x^T \cdot \M_{*,2i-1} + \tfrac{1}{n}$. A very symmetric argument will show that if $x^T \cdot \M_{*,{2i-1}} > x^T \cdot \M_{*,2i} + {1} / {n}$, then $ x_{2i} = 0$. 

While it is still intuitive for our choice of $\alpha' \gg n^2 \cdot \alpha $, that we can enforce $\tfrac{p(b_{2i-1})}{p(b_{2i})} = \tfrac{1+\alpha}{1-\alpha}$ when $x^T \cdot \M_{*,{2i}} > x^T \cdot \M_{*,2i-1} + \tfrac{1}{n}$, perhaps the most subtle and  mysterious  gadget in our reduction is the \emph{reverse ratio amplification} gadget. This is also the most crucial gadget as such a gadget cannot be realized in the linear exchange setting with divisible goods (all the other gadgets can be realized). We now briefly sketch the \emph{reverse ratio amplification} gadget.

\paragraph{Reverse Ratio Amplification Gadget.}
We sketch the reverse ratio amplification gadget for $\alpha' = 3\alpha /2$ and in the end elaborate how to make it work for $\alpha' \gg n^2 \cdot \alpha$. For each $i \in [n]$, we add an agent $\overline{a}_i$ that has an endowment of $\delta = n \alpha / 2$ of both chores $b'_{2i-1}$ and $b'_{2i}$ and has a disutility of $1$ for both chores $b_{2i-1}$ and $b_{2i}$ (and $\infty$ for all other chores), implying that $\overline{a}_i$ will earn her endowment money from the more expensive chore among $b'_{2i-1}$ and $b'_{2i}$. We now show that when $\tfrac{p(b_{2i-1})}{p(b_{2i})} = \tfrac{1 + \alpha}{1 - \alpha}$, then we have $\tfrac{p(b'_{2i-1})}{p(b'_{2i})} = \tfrac{1 - \alpha'}{1 + \alpha'}$. So assume that $\tfrac{p(b_{2i-1})}{p(b_{2i})} = \tfrac{1 + \alpha}{1 - \alpha}$. The \emph{price equality} property dictates that $p(b_{2i-1}) = 1+\alpha$ and $p(b_{2i}) = 1 - \alpha$. 

Now recall agents $a'_{2i-1}$ and $a'_{2i}$. They own $n$ units of chores $b_{2i-1}$ and $b_{2i}$ respectively and therefore their respective total endowment are $n \cdot (1+\alpha)$ and $n \cdot (1-\alpha)$. Like earlier, we now prove that if $a'_{2i-1}$ earns her entire endowment money of $n (1+\alpha)$ from the chore $b'_{2i-1}$, then this causes violation of the \emph{price regulation} property: If $a'_{2i-1}$ earns her entire endowment from $b'_{2i-1}$, we have the  price of $b'_{2i-1}$ to be at least $\tfrac{n (1+\alpha)}{c'n} = \tfrac{1+\alpha}{c'}$, where $c'n$ is the total endowment of chores $b'_{2i-1}$ and $b'_{2i}$ (they are equal and linear in $n$ by \emph{pairwise equal endowments} property). Similarly, the only other agents having finite disutility towards the chore $b'_{2i}$ are the agents $a'_{2i}$ and $\overline{a}_i$. Therefore, using $\delta=n \cdot \alpha / 2$, and $p(b'_{2i-1}) + p(b'_{2i}) =2$ by the \emph{price equality} property, we have,

\[
p(b'_{2i}) \leq \frac{n(1-\alpha) + \delta \cdot (p(b'_{2i-1}) + p(b'_{2i}))}{c'n}=\frac{n(1-\alpha) + 2\delta}{c'n} =\frac{n (1-\alpha) + n \alpha}{c'n} =\frac{1}{c'}.
\]

Therefore, we have that $p(b'_{2i-1}) \geq (1+\alpha) / c' > 1/c' \geq p(b'_{2i})$. This implies that the agent $\overline{a}_i$ earns her entire endowment price of $\delta \cdot (p(b'_{2i-1}) + p(b'_{2i})) = 2 \delta$ from the chore $b'_{2i-1}$ (as the disutility to price ratio of $b'_{2i-1}$ is strictly less than that of $b_{2i}$ for $\overline{a}_i$). Thus, we have that $p(b'_{2i-1}) \geq  \tfrac{n(1+\alpha) + 2\delta}{c'n}$ and $p(b'_{2i}) \leq \tfrac{n(1-\alpha)}{c'n}$, implying that, if we set $\alpha'= 3\alpha/2$, then

\[
%\begin{array}{r}
\frac{p(b'_{2i-1})}{p(b'_{2i})} \geq \frac{n \cdot (1+\alpha) + 2\delta}{n \cdot (1-\alpha)}= \frac{n \cdot (1+\alpha) + n\alpha}{n \cdot (1-\alpha)}=\frac{1+ 2\alpha}{1-\alpha}>\frac{1+\alpha'}{1-\alpha'} \quad \quad \quad \quad \text{(as $\delta = n \cdot \alpha / 2$)} 
%\end{array}
\]

Therefore, we have violation of \emph{price regulation}. This implies that at a  CE, agent $a'_{2i-1}$ earns some of her endowment money from the chore $b'_{2i}$. Since $a'_{2i-1}$ has a disutility of $(1-\alpha)$ for $b'_{2i-1}$ and $(1+\alpha)$ for $b'_{2i}$, she can do the chore $b'_{2i}$ only if $p(b'_{2i-1})/p(b'_{2i}) \geq (1-\alpha)/(1+\alpha)$. Then, by the \emph{price regulation property}, we have $p(b'_{2i-1})/p(b'_{2i}) = (1-\alpha)/(1+\alpha)$. The proof for the case $\tfrac{p(b_{2i-1})}{p(b_{2i})} = \tfrac{1 - \alpha}{1 + \alpha} \implies \tfrac{p(b'_{2i-1})}{p(b'_{2i})} = \tfrac{1 + \alpha'}{1 - \alpha'}$ is symmetric.\footnote{Note that we are crucially using the fact that increasing the price of a chore makes it more attractive to an agent (agent $\overline{a}_i$ is only interested in $b'_{2i-1}$ whenever it has a higher price than $b'_{2i}$). This property is not satisfied in the linear exchange setting with divisible goods (in fact this is the fundamental difference between the two settings) where we cannot realize this gadget.}

When $\alpha' \gg n^2 \cdot \alpha$, i.e., when we need significantly higher amplification, we create a sequence of $K \in \mathcal{O}(\log(n))$ sets of chores $B^1 = \{b^1_1, b^1_2, \dots ,b^1_{2n}\}$, $B^2 = \{b^2_1, b^2_2, \dots ,b^2_{2n}\}, \dots,  B^K = \{b^K_1, b^K_2,$ $ \dots ,b^K_{2n}\}$ with $B^1 = B$ and $B^K = B'$. We enforce reverse ratio amplification between every consecutive sets of chores, i.e., for all $\ell \in [K]$ and for all $i \in [n]$, if $\tfrac{p(b^{\ell}_{2i-1})}{p(b^{\ell}_{2i})} = \tfrac{1 + \alpha_{\ell}}{1 - \alpha_{\ell}}$, then we have $\tfrac{p(b^{\ell+1}_{2i-1})}{p(b^{\ell+1}_{2i})} = \tfrac{1 - \alpha_{\ell+1}}{1 + \alpha_{\ell+1}}$ and similarly when  $\tfrac{p(b^{\ell}_{2i-1})}{p(b^{\ell}_{2i})} = \tfrac{1 - \alpha_{\ell}}{1 + \alpha_{\ell}}$, then we have $\tfrac{p(b^{\ell+1}_{2i-1})}{p(b^{\ell+1}_{2i})} = \tfrac{1 + \alpha_{\ell+1}}{1 - \alpha_{\ell+1}}$, where $\alpha_{\ell+1} = 3 \alpha_{\ell}/2$, thereby giving us the desired amplification from $B$ to $B'$.

We re-emphasize that while the concepts of the gadgets used in our reduction are standard~\cite{chen2017complexity, ChenT09}, the key contribution is to realize these gadgets with just linear disutility functions. As a consequence, there are subtle, yet important technical differences to gadgets in~\cite{chen2017complexity, ChenT09}. For instance, our amplification sequence has to be \emph{alternating}, while the one in~\cite{chen2017complexity} is not.

The formal proof with all the details can be found in Section~\ref{ppadhardness}, where we give the complete construction and prove that the instance exhibits all of the above five properties, as well as meets the sufficient conditions (Theorem~\ref{sufficientcondition}) for existence of CE.

\section{Applications and Further Related Work}\label{sec:rw}
The CEEI model is applicable to allocation problems naturally arising %The allocation problems in the CEEI model arise naturally 
in a wide range of real-life settings such as dividing tasks among various team members, deciding teaching assignments between faculty, or splitting liabilities when dissolving a partnership. The exchange model also occurs naturally in many day to day scenarios, e.g., a set of university students teaching each other in a group study, to optimize the time and effort required. At a larger scale, \emph{timebanks}\footnote{\url{https://timebanks.org/}} are such reciprocal service exchange platforms which have around 30,000 to 40,000 users from the United States. In a timebank, individuals from a certain community give services to one another and earn \emph{time credit}. Thereafter, each individual uses their time credit to receive services. Competitive equilibrium (CE) provides a systematic way to do the exchange: it consists of prices (payment)\footnote{Equivalent of time credit in time banks.} for chores and an allocation such that all chores are completely assigned and each agent gets her most preferred bundle (\emph{optimal bundle}) subject to her budget constraint\footnote{Here the budget constraint of an agent is that she has to earn enough to pay for her initial set of chores.}. %according to the price of her initial bundle.

The problem of computing CE has been extensively studied. We only discuss previous work that appears most relevant. For linear CEEI (Fisher) model with goods, the CE set is captured by the Eisenberg-Gale convex program~\cite{EisenbergG59}, which maximizes the Nash welfare defined as the geometric mean of agents' utilities. Later, Shmyrev~\cite{Shmyrev09} obtained another convex program for this problem. ~\cite{ColeDGJMVY17} provides a dual connection between these and other related convex programs. A combinatorial polynomial-time algorithm for computing a CE is given in~\cite{DevanurPSV08}. Later, strongly polynomial-time algorithms are obtained for this problem~\cite{Orlin10,Vegh12}. 

For the linear exchange model with goods, many convex programming formulations were obtained; see~\cite{DevanurGV16} for details. Initial polynomial-time algorithms for this model are based on ellipsoid and interior-point methods~\cite{Jain07,Ye08}. The first combinatorial polynomial-time algorithm is given in~\cite{DuanM15}, which was later improved in~\cite{DuanGM16}. Recently,~\cite{GargV19} gives the first strongly polynomial-time algorithm for this problem.

The equilibrium computation problem in all models with goods is basically PPAD-hard whenever the CE set is non-convex. In the CEEI (Fisher) model, it is polynomial-time for weak gross substitutes (WGS) and homogeneous utility functions~\cite{CodenottiMPV05,Eisenberg61}, and PPAD-hard for separable piecewise linear concave (SPLC) utilities~\cite{ChenT09}. In the exchange model, it is polynomial-time for WGS utilities~\cite{CodenottiMPV05,GargHV21}, and beyond that, it is essentially PPAD-hard~\cite{chen2017complexity,ChenDDT09,CodenottiSVY06,GargMVY17}. More recently,~\cite{ChenCPY22} shows PPAD-hardness for the Hylland-Zeckhauser model~\cite{HyllandZ79}, which is essentially the CEEI model under linear preferences with additional matching constraints. 

For the linear CEEI model with chores,~\cite{BogomolnaiaMSY17} consider the case where all (dis)utility values to be finite. ~\cite{BogomolnaiaMSY17} show that critical points of the geometric mean of agent's disutilities on the (Pareto) efficiency frontier are the CE profiles. By building on this characterization,~\cite{BoodaghiansCM22} give an FPTAS for finding an approximate CE. For the special case of constantly many agents (or chores), polynomial-time algorithms are known for computing a CE~\cite{BranzeiS19,GargM20}. 

For the linear exchange model with chores,~\cite{ChaudhuryGMM21} gives a linear complementarity problem (LCP) formulation and a complementary pivot algorithm for computing a CE. 

%The fair allocation of \emph{indivisible} items is also an intensely studied problem for the case when all items are goods with a few recent exceptions~\cite{AzizRSW17,AzizCL19,HuangL19,AzizCL19a,AzizCIW19,SandomirskiyH19}. Since the standard notions of fairness such as envy-freeness are not applicable, alternate notions have been defined and studied for this case; see~\cite{LiptonMMS04,Budish11,CaragiannisKMPSW16, PlautR18, ChaudhuryGM20,GhodsiHSSY17,GargKK20} for a subset of notable work and references therein. The Nash welfare continues to serve as a major focal point in this case as well, for which approximation algorithms have been obtained under several classes of utility functions including linear~\cite{ColeG15,ColeDGJMVY17,AnariGSS17,AnariMGV18,BarmanKV18,GargHM18,ChaudhuryCGGHM18,ChaudhuryGMehta20,BarmanBKS20}.

\section{Combinatorial Algorithm}\label{algorithm}
In this section, we give an exact $\tilde{\mathcal{O}}(n^2m^2/\alpha^2)$-time algorithm for instances where every disutility value can be expressed as a power of $(1+\alpha)$, i.e., $d_{ij} = (1 + \alpha)^{k_{ij}}$ for some polynomially bounded integer $k_{ij}$. Then, we show how to modify the same algorithmic framework to get a FPTAS for determining a CEEI when we make no assumption on the disutility values.

Our algorithm is iterative and each iteration comprises of two phases: (1) the price update phase that changes the prices of the chores without changing the allocation and then (2) allocation update phase which changes the allocation without altering the prices. Throughout the algorithm, we maintain that all chores are completely allocated and all agents earn their money from their respective MPB chores. Thus, the main goal of our algorithm is to find identical earnings for the agents (as this ensures that we are at a CEEI) while maintaining the invariants. Thereafter, we show that at the end of each iteration, we can improve the Nash welfare of the allocation by some multiplicative factor. We now elaborate the two phases.

\subsection{Price Update}
In this phase, we identify the \emph{chores in demand}: Let $e_i$ denote the earning of agent $i$, i.e., $e_i = \sum_j x_{ij}p_j$ where $x_{ij}$ is the amount of chore $j$ assigned to $i$ and $p_j$ is the payment for doing unit amount of chore $j$. Let $S$ be the set of agents that have earnings less than one, i.e., $S = \{i \mid e_i < 1 \}$ and let $\Gamma(S)$ be the set of MPB chores for the agents in $S$, i.e., $\Gamma(S) = \{j  \mid d_{ij}/ p_j = \MPB_i \text{ for some } i \in S \}$. We call every chore in $\Gamma(S)$ as a chore in demand. Thereafter, we decrease the prices of all the chores in demand (chores in $\Gamma(S)$) until some agent in $S$ gets interested in a chore in $[m] \setminus \Gamma(S)$, i.e., we have a new MPB edge from an agent in $S$ to a chore outside $\Gamma(S)$. Note that although this update further decreases the earnings of the agents in $S$, it does not decrease Nash welfare as the disutility of each agent is independent of the prices (depends only on the allocation). The price update is summarized in Algorithm~\ref{price-update}.

\begin{algorithm}[H]
    \caption{Price-Update($x,p,S$)}
	\begin{algorithmic}
	  %\State $S \gets \{i \mid e_i < 1 \}$.
      \State $\Gamma(S) \gets \{j \mid \exists i \in S, \text{ s.t. } d_{ij}/p_j = \MPB_i \}$
      \State $\gamma \gets \mathit{min}_{i \in S, j \notin \Gamma(S)} \frac{MPB_i \cdot p_j}{d_{ij}}$.
      \For{$j \in \gamma(S)$}
         \State $p_j \gets \gamma \cdot p_j$.
       \EndFor
       \For{$i \in S$}
         \State $e_i \gets \gamma \cdot e_i$.
       \EndFor  
	\end{algorithmic}	
    \label{price-update}
\end{algorithm}

\subsection{Allocation Update}
This is the second phase of an iteration in our algorithm. In this phase, we aim to determine an allocation that balances the earnings of the agents with the help of the new MPB edges created during the price update. Let $E$ denote the set of new MPB edges that appear after the price update phase and let $J = \{ j \mid (i,j) \in E\}$. Let $\emax = \mathit{max}_{i \in S} e_i$ and $\emin = \mathit{min}_{i \notin S} e_i$. The earning balancing is done depending on the total prices of the chores in $J$. In particular, if $\sum_{j \in J} p_j > (\emin - \emax)/2$, then 
we compute an allocation (or equivalently a money flow) that maximizes the product of earnings (we call this procedure Balance-allocation$(\cdot)$)\footnote{We will show later in this section that we choose $S$ in such a way that $\emin \geq \emax$.}. If $\sum_{j \in J} p_j \leq (\emin - \emax)/2$, then we increase the earning of the agents in $S$ and decrease that of agents in $[n] \setminus S$ by completely allocating the chores in $J$ to agents in $S$ along the MPB edges. The whole procedure is summarized in Algorithm~\ref{allocation-update}.

\begin{algorithm}[H]
	\caption{Allocation-Update($x,p,S$).}
	\begin{algorithmic}
		\State $E \gets $ all new MPB edges from $S$ to $[m] \setminus \Gamma(S)$ after price-update.
		\State $J \gets \{j \mid (i,j) \in E\}$.
		\State $\emax \gets \mathit{max}_{i \in S} e_i$ and $\emin \gets \mathit{min}_{i \in [n] \setminus S} e_i$. 
		\If{$\sum_{j \in J} p_j >  (\emin - \emax)/2$}
		  \State $(x,p,S) \gets $ \textbf{Balance-allocation$(p)$}
	    \Else
	       \For{all $j \in J$}
	          \State Pick an arbitrary $i \in S$ such that $(i,j) \in E$.
	         \label{alg-line} \State for all $i' \in [n] \setminus S$ such that $x_{i'j} > 0$, set $x_{i'j} \gets 0$ and update $e_{i'} = e_{i'} - x_{i'j}p_j$
	         \State Set $x_{ij} \gets 1$ and update $e_i = e_i + p_j$.
	       \EndFor
	     \EndIf	
	\end{algorithmic}	
	\label{allocation-update}
\end{algorithm}

\begin{algorithm}[H]
	\caption{Balance-allocation$(p)$}
     \begin{algorithmic}
     	\State $x \gets \mathit{argmax}_{x \in \mathbb{R}^{nm}_{\geq 0} \text{and } \sum_{i \in [n]}x_{ij} =1 } \prod_{i \in [n]} e_i$, where $e_i = \sum_{j \in [m]} x_{ij}p_j$.
     	\State $S \gets $ the set of agents with lowest surplus.
     \end{algorithmic}
 \end{algorithm}
Overall, in the allocation-update phase, we determine an allocation that improves the Nash welfare of the allocation subject to the agents earning only along their MPB edges.

\subsection{Main Algorithm}
Our overall algorithm alternates between the price-update and allocation-update phases. We start with an initial set of prices (where each price is of the form $(1+\alpha)^k$ for some integer $k$\footnote{Note that $k$ may be negative.}) and an allocation of the chores to the agents along MPB edges. We set $S$ to be the set of agents with lowest earning. Thereafter, we run the price-update and allocation-update phases alternatively until all agents have equal budgets (see Algorithm~\ref{main-algorithm}). The convergence of our algorithm follows from the fact that after polynomially many iterations, we can observe an improvement of $1+ \Omega(\alpha^2)$ in the Nash welfare. Since the Nash welfare is upper bounded by $(n \cdot D_{max})^n$ where $D_{max}=\max_{(i,j):d_{ij} < \infty} d_{ij}$, our algorithm will converge after $\textup{poly}(n,m, \log (D_{max})/ \log(1+\alpha^2)) = \textup{poly}(n,m,\log(D_{max}),1/\alpha)$ iterations.

\begin{algorithm}[H]
	\caption{Full Algorithm}\label{alg:full}
	\begin{algorithmic}
	\For{ all chores $j \in [m]$}
	 \State $p_j \gets  \mathit{min}_{i \in [n]} d_{ij}$.
	\EndFor
	 \State $(x,p,S) \gets $\textbf{Balance-allocation}$(p)$.   
%	\While{$e_i \neq e_j$ for all $i,j \in [n]$} 
	\While{$\exists i,j\in [n] \text{ such that } e_i \neq e_j$} 
	  \State \textbf{Price-update}$(x,p,S)$.
	  \State \textbf{Allocation-update}$(x,p,S)$
	\EndWhile 
	\State \textbf{Return} $x,p$.
	\end{algorithmic}
    \label{main-algorithm}
\end{algorithm}

%In this phase, we determine an allocation that maximizes the product of budgets subject to earning along minimum pain per buck edges for each agent. Since this phase occurs after the price update phase, we show that such an allocation improves the Nash welfare with the help of the new MPB edges (as each such edge connects an agent with low budget to an agent with high budget, thereby allowing balancing). Therefore, in an iteration of our algorithm, the Nash welfare remains unaltered during the price update and it improves strictly during the allocation update phase.     

We first argue about the correctness of the algorithm. The next series of observations will show that all chores are allocated and are allocated only along MPB edges. Therefore, if we have  $e_i = e_j$ for all $i,j \in [n]$, then we are at a CEEI. 

\begin{observation}
	\label{powersofalpha}
	Throughout the algorithm, the prices of the chores are powers of $(1+\alpha)$. Furthermore, any price update phase decreases the prices of the chores in $\Gamma(S)$ and the earnings of the agents in $S$ by at least $1 + \alpha$.  
\end{observation}

\begin{proof}
	We prove our claim by induction on the iterations. The claim holds by definition during the first iteration as we initialize each price $p_j$ as $\mathit{max}_{i \in [n]} d_{ij}$ which is a power of $1+\alpha$. Now consider any arbitrary iteration of the algorithm. By induction hypothesis, we have the prices to be powers of $1+\alpha$. Since the disutilities are also powers of $1+\alpha$, we have that $\MPB_i$ is also a power of $1+\alpha$ for all $i \in [n]$.  Let $i^* \in S$ and $j^* \notin \Gamma(S)$ be such that $\gamma = \tfrac{p_{j^*} \cdot \MPB_{i^*}}{d_{i^*j^*} }$. Note that $\gamma$ is also a power of $1+\alpha$ and since the prices are scaled by $\gamma$, the new set of prices will also remain powers of $1+\alpha$.  We are left to show that $\gamma \leq (1 + \alpha)^{-1}$. Since $\gamma $ is a power of $(1+\alpha)$, it suffices to show that $\gamma < 1$ or equivalently $\MPB_{i^*} < d_{i^*j^*}/ p_{j^*}$, which is true as otherwise $j^* \in \Gamma(S)$.   
\end{proof}

\begin{observation}
	\label{Sproperty}
	At the beginning of every iteration, for all agents $i \in [n] \setminus S$ and chores $j \in \Gamma(S)$, we have $x_{ij} = 0$.
\end{observation}
\begin{proof}
	We prove our claim by induction on the iterations of the algorithm. We first show why this holds for the first iteration: we initialize $S$ as the set of agents with lowest earning after computing an allocation that maximizes the product of budgets. Assume that there is an agent $i' \in [n] \setminus S$ and a chore $j \in \Gamma(S)$ such that $x_{i'j} > 0$. Let $i$ be an agent in $S$ such that $d_{ij}/p_j = \MPB_i$ (note that such an $i$ exists by the definition of $S$ and $\Gamma(S)$). Also note that $e_i < e_{i'}$. We can change the allocation by increasing $x_{ij}$ by $\varepsilon < x_{i'j}$ and decreasing $x_{i'j}$ by $\varepsilon$ such that $e_i$ increases and $e_{i'}$ decreases, but the sum of the earnings remain the same (in particular $e_i + e_{i'}$ remains same). Thus, the product of budgets improves, contradicting the fact that we started with an allocation with maximum product.
	
	Now, consider any iteration of the algorithm. Let $S'$ and $\Gamma(S')$ denote the set $S$ and $\Gamma(S)$ at the beginning of the previous iteration.  If the set $S$ was updated from $S'$ through a call to balance-allocation$()$ in the previous iteration, then we can make a same argument as the induction base case. Otherwise, the set $S = S'$, however $\Gamma(S')$ contains more chores than $\Gamma(S)$ (as a result of new MPB edges appearing after the price-update phase).  Notice that for all chores $j$ in $\Gamma(S) \cap \Gamma(S')$ (which is $\Gamma(S')$), we have $x_{ij} =0$ for all $j \in [n] \setminus S$ by induction hypothesis. For all chores $j \in \Gamma(S) \setminus \Gamma(S')$, we have set $x_{ij} = 0$ for all $i \in [n] \setminus S$ and $j \in \Gamma(S)$ by allocating all the chores in $\Gamma(S) \setminus \Gamma(S')$ completely to agents in $S$ in Algorithm~\ref{allocation-update}.    
\end{proof}

Observation~\ref{Sproperty} helps us prove that throughout the algorithm, all agents earn their money from their MPB chores and all chores remain fully allocated.

\begin{lemma}
	\label{market-clearing} 
	 Throughout the algorithm,
	 \begin{itemize}
	 	\item each agent earns her money only from her MPB chores, and 
	 	\item for all $j\in [m]$ we have $\sum_{i \in [n]} x_{ij} = 1$.
	 \end{itemize}
\end{lemma} 

\begin{proof}	
	In the very first iteration of the algorithm, we initialize the prices such that every chore has at least one MPB edge incident to it: Note that for all $i \in [n]$, we have $\MPB_i \geq 1$ as $d_{ij} /p_j \geq 1$ for all $i \in [n]$ and $j \in [m]$ as $p_j$ is set to $\mathit{min}_{i \in [n]} d_{ij}$. Thus, for each chore $j$ there exists an agent $i \in argmin_{i \in [n]} d_{ij}$ such that $(i,j)$ is an MPB edge. Then, a call to Balance-allocation$()$ ensures that all chores are completely allocated along MPB edges to the agents as otherwise we can increase the earning of some agent by allocating an underallocated chore along an MPB edge and thereby increase the product of earnings.  We now show that the invariants are maintained during the price update and allocation update phases.
	
	\noindent
	\emph{Price Update}: During this phase, the prices of the chores in $[m] \setminus \Gamma(S)$ and the earnings of the agents in $S$ are scaled by $\gamma < 1$ until a new MPB edge appears from some agent(s) in $S$ to some chore in $\Gamma(S)$. Let $p'$ and $e'$ denote the new price and earning vector after price update. First note that there are no MPB edges from agents in $S$ to chores in $[m] \setminus \Gamma(S)$ (by definition of $\Gamma(S)$). Now consider an agent $i \in S$. When we decrease the prices of chores in $\Gamma(S)$, all MPB chores of $i$ will remain MPB chores as they are in $\Gamma(S)$ and the prices of all chores in $\Gamma(S)$ decreases by the same factor $\gamma$. So agent $i$ still earns from her MPB chores. Also, $e'_i = \gamma \cdot e_i =  \gamma \cdot \sum_{j \in [m]} x_{ij} \cdot p_j =  \sum_{j \in \Gamma(S)} x_{ij} \cdot (\gamma \cdot p_j) = \sum_{j \in \Gamma(S)} x_{ij} \cdot p'_j$. Thus, all agents in $S$ earn their entire earning from MPB chores after price update. Now, consider an agent $i' \notin S$. By Observation~\ref{Sproperty}, $x_{i'j} = 0$ for all $j \in \Gamma(S)$, i.e., agent $i'$ earns all of $e_{i'}$ from her MPB chores in $[m] \setminus \Gamma(S)$. Since the prices of the chores in $[m] \setminus S$ and the earnings of the agents in $S$ remains unchanged ($e'_{i'} = e_{i'}$), agent $i'$ earns all her money of $e'_{i'}$ from her MPB chores in $\Gamma(S)$. Thus, all agents in $[n]\setminus S$ also earn all their money from MPB chores after price update. Lastly, since the allocation remains unchanged during price update, all the chores remain completely allocated at the end of price update.
	
	\noindent 
	\emph{Allocation update}: In this phase, the prices remain unaltered, implying that the MPB configuration does not change. If the allocation is updated through a call to Balance-allocation$()$, then each agent will earn their money from MPB chores and all chores will be completely allocated as otherwise we can increase the earning of some agent by allocating an underallocated chore along an MPB edge and thereby increase the product of earnings. If the allocation update does not involve a call to Balance-allocation$()$, then all the chores in $J$\footnote{Recall from Algorithm~\ref{allocation-update}, that $E$ is the set of new MPB edges that appear from agents in $S$ to chores in $[m] \setminus \Gamma(S)$ and $J = \{ j \mid (i,j) \in E \}$} are allocated to agents in $S$ along MPB edges. The earnings of the agents are adjusted accordingly. Therefore, all agents earn their money from MPB chores and all chores are completely allocated. 
\end{proof}

We now prove a crucial invariant about the earnings of the agents in the $S$ that our algorithm maintains throughout. 
\begin{observation}
	\label{maxmin}
 At the beginning of every iteration of our algorithm, we have $\mathit{max}_{i \in S} e_i < \mathit{min}_{i \notin S} e_i$ where $e$ is the earning vector.  
\end{observation}	
\begin{proof}
	 We prove our claim by induction on the iterations. The base case is trivially true as in the first iteration,  $S$ is set by Balance-allocation$()$ to be the set of agents with lowest earning. Therefore we have $\emax < \emin$. Now consider any arbitrary iteration of the algorithm. If $S$ was set by a call to Balance-allocation$()$ in the previous iteration of the algorithm, then we trivially have $\emax < \emin$. So let us assume otherwise, i.e., Balance-allocation$()$ was not invoked in the previous iteration. Let $S'$ be the set $S$ and $e'$ be the earning vector at the beginning of the previous iteration. By induction hypothesis, we have $\mathit{max}_{i \in S'} e'_i < \mathit{min}_{ i \notin S'} e'_i$. Let $\tilde{e}$ be the earning vector after the price-update phase (and before the allocation update phase). In the price update phase, we decrease the prices of the chores in $\Gamma(S)$ and the earnings of the agents in $S$ and thus we have $\mathit{max}_{i \in S} \tilde{e}_i \leq \mathit{max}_{i \in S} e'_i \leq \mathit{min}_{i \notin S} e'_i  = \mathit{min}_{i \notin S} \tilde{e}_i$.    
	 
	 Let $E$ denote the set of new MPB edges that appear after price update and let $J = \{ j \mid (i,j) \in E \}$. Let ${\temax} = \mathit{max}_{i \in S} \tilde{e}_i$ and ${\temin} = \mathit{min}_{i \notin S} \tilde{e}_i$. Since we are in the case where Balance-allocation$()$ is not invoked, we have $(\sum_{j \in J} p_j < (\temin -\temax)/2)$. Observe that Algorithm~\ref{allocation-update} allocates the chores in $J$ to the agents in $S$ along the MPB edges. As a result, the total  earnings of the agents in $S$ increases and that of the agents in $[n] \setminus S$ decreases (by the same amount).  Let $\delta_i = e_i - \tilde{e}_i$ for all $i \in S$ and $\delta_i = \tilde{e}_i - e_i$ for all $i \notin S$. Note that $\sum_{i \in S} \delta_i = \sum_{i \in [n] \setminus S} \delta_i = \sum_{j \in J} p_j$. Therefore we have,
	 \begin{align*}
	 \mathit{max}_{i \in S}e_i &\leq  \temax + \sum_{i \in S}\delta_i\\
	 &= \temax + \sum_{j \in J}p_j\\
	 &< \temax + (\temin-\temax)/2\\
	 &= \temin -  (\temin-\temax)/2\\
	 &< \temin - \sum_{j \in J}p_j\\
	 &= \temin - \sum_{i \in S}\delta_i\\
	 &\leq \mathit{min}_{i \in [n] \setminus S} e_i\enspace . \qedhere
	 \end{align*}  
\end{proof}

We now show that the allocation-update can only increase the product of earnings of the agents.

\begin{lemma}
	\label{balanced_flow}
	Let $e = \langle e_1, e_2, \dots, e_n \rangle$ and $e' = \langle e'_1, e'_2, \dots , e'_n \rangle$ be the earning vector before and after the call to Balance-allocation($p$) in an allocation-update phase of the algorithm. We have $\prod_{i \in [n]} e'_i \geq (1+\alpha^2/16) \prod_{i \in [n]} e_i$. 
\end{lemma}

\begin{proof}
	  Let $\tilde{e} = \langle \tilde{e}_1, \tilde{e}_2, \dots , \tilde{e}_n \rangle$ be the earning vector before the price-update phase in the current iteration of the algorithm. By Observation~\ref{Sproperty}, we have $\mathit{max}_{i \in S} \tilde{e}_i \leq \mathit{min}_{i \notin S} \tilde{e}_i$. By Observation~\ref{powersofalpha}, we decrease the prices of the chores in $\Gamma(S)$ and the earnings of the agents in $S$ by a factor of at least $1+\alpha$. Let $e$ be the new earning vector before the allocation update phase. We have $(1+\alpha) \cdot e_i \leq \tilde{e}_i$ for all $i \in S$, and $e_i = \tilde{e}_i$ for all $i \notin S$. Therefore, we have $(1+ \alpha) \cdot \emax  <  \emin$ \footnote{Recall from Algorithm~\ref{allocation-update}, $\emin = \mathit{min}_{i \in [n] \setminus S} e_i$ and $\emax = \mathit{max}_{i \in S} e_i$.}.   
	
	 Since the balanced allocation outputs an allocation that maximizes the product of the earnings, it suffices to show the existence of an allocation $x'$  such that $\prod_{i \in [n]} e'_i \geq (1+\alpha^2/16) \prod_{i \in [n]} e_i$ where $e'_i = \sum_{j \in [m]} x'_{ij}p_j$. Let $x$ be the allocation before the allocation-update phase of the algorithm. Let $E$ denote the set of new MPB edges after the price update and let $J = \{ j \mid (i,j) \in E\}$. Note that prior to allocation update, $x_{ij} = 0$ for all $i \in S$ and $j \in J$. We initially set $x' =x$ and then change the allocation of the chores in $J$ as follows: For each new MPB edge $(i,j) \in E$, and for each agent $i' \notin S$ such that $x'_{i'j} > 0$, we increase the consumption of the chore $j$ for agent $i$ and  decrease the consumption of chore $j$ for $i'$ \emph{by the same amount} until
	 \begin{itemize}
	 	\item $i'$ does not consume $j$ anymore, i.e., $x'_{i'j} =0$, or 
	 	\item $i$ consumes $j$ entirely, i.e., $x'_{ij} =1$, or 
	 	\item the total increase in the earnings of the agents in $S$ is $(\emin - \emax) /2$, i.e.,$\sum_{i \in S} e'_i - \sum_{i \in S}e_i = (\emin -\emax)/2$.
	 \end{itemize}
	 
	 First note that $x'$ is a valid allocation as all the chores remain completely allocated in $x'$. Let $\mathbf{\Delta}$ denote the total increase in the earnings of the agents in $S$ when we update the allocation to $x'$ from $x$. Note that $\mathbf{\Delta}$ also equals the total decrease in the earnings of the agents in $[n] \setminus S$ as every time we increase the earning of an agent in $S$ by some $\delta>0$, we also decrease the earning of some agent in $[n] \setminus S$ by $\delta$. It is also clear that $\mathbf{\Delta} = (\emin - \emax) /2$ as otherwise all the goods in $J$ would be completely allocated to agents in $S$ and this would imply an increase in the earning of the agents in $S$ by $\sum_{j \in [m]}p_j$ which is larger than $ \emin - \emax > (\emin - \emax) /2$ (if Balance allocation$()$ is called, $\sum_{j \in J} p_j > \emin - \emax$). We are now ready to show the improvement in the product of earnings.
	 
	 Let $\delta_i = e'_i -e_i$ for all $i \in S$ and $\delta_i = e_i - e'_i$ for all $i \notin S$. Note that $\delta_i \geq 0$ for all $i \in [n]$ and  $\sum_{i \in S} \delta_i = \sum_{i \in [n] \setminus S} \delta_i = \mathbf{\Delta} = (\emin - \emax)/2$. %Let us assume w.l.o.g. that $e_1 \leq e_2 \leq \dots \leq e_n$. 
	 We now make a technical claim which helps us analyze the improvement in the product of earnings.
	 
	 \begin{claim}
	 	\label{technical}
	 	Let $0 \leq y_1 \leq y_2 \leq  \dots \leq y_k$. Let $\delta_i \geq 0$ and $\gamma_i \geq 0$ for all $i \in [k]$. Furthermore, let  $y_1 \geq \sum_{j \in [k]} \gamma_j$. Then, we have 
	 	\begin{enumerate}
	 		\item $\prod_{i \in [k]} (y_i + \delta_i) \geq \prod_{i \in [k-1]} y_i \cdot (y_k + \sum_{j \in [k]} \delta_j )$, and similarly,
	 		\item $\prod_{i \in [k]} (y_i - \gamma_i) \geq (y_1 - \sum_{j \in [k]} \gamma_j) \cdot \prod_{2\leq i \leq k}y_i$.
	 	\end{enumerate}
	 \end{claim} 
     \begin{proof}
	The proof follows immediately by the repeated applications of the following two facts. Given four numbers $x,y, \lambda_1, \lambda_2$ such that $\lambda_1 + \lambda_2 \leq x \leq y$ and 
	\begin{itemize}
		\item $(x+ \lambda_1)(y+\lambda_2) \geq x (y + \lambda_1 + \lambda_2)$: We have,
		\begin{align*}
		(x+ \lambda_1)(y+\lambda_2) &= xy + \lambda_1y  + \lambda_2x + \lambda_1\lambda_2\\
		&\geq xy + \lambda_1y  + \lambda_2x \\
		&\geq xy + \lambda_1x + \lambda_2x &(\text{as $x \leq y$})\\
		&= x (y + \lambda_1 + \lambda_2).
		\end{align*} 
		\item  $(x- \lambda_1)(y-\lambda_2) \geq  (x - \lambda_1 - \lambda_2)y$: We have,
		\begin{align*}
		(x - \lambda_1)(y - \lambda_2) &= xy - \lambda_1y  - \lambda_2x + \lambda_1\lambda_2\\
		&\geq xy - \lambda_1y  - \lambda_2x \\
		&\geq xy - \lambda_1y - \lambda_2y &(\text{as $x \leq y$})\\
		&= (x - \lambda_1 - \lambda_2)y. \qedhere 
		\end{align*}
	\end{itemize}
     \end{proof}
	 Observe that,
	 \begin{align}
	  \prod_{i \in [n]} e'_i &= \prod_{i \in S} (e_i + \delta_i) \cdot \prod_{i \in [n] \setminus S} (e_i - \delta_i) \nonumber \\ 
	                         &\geq \prod_{i \in S} e_i \cdot \frac{\emax + \mathbf{\Delta}}{\emax} \cdot \frac{\emin - \mathbf{\Delta}}{\emin} \cdot \prod_{i \in [n] \setminus S} e_i \nonumber  &\text{(by Claim~\ref{technical})}\\
	                         &= \prod_{i \in [n]} e_i \cdot \big(1 + \frac{\mathbf{\Delta}}{\emax} \big) \cdot \big(1 - \frac{\mathbf{\Delta}}{\emin} \big) \nonumber \\
	                         &=  \prod_{i \in [n]} e_i \cdot \big(1 + \mathbf{\Delta}\frac{\emin -\emax}{\emax \emin} - \frac{\mathbf{\Delta}^2}{\emax \emin} \big) \nonumber\\ 
	                         &=  \prod_{i \in [n]} e_i \cdot \big(1 + \frac{2\mathbf{\Delta}^2}{\emax \emin} - \frac{\mathbf{\Delta}^2}{\emax \emin} \big) \nonumber &\text{\big(Substituting $\mathbf{\Delta} = \frac{\emin -\emax}{2}$ \big)} \\
	                         &=  \prod_{i \in [n]} e_i \cdot \big(1 + \frac{\mathbf{\Delta}^2}{\emax \emin} \big) \label{eq1}
	 \end{align}
     Since $(1+ \alpha)^{-1} \emin \geq \emax$, we have 
     \begin{align*}
      \mathbf{\Delta} &\geq \frac{1 - (1+\alpha)^{-1}}{2} \emin\\
             &= \frac{\alpha}{2(1+\alpha)} \emin\\
             &\geq \frac{\alpha}{4} \emin &\text{(as $\alpha < 1$)}.
     \end{align*}
     Substituting this lower bound on $\mathbf{\Delta}$, and using the fact that $\emin \geq \emax$, we have $\prod_{i \in [n]} e'_i \geq (1 + \alpha^2/16)\prod_{i \in [n]} e_i$. 	 
\end{proof}

\begin{lemma}
	\label{ear_prod_increase}
	Let $e = \langle e_1, e_2, \dots, e_n \rangle$ and $e' = \langle e'_1, e'_2, \dots , e'_n \rangle$ be the earning vector before and after an allocation-update phase of the algorithm. We have,  $\prod_{i \in [n]} e'_i \geq  \prod_{i \in [n]} e_i$.
\end{lemma}

\begin{proof}
	If the allocation update involves a call to Balance-allocation$()$, then the product of earnings improves by Lemma~\ref{balanced_flow}. We now show that the product of the earnings improve. 
	
	Now consider the case where there is no call to Balance-allocation$()$. In this case, we have $(\sum_{j \in J} p_j < (\emin -\emax)/2)$\footnote{Recall that $E$ is the set of new MPB edges that appear from $S$ to $[m] \setminus S$ after the price update phase and $J = \{ j \mid (i,j) \in E \}$.}. Observe that Algorithm~\ref{allocation-update} allocates the chores in $J$ to the agents in $S$ along the MPB edges. As a result, the total  earnings of the agents in $S$ increases and that of the agents in $[n] \setminus S$ decreases (by the same amount).  Let $\delta_i = e'_i - e_i$ for all $i \in S$ and $\delta_i = e_i - e'_i$ for all $i \notin S$. Note that $\sum_{i \in S} \delta_i = \sum_{i \in [n] \setminus S} \delta_i = \sum_{j \in J } p_j$. Therefore we have,
	\begin{align*}
	 \prod_{i \in [n]} e'_i &= \prod_{i \in S} (e_i + \delta_i) \cdot \prod_{i \in [n] \setminus S} (e_i -\delta_i)\\
	                        &\geq \prod_{i \in S} e_i \cdot \frac{\emax + \sum_{i \in S} \delta_i}{\emax} \cdot \frac{\emin - \sum_{i \in [n] \setminus S} \delta_i}{\emin} \cdot \prod_{i \in [n] \setminus S} e_i \nonumber  &\text{(by Claim~\ref{technical})}\\
	                        &= \prod_{i \in [n]} e_i \cdot \frac{\emax + \sum_{j \in J} p_j}{\emax} \cdot \frac{\emin -\sum_{j \in J} p_j}{\emin}.
	\end{align*} 
	It suffices to show that $\frac{\emax + \sum_{j \in J} p_j}{\emax} \cdot \frac{\emin - \sum_{j \in J} p_j}{\emin} > 1$. This is indeed the case as $\emax + \sum_{j \in J} p_j \leq \emin - \sum_{j \in J}p_j$, i.e., the values $\emax$ and $\emin$ move closer to each other while still maintaining the same sum, implying that the product improves.
\end{proof}

We are now ready to show convergence of our algorithm.  We briefly sketch the overall idea: the price-update step does not change our potential as our potential only depends on the allocation which remains unaltered (even though the prices of the chores and the earnings of the agents change). Thereafter, with the help of new MPB edges we are able to improve the product of disutilities. However, note that during an allocation update, we do not change the prices of the chores. Therefore, the disutility of an agent is proportional to her earning as $D_i(x_i) = \MPB_i \cdot e_i$ or equivalently $\NSW(\mathbf{x}) = \prod_{i \in [n]} \MPB_i \cdot \prod_{i \in [n]} e_i$. Since an allocation update increases $\prod_{i \in [n]} e_i$, it also improves the Nash welfare of the allocation. We now present the theorem and its formal proof.

\begin{lemma}
	\label{convergence}
	 Algorithm~\ref{alg:full} computes a CEEI in $\mathcal{O}(nm/ \alpha^2 \cdot \log(nD_{max}))$ iterations which involves at most $\mathcal{O}(n / \alpha^2 \cdot \log(nD_{max}))$ calls to Balance-allocation$()$.  %In particular there are at most $\mathcal{O}$
\end{lemma}

\begin{proof}
	To this end, we first argue that the Nash welfare never decreases throughout the algorithm: Note that $\NSW(x) = \prod_{i \in [n]} D_i(x_i)$ is independent of the prices of the chores and therefore does not change during the price update phase. Prior to the allocation update, we can write each $D_i(x_i)$ as $\MPB_i \cdot e_i$. Therefore, $\NSW(x) = \prod_{i \in [n]} \MPB_i \cdot \prod_{i \in [n]} e_i$. Note that we do not alter the prices of the chores during an allocation update and thus $\prod_{i \in [n]} \MPB_i$ remains the same. By Lemma~\ref{ear_prod_increase}, $\prod_{i \in [n]} e_i$ increases during an allocation update phase and thus the Nash welfare increases.

	We now bound the total number of iterations that invokes a Balance-allocation$()$. By Lemma~\ref{balanced_flow}, $\prod_{i \in [n]} e_i$ increases by $(1 + \alpha^2/16)$ every time Balance-allocation$()$ is invoked. Since $\prod_{i \in [n]} \MPB_i$ remains unaltered during allocation update, and $\NSW(x) = \prod_{i \in [n]} \MPB_i \cdot \prod_{i \in [n]} e_i$, we can conclude that each time Balance-allocation$()$ is called, the Nash welfare increases by a factor of $1( + \alpha^2 / 16)$. Since $\NSW(x) \leq (nD_{max})^n$ and the Nash welfare never decreases throughout the algorithm by Lemma~\ref{ear_prod_increase}, we can have at most $\mathcal{O}(\frac{\log((nD_{max})^n)}{ \log (1 + \alpha^2 / 16)}) \in \mathcal{O} (n / \alpha^2 \log(nD_{max}))$ many calls to Balance-allocation$()$.
	
	We now bound the number of iterations where Balance-allocation$()$ is not invoked. In this case, note that $\lvert \Gamma(S) \rvert $ increases during an allocation update phase as the chores in $J$ get added to $\Gamma(S)$. Therefore, we can have at most $\mathcal{O}(m)$ consecutive iterations that does not invoke Balance-allocation$()$. This implies that the total number of iterations of the algorithm is at most $\mathcal{O}(nm/ \alpha^2 \cdot \log(nD_{max}))$.
\end{proof}

To obtain a quadratic running time for our algorithm, it suffices to show how to implement Balance-allocation$()$ in $\tilde{\mathcal{O}}(nm (n+ m))$ time. We now elaborate this.

\subsection{Implementing Balance-allocation$()$ in $\mathcal{O}(nm(n+m) \log(nm))$ Time}
We show that the problem of determining a Balance-allocation$()$ can be reduced to the problem of finding a lexicographically optimal flow~\cite{fujishige1980lexicographically}, which can be determined in $\mathcal{O}(nm \cdot (n+m) \log(nm))$ time. 

\begin{definition}{~\cite{fujishige1980lexicographically}}
	\label{lexicographically_optimal_flow}
	 We are given a directed graph $G = (V,E)$ with edge capacities. $S^+$ is a set of $k$ source vertices, i.e., vertices in $V$ without an incoming edge and $t$ is a  single sink vertex, i.e., a vertex in $V$ without an outgoing edge. Given any flow $f$ in $G$, let $\delta^+(v)$ denote the total outflow from a vertex $v \in V$. A lexicographically optimal flow is a maximum flow $f$ from $S^+$ to $t$ such that the vector $\langle \delta^+(s_1), \delta^+(s_2), \dots, \delta^+(s_k) \rangle$ is lexicographically maximum subject to $s_{\ell} \in S^+$ for all $\ell \in [k]$ and $\delta^+(s_1) \leq \delta^+(s_2) \leq \dots \leq \delta^+(s_k)$.    
\end{definition} 

We now show how  lexicographically optimum flows can be used to implement the Balance-allocation$()$ subroutine. Given a price vector $p$, we define the market-network $M_p = ([n] \cup [m] \cup \{t\}, E)$ as a digraph that contains the agents $[n]$, chores $[m]$ and a sink vertex $t$. The directed edges from the agents to the chores are MPB edges at the prices $p$, i.e., $(i,j) \in E$ is $d_{ij} / p_j = \MPB_i$. Edges from agents to chores have infinite capacity (meaning one can push any amount of flow through them). There is a directed edge from each chore $j$ to the sink $t$ with a capacity of $p_j$. Note that the set of agents have no incoming edges. The main goal of Balance-allocation$()$ is to find a flow in the MPB graph that maximizes the product of the outflow of the sources. Once we have such a flow $f$, it is easy to find the allocation as $x_{ij} = f_{ij} / p_j$.  We make an observation that any flow that maximizes the product of outflows of the sources is a lexicographically optimum flow.
\begin{lemma}
	\label{flowrelation}
		Given a market-network $G$, a maximum flow $f$  maximizes the product of outflows of the sources, $\prod_{s \in S^+} \delta^+(s)$ if and only if it is a lexicographically optimum flow.
\end{lemma}

\begin{proof}
	We crucially use the following characterization of a lexicographically optimum flow given in ~\cite{fujishige1980lexicographically}
	\begin{claim}(~\cite{fujishige1980lexicographically})
		\label{technicalflowlemma}
		A maximum flow $f$ is a lexicographically optimal flow if and only if it minimizes the $\ell_2$-norm of the outflow vector of the sources, i.e., the flow that minimizes $\sum_{s \in S^+} (\delta^+(s))^2$.
	\end{claim}
	Therefore, it suffices to show that any flow that minimizes the $\ell_2$-norm of the outflows of the sources also maximizes the product of the outflows from the sources. This follows immediately when we write the KKT conditions for both the convex programs. Consider the convex program for minimizing the $\ell_2$-norm of the outflows 
	
	\begin{equation*}
	\begin{array}{ll@{}ll}
	\text{minimize}  & \displaystyle \sum\limits_{i \in [n]}^{ } (\sum_{(i,j) \in E} f_{ij})^2 \\
	\text{subject to}& \displaystyle  \sum_{(i,j) \in E}  f_{ij} = p_j, & &\forall (i,j) \in E\\ 
	&     f_{ij} \geq  0,  & &\forall (i,j) \in E\\
	\end{array}
	\end{equation*}
	The optimum solution will satisfy the KKT conditions. Let $\lambda_j$ be the \emph{free} dual variable corresponding to the constraint $\sum_{(i,j) \in E} f_{ij} = p_j$ and $\mu_{ij} \geq 0$ be the dual variable corresponding to $f_{ij} \geq 0$. The KKT conditions require, in addition to primal and dual feasibility, 
	\begin{itemize}
		\item \textbf{Stationarity:} For all $(i,j) \in E$,   $2(\sum_{(i,j) \in E} f_{ij})  + \lambda_j - \mu_{ij} = 0 \implies \sum_{(i,j) \in E} f_{ij} = (\mu_{ij} -\lambda_j)/2$
		\item \textbf{Complementary Slackness:} For all $(i,j) \in E$, $\mu_{ij} \cdot f_{ij} = 0$.
	\end{itemize}	
	Consider any $i$ and $j$ such that $f_{ij} > 0$. Then $\mu_{ij} = 0$. Then, we have $\sum_{(i,j) \in E} f_{ij} = - \lambda_j/2$. This implies that all nodes $i$ that have positive surplus towards the node $j$ have the same outflow which is equal to $- \lambda_j/2$. Since $\mu_{ij} \geq 0$, this also implies that $- \lambda_j/2 \leq -(\lambda_j + \mu_{ij})/2$. This gives us the following characterization of the optimum flow.
	
	\begin{observation}
		\label{propertyflow}
		Let $f$ be any flow in the market-network that minimizes the $\ell_2$-norm of the outflows. If there are two source nodes $i$ and $i'$ that are adjacent to node $j$, and $f_{ij} > 0$, then (i) $\sum_{(i,k) \in E} f_{ik} = \sum_{(i',k) \in E} f_{i'k} $ or (ii) $\sum_{(i,k) \in E} f_{ik} < \sum_{(i',k) \in E} f_{i'k} $ and $f_{i'j} = 0$.
	\end{observation}

	We will now show that any flow that satisfies the property in Observation~\ref{propertyflow}, also maximizes the product of outflows, implying that any flow that minimizes the $\ell_2$-norm of the surpluses also maximizes the product of outflows. To this end, consider the program of maximizing the product of outflows or equivalently maximizing the sum of the logarithms of the outflows.
	
	\begin{equation}
	\label{eqprod}
	\begin{array}{ll@{}ll}
	\text{maximize}  & \displaystyle \sum\limits_{i \in [n]}^{ } (\log (\sum_{(i,j) \in E} f_{ij})) \\
	\text{subject to}& \displaystyle  \sum_{(i,j) \in E}  f_{ij} = p_j, & &\forall (i,j) \in E\\ 
	&     f_{ij} \geq  0,  & &\forall (i,j) \in E\\
	\end{array}
	\end{equation}
	Let $\lambda_j$ be the \emph{free} dual variable corresponding to the constraint $\sum_{(i,j) \in E} f_{ij} = p_j$ and $\mu_{ij} \geq 0$ be the dual variable corresponding to $f_{ij} \geq 0$. The KKT conditions require, in addition to primal and dual feasibility, 
	\begin{itemize}
		\item \textbf{Stationarity:} For all $(i,j) \in E$,  $ - 1 / (\sum_{(i,j) \in E} f_{ij})  + \lambda_j - \mu_{ij} = 0$  %\sum_{(i,j) \in E} f_{ij} = 1/ (\lambda_j - \mu_{ij})$
		\item \textbf{Complementary Slackness:} For all $(i,j) \in E$, $\mu_{ij} \cdot f_{ij} = 0$.
	\end{itemize}	
	Consider any flow $f$ that satisfies the property in Observation~\ref{propertyflow}. We will find the dual variables $\lambda_j$s and $\mu_{ij}$s that satisfy the KKT conditions of the convex program~\ref{eqprod}.  We set $\lambda_j = 1/ (\sum_{(i,j) \in E} f_{ij})$ where $i$ is a source node such that $f_{ij} > 0$. We set $\mu_{ij} = \lambda_j - 1/(\sum_{(i,j) \in E} f_{ij})$. We prove that $f_{ij}$'s, $\mu_{ij}$'s and $\lambda_j$'s satisfy the KKT conditions for the convex program~\eqref{eqprod}.
	
	\begin{itemize}
		\item \textbf{Stationarity:} We set $\mu_{ij} = \lambda_j - 1/(\sum_{(i,j) \in E} f_{ij})$ for all $(i,j) \in E$. Therefore, for all $i,j$, we have $- 1/(\sum_{(i,j) \in E} f_{ij}) + \lambda_j - \mu_{ij} = 0$.  
		
		\item \textbf{Complementary slackness:} For all $(i,j) \in E$ where $f_{ij} = 0$, we have $f_{ij} \cdot \mu_{ij} = 0$. For all $i,j$ such that $f_{ij} > 0$, we have $\lambda_j = 1/(\sum_{(i,j) \in E} f_{ij})$, implying that $\mu_{ij} =  \lambda_j - 1/(\sum_{(i,j) \in E} f_{ij}) = 0$.
		
		\item  \textbf{Primal feasibility:} This is satisfied by the definition of a flow in a market-network.
		
		\item \textbf{Dual feasibility:} We want to show that $\mu_{ij} \geq 0$ for all $i,j$. For all $i,j$ such that $f_{ij} > 0$, we have $\lambda_j = 1/(\sum_{(i,j) \in E} f_{ij})$, implying that $\mu_{ij} =  \lambda_j - 1/(\sum_{(i,j) \in E} f_{ij}) = 0$. For all $i,j$ such that $f_{ij} = 0$, let $i'$ be a source node such that $f_{i'j} > 0$. Note that by Observation~\ref{eqprod}, we have $\sum_{(i,j) \in E} f_{ij} \geq \sum_{(i,j) \in E}f_{i'j}$. Therefore, we have, 
		\begin{align*}
		\mu_{ij} &= \lambda_j - 1/(\sum_{(i,j) \in E} f_{ij})\\
		&\geq \lambda_j - 1/(\sum_{(i,j) \in E} f_{i'j})\\
		&= 0 &\text{as $\lambda_j = 1/(\sum_{(i,j) \in E} f_{i'j})$}.%\hspace{4.5cm}\qedhere
		\end{align*}	
	\end{itemize}
\end{proof}

 Thus, Balance-allocation$()$ can be implemented with one call to finding a lexicographically optimum flow in the MPB graph. ~\cite{fujishige1980lexicographically} give a fast algorithm to determine a lexicographically optimum flow.
 
 \begin{lemma}{~\cite{fujishige1980lexicographically}}
 	\label{lexflowcomp}
 	Given a digraph $G =(V,E)$ with edge capacities, source set $S^+$ and a sink $t$, one can find a lexicographically optimum flow in $\mathcal{O}(\lvert V \rvert \cdot \lvert E \rvert \log(\lvert V \rvert \cdot \lvert E \rvert))$.
 \end{lemma}

Lemma~\ref{lexflowcomp} immediately gives us an efficient algorithm to implement Balance-allocation$()$.
\begin{corollary}
	\label{balance-allocation}
	We can implement the subroutine Balance-allocation$()$ in $\mathcal{O}(nm^2 \cdot \log(nm))$ time
\end{corollary}

\begin{proof}
	Follows immediately from Lemma~\ref{lexflowcomp} after substituting $\lvert V \rvert = n+m $ and $\lvert E \rvert = nm$ and from the fact that $m > n$.
\end{proof}

We are ready to bound the running time of our algorithm.

\begin{theorem}
	\label{mainthmroundeddisutilities}
	 Given a chore division instance with rounded disutilities, we can determine a CEEI in $\mathcal{O}(n^2m^2/ \alpha^2 \cdot \log(nD_{max}) \log(nm)) \in \tilde{\mathcal{O}} (n^2m^2 / \alpha^2)$ time.
\end{theorem}

\begin{proof}
	By lemma~\ref{convergence}, Algorithm~\ref{main-algorithm} has a total of $\mathcal{O}(nm/ \alpha^2 \log(nD_{max}))$ iterations, out of which at most $\mathcal{O}(n / \alpha^2 \log(nD_{max}))$ involves invocation to Balance-allocation$()$. We now bound the total time taken by our algorithm for all price updates and all allocation updates. Note that each price update can be implemented in $\mathcal{O}(nm)$ as it mainly involves finding the agent $i \in S$ and chore $j \notin \Gamma(S)$ such that $d_{ij} / (\MPB_i \cdot p_j) $ is minimum. Since there are at most $\mathcal{O}(nm/ \alpha^2 \log(nD_{max}))$ iterations, the total time taken on all price update phases is $\mathcal{O}(n^2m^2/ \alpha^2 \log(nD_{max}))$.
	
	Now, we bound the total time taken on all allocation updates. First, look into all the allocation update calls that involves a call to Balance-allocation$()$. Each such iteration can be implemented in $\mathcal{O}(nm^2 \log(nm))$ time by Corollary~\ref{balance-allocation} and there are at most $\mathcal{O}(n/ \alpha^2 \log(nD_{max}))$ such iterations. Thus, the total time taken on all allocations that involve call to Balance-allocation$()$ is $\mathcal{O}(n^2m^2/ \alpha^2 \log(nD_{max}) \log(nm))$. Lastly, look at all allocation update phases that does not involve call to Balance-allocation$()$. Each such phase can be implemented in $\mathcal{O}(nm)$ time as it mainly involves allocating the chores in $J$ to agents in $S$ along MPB edges and there are at most $\mathcal{O}(nm)$ MPB edges. Since there are at most $\mathcal{O}(nm/ \alpha^2 \log(nD_{max}))$ many such iterations, the total time spend on iterations that do not involve call to Balance-allocation$()$ is $\mathcal{O}(n^2m^2/ \alpha^2 \log(nD_{max}))$. Overall, the total running time is $\mathcal{O}(n^2m^2/ \alpha^2 \cdot \log(nD_{max}) \log(nm)) \in \tilde{\mathcal{O}} (n^2m^2 / \alpha^2)$.
\end{proof}

%\item $\mathit{max}_{i \in S} e_i < \mathit{max}_{i' \in [n] \setminus S} e_{i'}$.

\subsection{Modification of Algorithm~\ref{main-algorithm} to get a FPTAS}
We show how to extend the combinatorial algorithmic framework to get a $(1-\varepsilon)$-CEEI in $\mathcal{O}(n^4m^2\cdot \log(nD_{max}))$ time. With  subtle changes, Algorithm~\ref{main-algorithm} can be adapted to give a $(1-\varepsilon)$-CEEI, when the disutility values are arbitrary. We make no changes to the price update phase. In the allocation update phase, the only change we make is in initializing and updating the set $S$. In particular, after each call to balance allocation,  the set $S$ is determined by the following procedure: Renumber the agents according to increasing order of their earnings. Let $i$ be such that $e_i / e_{i+1}$ is maximum. Then $S \gets [i]$. We run the algorithm as long as $\mathit{max}_{i \in [n]} e_i/ \mathit{min}_{i \in [n]} e_i \geq 1 + \varepsilon$. The analysis of convergence is a little more involved.

Note that Observation~\ref{Sproperty}, Lemma~\ref{market-clearing} and Observation~\ref{maxmin} still hold as they only rely on the fact that every time $S$ is updated, we have $\mathit{max}_{i \in S} e_i < \mathit{min} _{i \notin S} e_i$, which is always true in the new algorithm. We only need to argue about Lemma~\ref{balanced_flow}. This would require more work as we may not always have the multiplicative gap of $1+\alpha$ between $\mathit{max}_{i \in S} e_i$ and $\mathit{min}_{ i \notin S} e_i$. However, we circumvent this problem by showing that between every consecutive calls to balance the product of earnings increases by a multiplicative factor of $1 + \varepsilon^2/256n^2$.

Consider the set of iterations (say $\iter_1$ ro $\iter_r$ ) following the iteration involving a call to Balance-allocation$()$ until the next call to Balance-allocation$()$. Let $e$ be the earning vector, $p$ be the price vector, and $S = [i]$ (after renumbering the agents in increasing order of earnings) at the beginning of $\iter_1$. Since we have $e_1/e_n \geq 1 + \varepsilon$, we have $e_i/e_{i+1} \geq (1 + \varepsilon)^{1/n} \geq 1 + \varepsilon/2n$. Also note that $\emin/ \emax = \mathit{min}_{i \in S} e_i/ \mathit{max}_{i \in S} e_i = e_i/ e_{i+1} \geq 1 + \varepsilon/2n$. 

We now look into $\iter_{r}$. Let $e'$ be the earning vector, $p'$ be the price vector, and $S'$ be the set $S$ at the beginning of $\iter_{r}$. Note that $S' =S$, as the set $S$ has not changed (it only changes through an invocation to Balance-allocation$()$, which has not happened from $\iter_1$ to $\iter_{r-1}$), while some of the chores from the set $[m] \setminus \Gamma(S)$ in the earlier iterations got allocated along MPB edges to agents in $S$. If $\emin' / \emax' = \mathit{min}_{i \in S} e'_i/ \mathit{max}_{i \in S} e'_i \geq 1 + \varepsilon/4n$, then by the same analysis in Lemma~\ref{balanced_flow}, we can argue that the product of disutilities improve by a factor of $1 + \varepsilon^2/256n^2$.

Now, we consider the case that $\emin' / \emax'  \leq 1 + \varepsilon/4n$. In this case, we show that the ratio of the product of disutilities at the beginning of $\iter_{r}$ to that at the beginning of $\iter_1$ is  $1 + \varepsilon^2/16n^2$. Let $J$ denote the set of chores that were added to the set $\Gamma(S)$ from $\iter_1$ to $\iter_{r-1}$. Let $p_j$ denote the price of a chore $j \in J$ in the iteration in which it was added to $\Gamma(S)$ via MPB edges. Note that if a chore $j$ is added to $\Gamma(S)$ via new MPB edges in $\iter_{\ell}$, the price of chore $j$ remained unaltered from $\iter_1$ to $\iter_{\ell}$, i.e., $p_j$ is the price of the chore in $\iter_1$ also. Let $\beta$ denote the total multiplicative decrease across all price-update phases from $\iter_1$ to $\iter_{r-1}$, i.e., $\beta = \prod_{\ell \in [r-1]} \gamma_{\ell}$, where $\gamma_{\ell}$ is the amount by which the prices of the chores in $\Gamma(S)$ has been reduced in the price update phase of $\iter_{\ell}$. Note for each $i \in S$, we have $e'_i \geq  \beta \cdot (e_i + \delta_i)$, where $\delta_i = \sum_{j \in J} x'_{ij}p_j$ where $x'$ is the allocation at the end of $\iter_{r-1}$ (or equivalently, at the beginning of $\iter_1$): This is due to the fact that if $z$ units of chore $j \in J$ is allocated to agent $i$ in $\iter_{\ell}$, then the price of this fractional amount of chore $j$ gets scaled down by $\prod_{ \ell+1 \leq q \leq r-1} \gamma_q \geq \beta$ until $\iter_{r-1}$ in the earning of agent $i$. Similarly, for each $i \notin S$, we have $e'_i = e_i - \delta_i$ where $\delta_i = \sum_{j \in J}x_{ij}p_j$ where $x$ is the allocation at $\iter_1$: This is due to the fact that the only decrease in the earnings of the agents outside $S$ comes from loosing the consumption of the chores in $J$. Note that $\sum_{i \in S} \delta_i = \sum_{i \in S}\sum_{j \in J}x'_{ij}p_j = \sum_{j \in J}p_j$ as by the end of $\iter_{r-1}$ all chores in $J$ are allocated to agents in $S$ and similarly $\sum_{i \notin S} \sum_{j \in J} x_{ij} p_j = \sum_{j \in J} p_j$ as initially (at the beginning of $\iter_1$), all chores in $J$ were allocated to agents in $[n] \setminus S$. Therefore, we have $\sum_{i \in S} \delta_i = \sum_{i \notin S} \delta_i$. 

Let $i' \in S$ be such that $\emax' = e'_{i'}$ and $\tilde{i} \in [n] \setminus S$ be such that $\emin' = e'_{\tilde{i}}$. Similar to the proof of Lemma~\ref{balanced_flow}, we give a lower-bound on $\Delta = \sum_{i \in [n]} \delta_i$ by giving a lower-bound on $ \delta_{i'} + \delta_{\tilde{i}}$. Since we have $\emin'/ \emax' \leq 1 + \varepsilon/4n$, we can conclude that 
\begin{align*}
&\frac{\emin - \delta_{\tilde{i}}}{\beta(\emax + \delta_{i'})} \leq 1 + \frac{\varepsilon}{4n}\\
&\implies \emin - \delta_{\tilde{i}}  \leq \beta (\emax + \delta_{i'}) + \frac{\varepsilon}{4n} (\beta) (\emax + \delta_{i'})\\
&\implies \beta \delta_{i'}  + \delta_{\tilde{i}} \geq \emin - \beta \emax - \frac{\varepsilon}{4n} (\beta (\emax + \delta_{i'}))\\
\end{align*}
Since $\emin \geq \emin' \geq \emax' \geq \beta (\emax + \delta_{i'})$, we have ,
\begin{align*}
&\beta \delta_{i'}  + \delta_{\tilde{i}} \geq \emin - \beta \emax - \frac{\varepsilon}{4n} (\emin)\\
&\implies \beta \delta_{i'} + \delta_{\tilde{i}} \geq (1- \frac{\varepsilon}{4n}) \emin - \beta \emax\\
&\implies \delta_{i'} + \delta_{\tilde{i}} \geq (1- \frac{\varepsilon}{4n}) \emin -  \emax &(\text{as $\beta < 1$})\\
&\implies \delta_{i'} + \delta_{\tilde{i}} \geq (1- \frac{\varepsilon}{4n}) \emin - (1+ \frac{\varepsilon}{2n})^{-1} \emin &\text{(as $\emin \geq (1 + \frac{\varepsilon}{2n}) \emax$)}\\
&\implies \delta_{i'} + \delta_{\tilde{i}} \geq \big ((1- \frac{\varepsilon}{4n}) -  (1+ \frac{\varepsilon}{2n})^{-1} \big ) \emin \\
&\implies \delta_{i'} + \delta_{\tilde{i}} \geq \bigg (\frac{(1- \frac{\varepsilon}{4n})(1+ \frac{\varepsilon}{2n})-1}{1+ \frac{\varepsilon}{2n}} \bigg) \emin \\
&\implies \delta_{i'} + \delta_{\tilde{i}} \geq \frac{\varepsilon}{16n} \cdot \emin.
\end{align*}
Therefore, $\Delta \geq  (\varepsilon/16n) \cdot \emin$. We are now ready to show that ratio of the product of disutilities at the beginning of $\iter_{r}$ to that at the beginning of $\iter_1$ is  $1 + \varepsilon^2/256n^2$. Let $\MPB_i$ denote the MPB value of agent $i$ in $\iter_1$ and $\MPB'_i$ denote the MPB value of agent $i$ in $\iter_{r-1}$. Note that $\MPB'_i = (1/\beta) \cdot \MPB_i$ for all $i \in S$ and $\MPB'_i = \MPB_i$ for all $i \notin S$. Therefore, we have the product of disutilities in $\iter_{r-1}$ to that in $\iter_1$ as 
\begin{align*}
&\frac{\prod_{i \in S} \MPB'_i \cdot \prod_{i \notin S} \MPB'_i \cdot \prod_{i \in [n]} e'_i}{\prod_{i \in S} \MPB_i \cdot \prod_{i \notin S} \MPB_i \cdot \prod_{i \in [n]} e_i}\\
&= \frac{(1/\beta)^{\lvert S \rvert} \cdot \prod_{i \in [n]} e'_i}{\prod_{i \in [n]} e_i}\\
&\geq \frac{(1/\beta)^{\lvert S \rvert} \cdot \prod_{i \in S} \beta (e_i + \delta_i) \cdot \prod_{i \notin S} (e_i -\delta_i)}{\prod_{i \in [n]} e_i}\\
&= \frac{ \prod_{i \in S} (e_i + \delta_i) \cdot \prod_{i \notin S} (e_i -\delta_i)} { \prod_{i \in [n]} e_i}\\
&\geq \frac{(\emax + \Delta) (\emin-\Delta)}{\emax \emin} &\text{(by Claim~\ref{technical})} .  
\end{align*}
From here on, we can continue the same analysis as in the proof of Lemma~\ref{balanced_flow} and get the desired bound.

Therefore, there can be at most $\mathcal{O} ( \log_{1 + \varepsilon^2/256n^2} ( (nD_{max})^n)) \in \mathcal{O}(n^3 \log(nD_{max})/ \varepsilon^2 ) $ calls to Balance-allocation$()$. Between any two consecutive calls to Balance-allocation$()$, we can have at most $\mathcal{O}(m)$ iterations (as the size of $\Gamma(S)$ strictly increases in each of these iterations). Therefore, there are at most $\mathcal{O}(n^3m \log(nD_{max})/ \varepsilon^2)$ many iterations that do not involve a call to Balance-allocation$()$. Each call to Balance-allocation$()$ can be implemented in $\mathcal{O}(nm(n+m))$ time and each iteration that does not involve a call to Balance-allocation$()$ can be implemented in $\mathcal{O}(nm)$ time. Therefore, the desired running time is $\mathcal{O} ( \frac{n^3 \log(nD_{max})}{\varepsilon^2} \cdot nm(n+m)\log(nm) + \frac{n^3m \log(nD_{max})}{\varepsilon^2} \cdot nm ) \in \mathcal{O}(n^4m^2 \log(nD_{max})\log(nm) / \varepsilon^2)$.

\begin{theorem}
	\label{mainthmFPTAS}
	There exists a $\mathcal{O}((n^4m^2 \log(nD_{max}) \log(nm)) / \varepsilon^2)$ time combinatorial algorithm that determines a $(1-\varepsilon)$-CEEI when agents have arbitrary disutility values. 
\end{theorem}

Observe that if we do not stop our algorithm, until earning of all the agents are the same, then we reach an exact CE. Since each balanced flow computation finds the Nash-welfare maximizing allocation for the given MPB configuration, it follows that the MPB configuration cannot repeat. This together with Nash-welfare as the potential function puts the problem in PLS. This algorithm and argument extends to the Fisher model as well, where agents may have different earning requirements -- if the earning requirement of agent $i$ is $\eta_i$ then in every iteration sort them in the decreasing order of $\frac{e_i}{\eta_i}$.

\begin{corollary}
The problem of finding competitive equilibrium in the Fisher model is in PLS. 
\end{corollary}

\section{PPAD-Hardness for the Exchange Model}\label{ppadhardness}
%of Finding CE under the Sufficient Condition}
In this section, we show that finding (approximate) CE is intractable in the exchange model with chores and linear valuations. %is intractable even if they exist under the conditions of Theorem~\ref{}. %chore division may still be intractable even for the instances that satisfy Conditions 1 and 2 mentioned in Section~\ref{sufficiency}. 
It is well-known that under exchange model a CE may not exist, and recently \cite{ChaudhuryGMM22} showed existence under certain mild sufficiency conditions, namely conditions $SC_1$ and $SC_2$ of Theorem \ref{sufficientcondition}. 
We show that even for the instances satisfying these conditions, it is PPAD-hard to find a $(1-1/poly(n))$-approximate CE (Definition \ref{def:ce}). In particular, we will show that any polynomial time algorithm that determines a $(1-1/poly(n))$-approximate CE on instances that satisfy conditions $SC_1$ and $SC_2$ of Theorem~\ref{sufficientcondition}, will yield an algorithm to find a $1/n$-approximate Nash equilibrium in a \emph{normalized polymatrix game}. The latter is known to be PPAD-hard~\cite{chen2017complexity}. Next we recall the normalized polymatrix game problem:

\begin{problem*}\textbf{(Normalized Polymatrix Game)~\cite{chen2017complexity}}\\
	\textbf{Given}: A $2n \times 2n$ rational matrix $\M$ with every entry in $[0,1]$ and $\M_{i,2j-1} + \M_{i,2j} = 1$ for all $i \in [2n]$ and $j \in [n]$ .\\
	\textbf{Find}: An approximate Nash equilibrium strategy vector $x \in \mathbb{R}^{2n}_{\geq 0}$ such that $x_{2i-1} + x_{2i} = 1$ and 
	\begin{align*}
	& x^T \cdot \M_{*,{2i-1}} > x^T \cdot \M_{*,2i} + \tfrac{1}{n} \implies x_{2i} = 0.\\
	& x^T \cdot \M_{*,{2i}} > x^T \cdot \M_{*,2i-1} + \tfrac{1}{n} \implies x_{2i-1} = 0.
	\end{align*}
	where $M_{*,k}$ represents the $k^{\mathit{th}}$ column of the matrix $\M$.	
\end{problem*}

From the next section onward, we elaborate our construction and proof of reduction: We first introduce all agents and chores. Thereafter, we define the disutility matrix and endowment matrix and show that our instance satisfies the sufficiency conditions of Theorem~\ref{sufficientcondition}, and therefore admits a CE. Then, we show that our instance and the prices at $(1-1/poly(n))$-approximate CE exhibits the five properties of \emph{pairwise equal endowments}, \emph{(approximate) fixed earning}, \emph{(approximate) price equality}, \emph{price regulation} and \emph{reverse ratio amplification} (as discussed in Section~\ref{mainres3}), and thus in polynomial-time we can construct the equilibrium strategy vector $x$ for $I$ from any $(1-\frac{1}{poly(n)})$-approximate CE in $E(I)$. The reader is highly encouraged to read Section~\ref{mainres3} before reading the elaborate version of the proof to get the idea of the overall proof sketch. 

\subsection{Agent and Chore Sets}
We define the set of $K = 2c \cdot \lceil \log(n) \rceil$ many sets of chores, where $c = 3$ (observe crucially that $K$ is even),
\begin{align*}
B_k &= \left\{\cup_{i \in [2n]} b^k_i \right\} &\text{for all $k \in [K]$}, 
\end{align*} 
and $K$ many sets of agents 
\begin{align*}
A_k &=\begin{cases} 
                \left\{a^1_i\ |\ i \in [2n]\right\} \cup \left\{a'_i\ |\ i \in [2n]\right\}  &\text{when $k=1$}, \\
				\left\{a^k_i\ |\ i \in [2n]\right\}  \cup \left\{\A^k_i\ |\ i \in [n] \right\} &\text{when $ 2 \leq k \leq K-1$}, \\
				\left\{a^K_{i,j}\ |\ i,j \in [2n]\right\} \cup \left\{\A^K_i\ |\ i \in [n] \right\}   &\text{when $k = K$}.
	\end{cases}
\end{align*} 

We remark that the sets $A_1$, $A_K$ of agents and sets $B_1$, $B_K$ of chores are to enforce the fixed earning, price equality and price regulation properties as mentioned in sketch of the reduction in Section~\ref{mainres3}, while the sets $A_k$ of agents and $B_k$ of chores for all $2 \leq k \leq K-1$ are to primarily enforce reverse ratio amplification property as mentioned in Section~\ref{mainres3}.  We now define the disutility matrix and the endowment matrix of the instance.

\paragraph{Disutility Matrix and the Disutility Graph.} The disutility graph for our instance will be a disjoint union of  complete bipartite graphs and the entries in our disutility matrix will be to enforce price-regulation and reverse ratio-amplification properties. We now describe the disutility matrix: We define only the disutility values in the matrix that are finite (the disutility of all agent-chore pair not mentioned should be assumed to be $\infty$).  For all $k \in [K]$, for each pair of chores $b^k_{2i-1}$ and $b^k_{2i}$, there are a set of agents that have finite disutility towards them and have infinite disutility towards all other chores; Additionally, these agents also happen to be either in $A_k$ or $A_{k-1}$ (indices are modulo $K$). We now outline these agents and their disutilities for every $k \in [K]$. To define the finite entries in the disutility matrix, we introduce the scalars $\tfrac{1}{n^{3c}} = \alpha_1, \alpha_2, \dots , \alpha_K$ such that each $\alpha_{i+1} = \tfrac{3}{2} \cdot \alpha_i$ for all $i \in [K-1]$. Before we define the disutility matrix, we make an obvious claim about the scalars $\alpha_i$ for all $i \in [K]$, which will be useful later,
\begin{claim}
	\label{alpha-technical}
	We have $ n^c \cdot \alpha_1 < \alpha_K \leq \tfrac{1}{n^c}$.
\end{claim} 
\begin{proof}
	We first show the lower bound. We have $\alpha_K = (\tfrac{3}{2})^{K-1} \cdot \alpha_1 = (\tfrac{3}{2})^{2c \lceil \log (n) \rceil-1} \cdot \alpha_1 > 2^{c \log(n)} \cdot \alpha_1 = n^c \cdot \alpha_1$. Similarly, for the upper bound, we have, $\alpha_K = (\tfrac{3}{2})^{K-1} \cdot \alpha_1 = (\tfrac{3}{2})^{2c \lceil \log (n) \rceil -1} \cdot \alpha_1 < 2^{2c \log(n)} \cdot \alpha_1 = n^{2c} \cdot \alpha_1 = \tfrac{1}{n^{c}}$ (as $\alpha_1 = \tfrac{1}{n^{3c}}$). 
\end{proof}
We now define the disutility matrix:
\begin{itemize}
	\item $k =1$: For each $i \in [n]$, we first define the disutilities of the agents that have finite disutility for chores $b^1_{2i-1}$ and $b^1_{2i}$. For each $i \in [n]$ we have,
							\begin{align*}
								d(a^K_{i',2i-1},b^{1}_{2i-1}) &= (1- \alpha_1) &\text{and}&    & d(a^K_{i',2i-1},b^{1}_{2i}) &= (1+ \alpha_1) &\text{for all $i' \in [2n]$}\\
							    d(a^K_{i',2i},b^{1}_{2i-1}) &= (1+ \alpha_1)   &\text{and}&    & d(a^K_{i',2i},b^{1}_{2i}) &= (1 - \alpha_1) &\text{for all $i' \in [2n]$}\\
							    d(a'_{2i-1},b^{1}_{2i-1}) &= (1- \alpha_1) &\text{and}&    & d(a'_{2i-1},b^{1}_{2i}) &= (1 +  \alpha_1)\\
							    d(a'_{2i},b^{1}_{2i-1}) &= (1+ \alpha_1)   &\text{and}&    & d(a'_{2i},b^{1}_{2i}) &= (1 - \alpha_1).
							\end{align*}   
				   		Therefore, for each $i \in [n]$, we have a component $\DG_i^1$ in the disutility graph which is a complete bipartite graph comprising of agents $\left\{a^K_{i',2i-1}\ |\ i' \in [2n] \right\} \bigcup \left\{ a^K_{i',2i}\ |\ i' \in [2n] \right\} \bigcup \left\{a'_{2i-1},a'_{2i} \right\}$ and chores $\left\{ b_{2i-1}^1, b_{2i}^1 \right\}$ (see Figure~\ref{disutility-graph} (left subfigure) for an illustration).
				   		
	\item $ 2 \leq k \leq K$: For each $i \in [n]$ we have,
				   \begin{align*}
					  d(a^{k-1}_{2i-1},b^{k}_{2i-1}) &= (1- \alpha_k) &\text{and}&    & d(a^{k-1}_{2i-1},b^{k}_{2i}) &= (1 +  \alpha_k)\\
					  d(a^{k-1}_{2i},b^{k}_{2i-1}) &= (1+ \alpha_k)   &\text{and}&    & d(a^{k-1}_{2i},b^{k}_{2i}) &= (1 - \alpha_k)\\
					  d(\A^k_{i},b^k_{2i-1}) &= (1- \alpha_k)         &\text{and}&    & d(\A^k_{i},b^k_{2i}) &= (1- \alpha_k) \enspace .\\
				  \end{align*}
				 Therefore, for every $k$ such that $2 \leq k \leq K$, for each $i \in [n]$,  we have a connected component $\DG_i^k$ in the disutility graph which is a complete bipartite graph comprising of agents $\left\{ a^{k-1}_{2i-1}, a^{k-1}_{2i}, \A^k_i \right\}$ and chores $\left\{ b_{2i-1}^k , b_{2i}^k \right\}$ (see Figure~\ref{disutility-graph} (right subfigure) for an illustration).   
\end{itemize}
It is clear that the disutility graph is a disjoint union of complete bipartite graphs, namely, the union of $\DG^k_i$ for all $i \in [n]$ and $k \in [K]$. Therefore,
\begin{center}
  $E(I)$ satisfies condition $SC_2$ of Theorem~\ref{sufficientcondition}.
\end{center}	

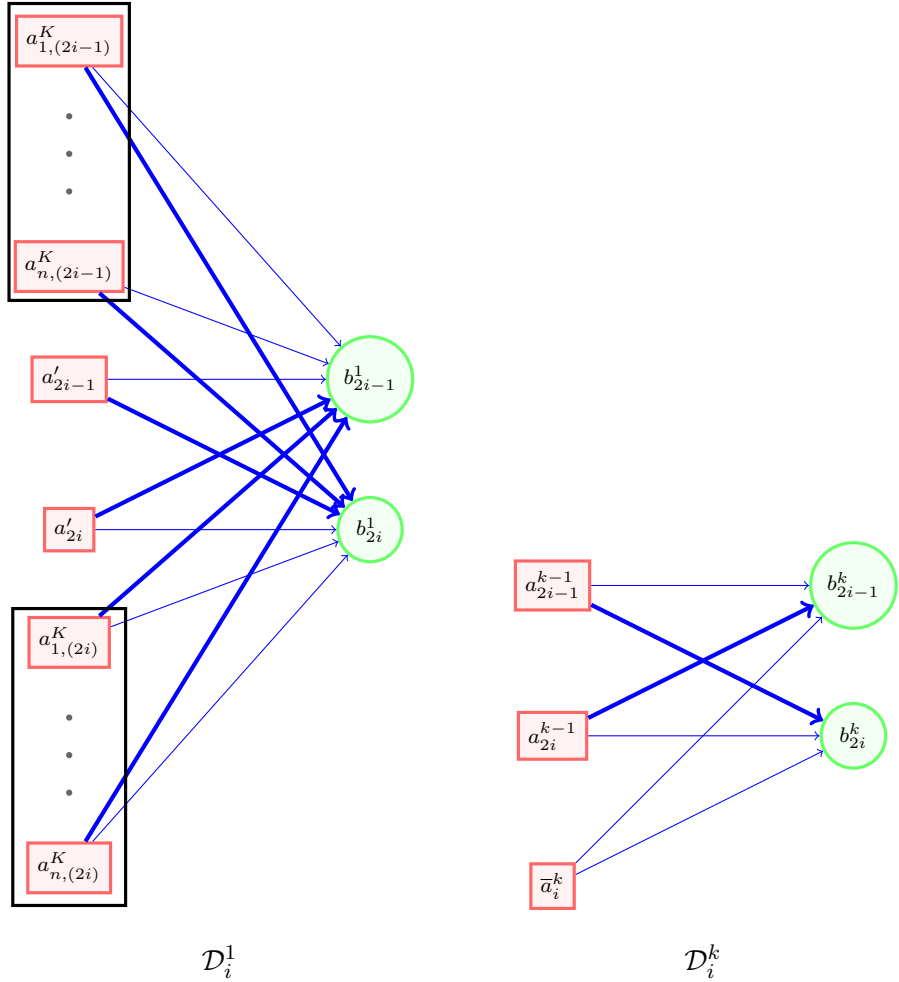
\begin{figure}[htbp]
	\centering
	\begin{subfigure}[b]{0.4 \textwidth}
\begin{tikzpicture}[
roundnode/.style={circle, draw=green!60, fill=green!5, very thick, minimum size=7mm},
squarednode/.style={rectangle, draw=red!60, fill=red!5, very thick, minimum size=5mm},
]
%Agents
\node[squarednode]      (a11)      at (0,8+2)                       {$\scriptstyle{a^K_{1,(2i-1)}}$};
\node[squarednode]      (an1)      at (0,5+2)                        {$\scriptstyle{a^K_{n,(2i-1)}}$};
\node[squarednode]      (a12)      at (0,4-2)                       {$\scriptstyle{a^K_{1,(2i)}}$};
\node[squarednode]      (an2)      at (0,1-2)                        {$\scriptstyle{a^K_{n,(2i)}}$};
\node[squarednode]      (a'1)      at (0,5.5)                       {$\scriptstyle{a'_{2i-1}}$};
\node[squarednode]      (a'2)      at (0,3.5)                        {$\scriptstyle{a'_{2i}}$};

%Goods

\node[roundnode]      (b1)      at (4,5.5)                       {$\scriptstyle{b_{2i-1}^1}$};
\node[roundnode]      (b2)      at (4,3.5)                        {$\scriptstyle{b_{2i}^1}$};

%Lines
\draw[blue,->] (a11)--(b1);
\draw[blue,->] (an1)--(b1);
\draw[blue,->] (a12)--(b2);
\draw[blue,->] (an2)--(b2);

\draw[blue,->,ultra thick]  (a11)--(b2);
\draw[blue,->, ultra thick] (an1)--(b2);
\draw[blue,->, ultra thick] (a12)--(b1);
\draw[blue,->, ultra thick] (an2)--(b1);

\draw[blue,->] (a'1)--(b1);
\draw[blue,->] (a'2)--(b2);
\draw[blue,->,ultra thick]  (a'1)--(b2);
\draw[blue,->, ultra thick] (a'2)--(b1);

%Rectangle
\draw[black, very thick] (-0.8,4.55+2) rectangle (0.8,8.5+2);
\draw[black, very thick] (-0.75,0.5-2) rectangle (0.75,4.45-2);

%Dots
\filldraw[color=black!60, fill=black!5, very thick](0,7+2) circle (0.02);
\filldraw[color=black!60, fill=black!5, very thick](0,6.5+2) circle (0.02);
\filldraw[color=black!60, fill=black!5, very thick](0,6+2) circle (0.02);

\filldraw[color=black!60, fill=black!5, very thick](0,3-2) circle (0.02);
\filldraw[color=black!60, fill=black!5, very thick](0,2.5-2) circle (0.02);
\filldraw[color=black!60, fill=black!5, very thick](0,2-2) circle (0.02);

\node at (2,-2.25) {$\DG^1_i$};
\end{tikzpicture}
	\end{subfigure}		
	\begin{subfigure}[b]{0.4 \textwidth}
\begin{tikzpicture}
%Agents
\node[rectangle, draw=red!60, fill=red!5, very thick, minimum size=5mm]      (a11)      at (0,5)                       {$\scriptstyle{a^{k-1}_{2i-1}}$};
\node[rectangle, draw=red!60, fill=red!5, very thick, minimum size=5mm]      (a22)      at (0,3)                        {$\scriptstyle{a^{k-1}_{2i}}$};
\node[rectangle, draw=red!60, fill=red!5, very thick, minimum size=5mm]      (a3)      at (0,1)                       {$\scriptstyle{\A^k_{i}}$};
%Goods
\node[circle, draw=green!60, fill=green!5, very thick, minimum size=7mm]      (b1)      at (4,5)                       {$\scriptstyle{b_{2i-1}^k}$};
\node[circle, draw=green!60, fill=green!5, very thick, minimum size=7mm]      (b2)      at (4,3)                        {$\scriptstyle{b_{2i}^k}$};
%Lines
\draw[blue,->] (a11)--(b1);
\draw[blue,->] (a22)--(b2);
\draw[blue,->] (a3)--(b1);
\draw[blue,->] (a3)--(b2);
\draw[blue,->,ultra thick]  (a11)--(b2);
\draw[blue,->, ultra thick] (a22)--(b1);
\node at (2,0) {$\DG^k_i$};
\end{tikzpicture}
	\end{subfigure}	
	\caption{Illustration of the disutility graph corresponding to the disutility matrix: On the left, we have the component $\DG_i^1$, and on the right we have $\DG_i^k$ when $2 \leq k \leq K$. The edges are colored in order to also encode the disutility matrix. The thin blue edges from agents to chores depict a disutility of $(1 - \alpha_1)$ for $\DG^1_i$ (left), and $(1 - \alpha_k)$ for $\DG^k_i$ when $2 \leq k \leq K$ (right). Similarly, the thick blue edges from agents to chores depict a disutility of $(1 + \alpha_1)$ for $\DG^1_i$ (left) and $(1 + \alpha_k)$ for $\DG^k_i$ (right).} 
	\label{disutility-graph}
\end{figure}

\paragraph{Endowment Matrix.} All agents in $A_k$ have endowments of chores only in $B_k$ for all $k \in [K]$. We only mention the non-zero agent-chore endowments (all agent-chore endowments, if not mentioned, are zero). 
\begin{itemize}
	\item $k=1$: For each $i \in [2n]$ we have,
	\begin{align*}
	w(a^1_i,b^1_i) &=n.
	\end{align*}
	Also, for each $i \in [n]$ we have 
	\begin{align*}
	w(a'_{2i-1},b^1_{2i-1})=w(a'_{2i-1},b^1_{2i})  &= \frac{1}{2} \cdot (1 - \alpha_K) \cdot (2n - \sum_{j \in [2n]} \M_{j,2i-1})\\
	w(a'_{2i},b^1_{2i-1})  = w(a'_{2i},b^1_{2i}) &= \frac{1}{2} \cdot  (1 - \alpha_K) \cdot (2n - \sum_{j \in [2n]} \M_{j,2i}).
	\end{align*}
	\item $2 \leq k \leq K-1$: For each $i \in [n]$, we have,
	\begin{align*}
	w(a^k_{2i-1},b^k_{2i-1}) &=n      &\text{and}&     & w(a^k_{2i},b^k_{2i})&=n\\
	w(\A^k_i, b^k_{2i-1}) &= \delta_k &\text{and}& & w(\A^k_i, b^k_{2i}) &= \delta_k\enspace ,
	\end{align*}
	where $\delta_k  = \tfrac{n \cdot \alpha_k}{2}$. The reason behind the exact choice of the value of $\delta_k$ will become explicit when we show that our instance satisfies the reverse ratio amplification property in Section~\ref{property-satisfication}. As of now, the reader is encouraged to think of it just as a small scalar.
	\item $k =K$: For each $i \in [n]$ we have,
	\begin{align*}
	w(a^K_{2i-1,j},b^K_{2i-1}) &=\M_{2i-1,j}  &\text{and}&   & w(a^K_{2i,j},b^K_{2i}) &=\M_{2i,j}& &\text{ for all $j \in [2n]$}\\
	w(\A^K_i, b^K_{2i-1}) &= \delta_K        &\text{and}&    & w(\A^K_i, b^K_{2i}) &= \delta_K\enspace ,
	\end{align*}
	where $\delta_K = \tfrac{n \cdot \alpha_K}{2}$ (the reason behind the choice of value will become explicit in Section~\ref{property-satisfication}).   	     
\end{itemize}
 
\paragraph{Strongly Connected Economy Graph.} We now show that the economy graph $G$ of our instance is strongly connected. For ease of explanation, we introduce the notion of \emph{economy graph of components} $W = ([d], E_W)$, where there is an edge from $i \in [d]$ to $j \in [d]$, if and only if, there is an agent $a \in \DG_i$ that has a positive endowment of some chore in $b \in \DG_j$. We now make a claim that strong connectivity of $W$ implies strong connectivity of the economy graph $G$.

\begin{claim}
	\label{economygraphcomponents}
	 If $W$ is strongly connected then $G$ is also strongly connected.
\end{claim}

\begin{proof}
		Consider any two agents $a$ and $a'$. Let $a \in \DG_i$ and $a' \in \DG_j$.\footnote{Note that $j$ could also be equal to $i$.} Consider any chore $b$ that agent $a$ has a positive endowment of and let $\DG_{i'}$ be the component in the disutility graph that contains $b$.\footnote{Again, $i'$ could also be equal to $i$.} Then since $\DG_{i'}$ is a biclique in our instance, every agent in $\DG_{i'}$ has finite disutility for the chore $b$. Therefore, every agent in $\DG_{i'}$ is reachable from $a$ with an edge in the economy graph $G$. Now, since $W$ is strongly connected, there is a path $\ell_1 \rightarrow \ell_2 \rightarrow \dots \rightarrow \ell_k$ from $\ell_1 = i'$ to $\ell_k = j$. Let $a_{\ell_r}$ be the agent in the component $\DG_{\ell_r}$, that has a positive endowment of some chore in the component $\DG_{\ell_{r+1}}$ for all $r \in [k-1]$. Again, since each $\DG_{\ell_r}$ is a biclique, every agent in $\DG_{\ell_r}$ has a finite disutility for every chore in $\DG_{\ell_r}$. Thus, there is an edge in the economy graph $G$ from $a_{\ell_r}$ to every agent in $\DG_{\ell_r}$, in particular there is an edge between $a_{\ell_{r}}$ and $a_{\ell_{r+1}}$ in $G$. Thus, we have a path $a \rightarrow a_{\ell_1} \rightarrow \dots \rightarrow a_{\ell_{k-1}} \rightarrow a'$ in $G$. Therefore, if $W$ is strongly connected, then there is a path between any two agents in $G$, implying that $G$ is also strongly connected.     
\end{proof}

 From here on, we show that $W$ is strongly connected. Observe that the disutility graph consists of connected components $\DG^k_i$ for  $k \in [K]$ and $i \in [n]$. Also observe that every component $\DG^k_i$ in the disutility graph comprises of exactly two chores $b^k_{2i-1}$ and $b^k_{2i}$. Therefore, to show that there exists an edge from component $\DG^{k'}_{i'}$ to $\DG^k_i$ in $W$, it suffices to show that $\DG^{k'}_{i'}$ contains agents that own parts of chores $b^k_{2i-1}$ and $b^k_{2i}$. We now outline the edges in our exchange graph (see Figure~\ref{exchange-graph}):
 \begin{itemize}
 	\item For all $i \in [n]$, and $2 \leq k \leq K$ there is an edge in $W$ from $\DG^k_i$ to $\DG^{k-1}_i$: $\DG^k_i$ contains the agents $a^{k-1}_{2i-1}$ and $a^{k-1}_{2i}$ that own parts of chores $b^{k-1}_{2i-1}$ and $b^{k-1}_{2i}$ respectively (see Figure~\ref{exchange-graph}).
 	\item For all $i \in [n]$, there is an edge in $W$ from $\DG^1_i$ to $D^K_j$ for all $j \in [n]$: Consider any $j \in [n]$. Observe that the component $\DG^1_i$ contains the agents $a^K_{2j-1,2i}$ and $a^K_{2j,2i}$ and the agents $a^K_{2j-1,2i}$ and $a^K_{2j,2i}$ own parts of chores $b^K_{2j-1}$ and $b^K_{2j}$ respectively (see Figure~\ref{exchange-graph}).
 \end{itemize}

\begin{figure}
	\begin{center}
\begin{tikzpicture}[scale=1.25,
roundnode/.style={circle, draw=green!60, fill=green!5, very thick, minimum size=7mm},
squarednode/.style={rectangle, draw=red!60, fill=red!5, very thick, minimum size=5mm},
]
%Vertices
\node[roundnode]      (D11)      at (0,6)                       {$\scriptstyle{D^1_1}$};
\node[roundnode]      (Dk1)      at (1.5,6)                       {$\scriptstyle{D^{K}_1}$};
\node[roundnode]      (Dk'1)      at (3,6)                       {$\scriptstyle{D^{K-1}_1}$};
\node[roundnode]      (D21)      at (6,6)                       {$\scriptstyle{D^2_1}$};

\node[roundnode]      (D12)      at (0,3)                       {$\scriptstyle{D^1_i}$};
\node[roundnode]      (Dk2)      at (1.5,3)                       {$\scriptstyle{D^{K}_i}$};
\node[roundnode]      (Dk'2)      at (3,3)                       {$\scriptstyle{D^{K-1}_i}$};
\node[roundnode]      (D22)      at (6,3)                       {$\scriptstyle{D^2_i}$};

\node[roundnode]      (D13)      at (0,0)                       {$\scriptstyle{D^1_n}$};
\node[roundnode]      (Dk3)      at (1.5,0)                       {$\scriptstyle{D^{K}_n}$};
\node[roundnode]      (Dk'3)      at (3,0)                       {$\scriptstyle{D^{K-1}_n}$};
\node[roundnode]      (D23)      at (6,0)                       {$\scriptstyle{D^2_n}$};

%Dots
\filldraw[color=black!60, fill=black!5, very thick](5,6) circle (0.02);
\filldraw[color=black!60, fill=black!5, very thick](4,6) circle (0.02);
\filldraw[color=black!60, fill=black!5, very thick](4.5,6) circle (0.02);

\filldraw[color=black!60, fill=black!5, very thick](5,3) circle (0.02);
\filldraw[color=black!60, fill=black!5, very thick](4,3) circle (0.02);
\filldraw[color=black!60, fill=black!5, very thick](4.5,3) circle (0.02);

\filldraw[color=black!60, fill=black!5, very thick](5,0) circle (0.02);
\filldraw[color=black!60, fill=black!5, very thick](4,0) circle (0.02);
\filldraw[color=black!60, fill=black!5, very thick](4.5,0) circle (0.02);

\filldraw[color=black!60, fill=black!5, very thick](0,4) circle (0.02);
\filldraw[color=black!60, fill=black!5, very thick](0,4.75) circle (0.02);
%\filldraw[color=black!60, fill=black!5, very thick](0.5.5) circle (0.02);

\filldraw[color=black!60, fill=black!5, very thick](2,4) circle (0.02);
\filldraw[color=black!60, fill=black!5, very thick](2,4.75) circle (0.02);
%\filldraw[color=black!60, fill=black!5, very thick](2.5.5) circle (0.02);

\filldraw[color=black!60, fill=black!5, very thick](6,4) circle (0.02);
\filldraw[color=black!60, fill=black!5, very thick](6,4.75) circle (0.02);
%\filldraw[color=black!60, fill=black!5, very thick](6.5.5) circle (0.02);

\filldraw[color=black!60, fill=black!5, very thick](0,4-3) circle (0.02);
\filldraw[color=black!60, fill=black!5, very thick](0,4.75-3) circle (0.02);
%\filldraw[color=black!60, fill=black!5, very thick](0.5.5-3) circle (0.02);

\filldraw[color=black!60, fill=black!5, very thick](2,4-3) circle (0.02);
\filldraw[color=black!60, fill=black!5, very thick](2,4.75-3) circle (0.02);
%\filldraw[color=black!60, fill=black!5, very thick](2.5.5-3) circle (0.02);

\filldraw[color=black!60, fill=black!5, very thick](6,4-3) circle (0.02);
\filldraw[color=black!60, fill=black!5, very thick](6,4.75-3) circle (0.02);
%\filldraw[color=black!60, fill=black!5, very thick](6.5.5-3) circle (0.02);

%Edges
\draw[blue,->, thick] (D11)--(Dk1);
\draw[blue,->,thick] (Dk1)--(Dk'1);
\draw[blue,->,thick] (Dk'1)--(4,6);
\draw[blue,->, thick] (5,6)--(D21);

\draw[blue,->,thick] (D12)--(Dk2);
\draw[blue,->,thick] (Dk2)--(Dk'2);
\draw[blue,->,thick] (Dk'2)--(4,3);
\draw[blue,->,thick] (5,3)--(D22);

\draw[blue,->,thick] (D13)--(Dk3);
\draw[blue,->,thick] (Dk3)--(Dk'3);
\draw[blue,->,thick] (Dk'3)--(4,0);
\draw[blue,->,thick] (5,0)--(D23);

%Long lines

\draw[blue,->, thick] (D11)--(Dk2);
\draw[blue,->, thick] (D11)--(Dk3);
\draw[blue,->, thick] (D12)--(Dk1);
\draw[blue,->, thick] (D12)--(Dk3);
\draw[blue,->, thick] (D13)--(Dk2);
\draw[blue,->, thick] (D13)--(Dk1);

%Paths
\path[blue,->,thick, out=120,in=45]    (D21) edge (D11);
\path[blue,->,thick, out=120,in=45]    (D22) edge (D12);
\path[blue,->,thick, out=120,in=45]    (D23) edge (D13);
\end{tikzpicture}
	\end{center}
	\caption{Illustration of the strong connectivity of the economy graph of components of our instance. Observe that all nodes are reachable from any $\DG^1_i$ ($i \in [n]$). Also, from any arbitrary $\DG^{k'}_{i'}$, the node $\DG^{1}_{i'}$ is reachable and since every node is reachable from $\DG^1_{i'}$, every node is also reachable from $\DG^{k'}_{i'}$ as well. Therefore, the economy graph of components is strongly connected.}
	\label{exchange-graph} 
\end{figure}
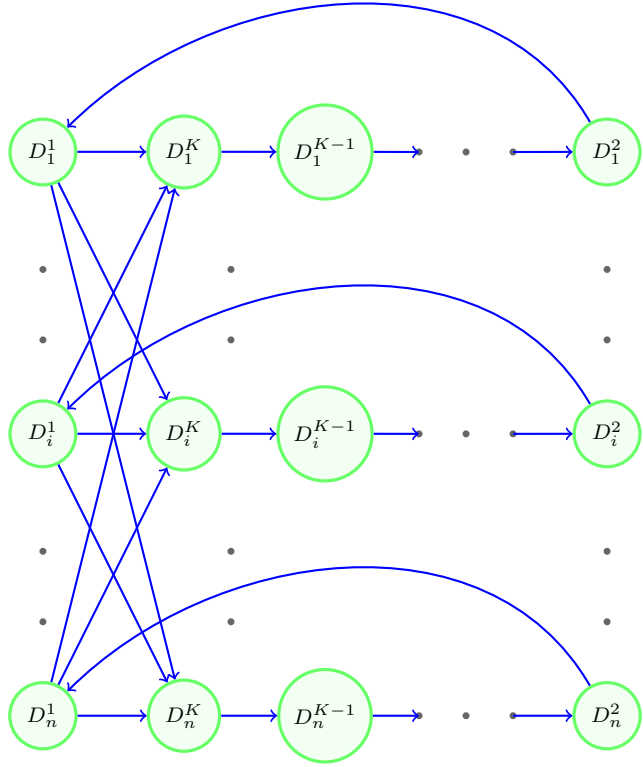

Observe that all nodes are reachable from any $\DG^1_i$ ($i \in [n]$). Also, from any arbitrary $\DG^{k'}_{i'}$, the node $\DG^{1}_{i'}$ is reachable and since every node is reachable from $\DG^1_{i'}$, every node is also reachable from $\DG^{k'}_{i'}$ as well. Therefore, the economy graph of components $W$, is strongly connected. Therefore, by Claim~\ref{economygraphcomponents} we have that,

\begin{center}
	$E(I)$ satisfies the condition $SC_1$ of Theorem~\ref{sufficientcondition}.
\end{center}

Thus, $E(I)$ satisfies conditions of Theorem~\ref{sufficientcondition} and therefore admits a CE and thereby a \eA CE for any $\epsilon\ge 0$ (see Definition \ref{def:ce}). Let $p(b^k_i)$ denote the price of chore $b^k_i$ at an \eA CE for \[\epsilon=\frac{\alpha_1}{200\cdot n} < \frac{1}{n^{3c+1}} \Rightarrow \epsilon=O(1/n^{10}).\] We now prove that our instance satisfies the required properties of \emph{pairwise equal endowments}, \emph{(approximate) price equality}, \emph{(approximate) fixed earning}, \emph{price regulation} and \emph{reverse ratio amplification}.

\subsection{$E(I)$ Satisfies All the Properties (Approximately)}
\label{property-satisfication}

\paragraph{Pairwise Equal Endowments.} Here, we show that for all $i \in [n]$ and for all $k \in [K]$ the total endowment of $b^k_{2i-1}$ equals the total endowment of $b^k_{2i}$ and the total endowments of each chore in $E(I)$ is $\mathcal{O}(n)$.
\begin{lemma}\label{total-endowments}
	For all $i \in [2n]$, the total endowments of chores $b^k_{2i-1}$ and $b^{k}_{2i}$ is 
	\begin{enumerate}
		\item $n + n \cdot (1- \alpha_K)$, if $k=1$. In particular, $a'_{2i-1}$ and $a'_{2i}$ \emph{together}, own $n \cdot (1- \alpha_K)$ units of chores $b^k_{2i-1}$ and $b^{k}_{2i}$ each.  
		\item $n + \delta_k$, if $2 \leq k \leq K$. 
	\end{enumerate}
\end{lemma}

\begin{proof}
	When $k=1$, the only agents that have positive endowments of $b^1_{2i}$  are $a^1_{2i}$ (has an endowment of $n$ ) , $a'_{2i}$ (has an endowment of $\tfrac{1}{2} \cdot (1 - \alpha_K) \cdot (2n - \sum_{j \in [2n]} \M_{j,2i})$) and $a'_{2i-1}$(has an endowment of  $\tfrac{1}{2} \cdot (1 - \alpha_K) \cdot (2n - \sum_{j \in [2n]} \M_{j,2i-1})$). Therefore, the total endowment of $b^1_{2i}$ from the agents $a'_{2i}$ and $a'_{2i-1}$ is 
	\begin{align*}
	&= \frac{1}{2} \cdot (1 - \alpha_K) \cdot (2n - \sum_{j \in [2n]} \M_{j,2i}) + \frac{1}{2} \cdot (1 - \alpha_K) \cdot (2n - \sum_{j \in [2n]} \M_{j,2i-1})\\
	&= \frac{1}{2} \cdot (1 - \alpha_K) \cdot (4n - \sum_{j \in [2n]} (\M_{j,2i} + \M_{j,2i-1})).
	\end{align*} 
	Recall that $\M_{j,2i} + \M_{j,2i-1} = 1$. Therefore, the total endowment of $b^1_{2i}$ from the agents $a'_{2i}$ and $a'_{2i-1}$ is 
	\begin{align*}
	&= \frac{1}{2} \cdot (1 - \alpha_K) \cdot (4n - 2n)\\
	&= (1 - \alpha_K) \cdot n.
	\end{align*} 
	Therefore, the total endowment of chore $b^1_{2i}$ is $n + n \cdot (1- \alpha_K)$.
	A similar argument will show that the total endowment of chore $b^1_{2i-1}$ is also $n + n \cdot (1 - \alpha_K)$ and that agents $a'_{2i-1}$ and $a'_{2i}$ \emph{together}, own $n \cdot (1- \alpha_K)$ units of it. 
	
	When $2 \leq k \leq K-1$, the only agents that have positive endowments of $b^k_{2i}$ are $a^k_{2i}$ (has an endowment of $n$) and  $\A_i^k$ (has an endowment of $\delta_k$). Therefore, the total endowment is $n + \delta_k$. A similar argument will show that the total endowment of chore $b^k_{2i-1}$ is also $n + \delta_k$. 
	
	When $k =K$, the only agents that have positive endowments of $b^K_{2i}$ are the agents $a_{2i,j}^K$ (has an endowment of $\M_{2i,j}$) for all $j \in [2n]$ and the agent $\A^K_i$ (has an endowment of $\delta_K$). Therefore, the total endowment of chore $b^K_{2i}$ is
	\begin{align*}
	&=\sum_{j \in [2n]} \M_{2i,j} + \delta_K\\
	&=\sum_{j \in [n]} (\M_{2i,2j-1} + \M_{2i,2j}) + \delta_K\\
	&=\sum_{j \in [n]} 1 + \delta_K \\
	&=n + \delta_K
	\end{align*}
	A similar argument will show that the total endowment of chore $b^K_{2i-1}$ is also $n + \delta_K$.
\end{proof}

\paragraph{Price Equality.} Here we will show that the sum of prices of chores $b^k_{2i-1}$ and $b^k_{2i}$ are almost same for all $i\in [n]$ and $k\in[K]$. Let us define \[\pi^k_i = p(b^k_{2i-1}) + p(b^k_{2i}),\ \ \ \forall i \in [n], k \in [K]\enspace .\] 

Here on we will use the properties of approximate CE (Definition \ref{def:ce}). Recall that, at an \eA CE, we have $(i)$ complete allocation each chore, $(ii)$ every agent only consumes her minimum pain-per-buck chores, and $(iii)$ every agent $i$ earns total price of her endowment up to $(1 \pm \epsilon)$ factor. And that, for two quantities $x$ and $y$, by $x = (1\pm \epsilon) y$ we mean $(1- \epsilon) y\le x \le (1+ \epsilon) y$. In addition, now on \[ \mbox{by $x = (1\pm \epsilon)^d y$ we mean $(1- \epsilon)^d y\le x \le (1+ \epsilon)^d y$.}\]
%Since the prices corresponding to a  CE is scale-invariant, we can assume without loss of generality that $\pi^1_1 = 2$. We now state the main lemma of price equality:

\begin{lemma}
	\label{price-equality}
	For all $i,i' \in [n]$ and for all $k,k' \in [K]$, we have $\pi^{k'}_{i'}= (1\pm O(n\epsilon)) \pi^k_i$. 
\end{lemma}

\begin{proof}
Since $\epsilon < 1/n$ it suffices to show that $\pi^{k'}_{i'} = (1+\pm \epsilon))^{O(n)} \pi^k_i$. We show this in two steps: First we show that for each $i\in [n]$, we have $\pi^(k+1)_i =  (1\pm \epsilon)^2 \pi^k_i$ for all $k < K$, implying that $\pi^{k'}_{i}= (1\pm \epsilon))^{O(n)} \pi^k_i$ for all $k,k'\in [K]$. 

Then we show that $(1-\epsilon)^2 (\sum_{j \in [n]} \pi^K_j) \le n\pi^1_i = (1\pm \epsilon)^2 (\sum_{j \in [n]} \pi^K_j)$ for all $i \in [n]$. This together with the above imply that $\pi^1_{i'} \le (1\pm \epsilon))^{O(n)} \pi^1_i$ for all $i,j \in [n]$. Putting these together proves the lemma. 
%Since for all $i \in [n]$ and $k \in [K]$, $\pi^k_i = \pi^1_i$ and $\pi^1_i = \pi^1_1$, we will have that $\pi^k_i = \pi^1_1 = 2$. 
 
We first show $\pi^(k+1)_i\le (1\pm \epsilon)^2 \pi^k_i$ for all $i\le n, k < K$, wlog let $i=1$. Observe that the agents $a^{k}_{1}$, $a^{k}_{2}$, $\A^{(k+1)}_1$ and chores $b^{(k+1)}_{1}$, $b^{(k+1)}_{2}$ form the connected component $\DG^{(k+1)}_1$ in the disutility graph. That means these three agents earn all their money by consuming only these two chores, and no one else consumes these chores. The total money supply of these two chores is $(n+\delta_{(k+1)})\pi^{(k+1)}_1$. 

As per Definition \ref{def:ce}, at \eA CE, an agents earn as much as the total price of their endowments up to $(1\pm \epsilon)$ factor. Now since $\A^{(k+1)}_1$ owns $\delta_{(k+1)}$ units of both $b^{(k+1)}_{1}$ and $b^{(k+1)}_{2}$ only, her total cost is $\delta_{(k+1)} \pi^{(k+1)}_1$. Total price of chores owned by agents $a^{k}_{1}$ and $a^{k}_{2}$ is $n\pi^k_1$. At equilibrium demand should be equals supply, implying that,
\[
\begin{array}{ll}
& (1-\epsilon) (\delta_{(k+1)} \pi^{(k+1)}_1 + n\pi^k_1) \le (n+\delta_{(k+1)})\pi^{(k+1)}_1
\le (1+\epsilon) (\delta_{(k+1)} \pi^{(k+1)}_1 + n\pi^k_1) \\
\Rightarrow & (1-\epsilon) n \pi^k_1 - \epsilon \delta_{(k+1)} \pi^{(k+1)}_1 \le n \pi^{(k+1)}_1 \le (1+\epsilon) n \pi^k_1 + \epsilon \delta_{(k+1)} \pi^{(k+1)}_1 \\
\Rightarrow & \frac{(1-\epsilon)}{(1+\epsilon\delta_{(k+1)}/n)}  \pi^k_1 \le  \pi^{(k+1)}_1 \le \frac{(1+\epsilon)}{(1-\epsilon\delta_{(k+1)}/n)}  \pi^k_1 \\
\Rightarrow & (1-\epsilon)^2 \pi^k_1 \le  \pi^{(k+1)}_1 \le (1+\epsilon)^2 \pi^k_1 
\end{array}
\]

where the last implication follows from the fact that for $\epsilon<1/n$ we have $(1-\epsilon) < \frac{1}{(1+\epsilon\delta_{(k+1)}/n)}$ and $(1+\epsilon) < \frac{1}{(1-\epsilon\delta_{(k+1)}/n)}$.

 We now show that $n\pi^1_i = (1\pm)\epsilon)^2\sum_{j \in [n]} \pi^K_j$: This time, we look into the connected component $\DG^1_i$ of the disutility graph. We can claim that the agents $\{a^K_{j,2i-1}\ |\ j \in [2n]\}$, $\{a^K_{j,2i}\ |\ j \in [2n]\} $ and the agents $a'_{2i-1}$ and $a'_{2i}$ earn all of their money at an approximate CE from chores $b^1_{2i-1}$ and $b^1_{2i}$, and these are the only agents who consume these two chores. Total money supply of these two chores is $(n+n(1-\alpha_K))\pi^1_i$ (by Lemma \ref{total-endowments})
 
Observe that the total endowment of agents $a'_{2i-1}$ and $a'_{2i}$ is $n(1-\alpha_K)$ units of chores $b^1_{2i-1}$ and $b^1_{2i}$ only, and hence together they must earn $n(1-\alpha_K)\pi^1_i$ up to $(1\pm \epsilon)$. For agents $\{a^K_{j,2i-1}\ |\ j \in [2n]\}$ and  $\{a^K_{j,2i}\ |\ j \in [2n]\}$, recall that each $a^{K}_{\ell,\ell'}$ owns $\M_{\ell,\ell'}$ units of $b^K_{\ell}$. Therefore, together they need to earn the following amount up to $(1\pm \epsilon)$ factor. 
\[
\begin{array}{l}
\sum_{j \in [2n]} \M_{j,2i} \cdot p(b^K_j) +  \sum_{j \in [2n]} \M_{j,2i-1} \cdot p(b^K_j)\\
 = \sum_{j \in [2n]}(\M_{j,2i} + \M_{j,2i-1}) \cdot p(b^K_j)\\
 = \sum_{j \in [2n]} p(b^K_j) \ \ \text{(using $\M_{j,2i-1} + \M_{j,2i} =1$)}\\
 =\sum_{j \in [n]} \pi^K_j 
\end{array}
\]

Again, equating supply with demand we get,
\[
\begin{array}{ll}
& (1-\epsilon)(n(1-\alpha_K)\pi^1_i + \sum_{j \in [n]} \pi^K_j) \le (n+n(1-\alpha_K))\pi^1_i \le (1+\epsilon)(n(1-\alpha_K)\pi^1_i + \sum_{j \in [n]} \pi^K_j)\\
\Rightarrow & 
(1-\epsilon)\sum_{j \in [n]} \pi^K_j - \epsilon (n(1-\alpha_K)\pi^1_i \le n \pi^1_i \le (1+\epsilon)\sum_{j \in [n]} \pi^K_j + \epsilon (n(1-\alpha_K)\pi^1_i \\
\Rightarrow & 
\frac{(1-\epsilon)}{(1+\epsilon(1-\alpha_K))} \sum_{j \in [n]} \pi^K_j \le n \pi^1_i \le \frac{(1+\epsilon)}{(1-\epsilon(1-\alpha_K))} \sum_{j \in [n]} \pi^K_j\\
\Rightarrow & 
(1-\epsilon)^2 \sum_{j \in [n]} \pi^K_j < n \pi^1_i < (1+\epsilon)^2 \sum_{j \in [n]} \pi^K_j
\end{array}
\]

where the last implication follows from the fact that for $\epsilon,\alpha_K < 1/n$ we have $(1-\epsilon) < \frac{1}{(1+\epsilon(1-\alpha_K))}$ and $(1+\epsilon) > \frac{1}{(1-\epsilon(1-\alpha_K))}$.
\end{proof}

\paragraph{(Approximately) Fixed Earning.} Here, we show that in every  CE, the earning of each agent $a'_{i}$ for $i \in [2n]$ is fixed up to $(1\pm O(n)\epsilon)$ factor. Note that (approximate) CE prices are scale invariant.

\begin{lemma}
	\label{fixed-earning}
	Let $\min_{i\in[n], k\le K} \pi^i_k=2$. Then for all $i \in [2n]$, we have that the earning of agent $a'_i$ is $(1\pm \epsilon O(n)) (1- \alpha_K) \cdot (2n - \sum_{j \in [2n]} \M_{j,i})$. 
\end{lemma}

\begin{proof}
	Let $i = 2i'$. Then agent $a'_{2i'}$ owns $\tfrac{1}{2} \cdot (1- \alpha_K) \cdot (2n - \sum_{j \in [2n]} \M_{j,2i'})$ units of both chores $b^1_{2i'-1}$ and $b^1_{2i'}$. Since the earning of any agent at an \eA CE equals the sum of prices of chores she owns up to $(1\pm \epsilon)$ factor, we have that the earning of agent $2i'$ is  
	\begin{align*}
	&=\frac{(1\pm \epsilon)}{2} \cdot(1- \alpha_K) \cdot (2n - \sum_{j \in [2n]} \M_{j,2i'}) \cdot (p(b^1_{2i'-1}) + p(b^1_{2i'}))\\ 
	&=\frac{(1\pm \epsilon)}{2} \cdot(1- \alpha_K) \cdot (2n - \sum_{j \in [2n]} \M_{j,2i'}) \cdot \pi^1_{i'} \\
	&=\frac{(1\pm \epsilon)}{2} \cdot (1- \alpha_K) \cdot (2n - \sum_{j \in [2n]} \M_{j,2i'}) \cdot 2 (1\pm O(n)\epsilon) &\text{(by Lemma~\ref{price-equality})}.\\
	&= (1\pm O(n)\epsilon) (1- \alpha_K) \cdot (2n - \sum_{j \in [2n]} \M_{j,2i'}).
	\end{align*}
	Similarly, when $i = 2i'-1$ we can show that the total earning of agent $a'_{2i'-1}$ is $(1- \alpha_K) \cdot (2n - \sum_{j \in [2n]} \M_{j,2i'-1})$. Thus the total earning of any agent $a'_i$ in a  CE is $(1\pm O(n)\epsilon) (1- \alpha_K) \cdot (2n - \sum_{j \in [2n]} \M_{j,i})$. 
\end{proof}

\paragraph{Price Regulation.} Here, we show that for all $k \in [K]$ and $i \in [2n]$ the ratio of the prices of chores $b^k_{2i-1}$ and $b^k_{2i}$ is bounded.

\begin{lemma}
	\label{price-regulation}
	For all $k \in [K]$ and for all $i \in [n]$,  we have $\tfrac{1 - \alpha_k}{1 + \alpha_k} \leq \tfrac{p(b^k_{2i-1})}{p(b^k_{2i})} \leq \tfrac{1 + \alpha_k}{1 - \alpha_k}$.
\end{lemma}

\begin{proof}
	We prove the lower bound ( $\tfrac{1 - \alpha_k}{1 + \alpha_k} \leq \tfrac{p(b^k_{2i-1})}{p(b^k_{2i})}$) by contradiction. The proof for the upper bound is symmetric. So assume that $\tfrac{1 - \alpha_k}{1 + \alpha_k} > \tfrac{p(b^k_{2i-1})}{p(b^k_{2i})}$. In that case, none of the agents in the connected component $\DG^k_i$ will do any part of chore $b^k_{2i-1}$ (as the disutility to price ratio of $b^k_{2i-1}$ will be strictly more than that of $b^k_{2i}$). Since all the other agents have a disutility of $\infty$ for $b^k_{2i-1}$, it will remain unallocated. Therefore, the current prices for chores are not the prices corresponding to a  CE, which is a contradiction. 
\end{proof}

\paragraph{Reverse Ratio Amplification.} Lastly, we show the property that when the price of chore $b^k_i$ is at a limit, then the price of chore $b^{k+1}_i$ is at the opposite limit, i.e., when $p(b^k_i) = 1 + \alpha_k$, then we have $p(b^{k+1}_i) = 1 - \alpha_{k+1}$ and similarly when $p(b^k_i) = 1 - \alpha_k$, then we have $p(b^{k+1}_i) = 1 + \alpha_{k+1}$.  
 
\begin{lemma}
	\label{rev-ratio-amplification1}
	 For all $1 \leq k < K$ and $i \in [n]$, we have that,
	 \begin{enumerate}
	 	\item if $\tfrac{p(b^{k}_{2i-1})}{p(b^k_{2i})} = \tfrac{1 - \alpha_k}{1+\alpha_k}$, then $\tfrac{p(b^{k+1}_{2i-1})}{p(b^{k+1}_{2i})} = \tfrac{1 + \alpha_{k+1}}{1-\alpha_{k+1}}$, and  
	 	\item if $\tfrac{p(b^{k}_{2i-1})}{p(b^k_{2i})} = \tfrac{1 + \alpha_k}{1 - \alpha_k}$, then $\tfrac{p(b^{k+1}_{2i-1})}{p(b^{k+1}_{2i})} = \tfrac{1 - \alpha_{k+1}}{1 + \alpha_{k+1}}$.
	 \end{enumerate}
\end{lemma}

\begin{proof}
	We just show the proof of part 1. The proof for part 2 is symmetric. Let us assume that $\tfrac{p(b^{k}_{2i-1})}{p(b^k_{2i})} = \tfrac{1 - \alpha_k}{1+\alpha_k}$, and that $p(b^{k}_{2i-1})+p(b^k_{2i})=2$ (using scale invariance of CE prices). Then, $p(b^k_{2i-1}) = 1 - \alpha_k$ and $p(b^k_{2i}) = 1+ \alpha_k$. From the proof of Lemma \ref{price-equality}, we have that $p(b^{k+1}_{2i-1})+p(b^{k+1}_{2i})=2(1\pm \epsilon)^2$ and by Lemma \ref{price-regulation} we have $\tfrac{1 - \alpha_{k+1}}{1+\alpha_{k+1}} \le \tfrac{p(b^{k+1}_{2i-1})}{p(b^{k+1}_{2i})} \le \tfrac{1 + \alpha_{k+1}}{1-\alpha_{k+1}}$. Therefore, $(1-\alpha_{k+1})(1-\epsilon)^2 \le p(b^{k+1}_{2i-1}), p(b^{k+1}_{2i}) \le (1+\alpha_{k+1})(1+\epsilon)^2$.

	Observe that agent $a^k_{2i}$ owns $n$ units of chore $b^k_{2i}$ and has finite disutility only for the chores $b^{k+1}_{2i-1}$ and $b^{k+1}_{2i}$ ($a^k_{2i}$ belongs in the connected component $\DG^{k+1}_i$). Since at a \eA CE, the total earning of agent $a^k_{2i}$ equals the sum of prices of chores she owns up to $(1\pm \epsilon)$ factor, we have that $a^k_{2i}$ earns $(1\pm \epsilon) n \cdot p(b^k_{2i}) = (1\pm \epsilon) n(1 + \alpha_k)$ amount of money from chores $b^{k+1}_{2i-1}$ and $b^{k+1}_{2i}$. We claim that it suffices to show that $a^k_{2i}$ earns some of her money from the chore $b^{k+1}_{2i-1}$: This would immediately imply that $\tfrac{d(a^k_{2i}, b^{k+1}_{2i-1})}{p(b^{k+1}_{2i-1})} \leq \tfrac{d(a^k_{2i}, b^{k+1}_{2i})}{p(b^{k+1}_{2i})}$, further implying that $\tfrac{p(b^{k+1}_{2i-1})}{p(b^{k+1}_{2i})} \geq \tfrac{1 + \alpha_{k+1}}{1- \alpha_{k+1}}$. However, by Lemma~\ref{price-regulation}, we have that $\tfrac{p(b^{k+1}_{2i-1})}{p(b^{k+1}_{2i})} \leq \tfrac{1 + \alpha_{k+1}}{1- \alpha_{k+1}}$, and thus we can conclude that $\tfrac{p(b^{k+1}_{2i-1})}{p(b^{k+1}_{2i})} = \tfrac{1 + \alpha_{k+1}}{1- \alpha_{k+1}}$.

    For the rest of the proof, we show that $a^k_{2i}$ earns positive amount of money from chore $b^{k+1}_{2i-1}$. We prove this by contradiction. So let us assume that $a^k_{2i}$ earns all her money of at least $(1- \epsilon) n \cdot (1+\alpha_k)$, only from chore $b^{k+1}_{2i}$. We will now show that the current prices of chores are not the prices corresponding to a  CE by distinguishing between two cases. Recall that $\delta_{k}=\frac{n}{2}\alpha_k$ for all $k\in [K]$.
	\begin{itemize}
		\item \textbf{$p(b^{k+1}_{2i}) > p(b^{k+1}_{2i-1})$:} 
		Observe that in this case, agent $\A^{k+1}_i$ will also earn all of her money from $b^{k+1}_{2i}$ only (as the disutility to price ratio of $b^{k+1}_{2i}$ is strictly smaller than that of $b^{k+1}_{2i-1}$). Therefore, we have that the total money agents $a^k_{2i}$ and $\A^{k+1}_i$ earn from $b^{k+1}_{2i}$ is,
		\begin{align*}
		  &\ge (1-\epsilon) (\pi^{k+1}_i \delta_{k+1} + n \cdot (1 + \alpha_{k}))\\
		  &\ge (1-\epsilon) (2(1-2\epsilon) \frac{n}{2}\alpha_{k+1} + n \cdot (1 + \frac{2}{3} \cdot \alpha_{k+1})\\ & \quad \quad (\text{as $\delta_{k+1}=\frac{n}{2}\alpha_{k+1}$, $\alpha_{k+1} = \frac{3}{2} \cdot \alpha_k$, and $\pi^{k+1}_i \ge 2(1-2\epsilon)$})\\
		  &=  (1-\epsilon) \cdot (n(1 + \alpha_{k+1}) + \tfrac{2n}{3} \alpha_{k+1} - 2\epsilon n \alpha_{k+1})\\
		  &= \frac{(1-\epsilon)}{(1+3\epsilon)} (1+3\epsilon)[(n+\tfrac{n}{2}\alpha_{k+1}) (1+\alpha_{k+1}) - \tfrac{n}{2}\alpha_{k+1} \tfrac{2n}{3} \alpha_{k+1} - 2\epsilon n \alpha_{k+1}]\\
		  &= (1-\tfrac{4\epsilon}{(1+3\epsilon)}) (1+3\epsilon)[(n+\delta_{k+1}) (1+\alpha_{k+1}) - \tfrac{n}{2}\alpha_{k+1} + \tfrac{2n}{3} \alpha_{k+1} - 2\epsilon n \alpha_{k+1}]\\
		  &> (1+3\epsilon)(n+\delta_{k+1})(1+\alpha_{k+1}) + (1+3\epsilon)n\alpha_{k+1} \left(\tfrac{2}{3} - \tfrac{1+\alpha_{k+1}}{2} -2\epsilon\right) - 4\epsilon (4n)  \\
		  &\quad \quad (\text{as $[(n+\delta_{k+1}) (1+\alpha_{k+1}) - \tfrac{n}{2}\alpha_{k+1} + \tfrac{2n}{3} \alpha_{k+1} - 2\epsilon n \alpha_{k+1}]< 4n$}) \\
		  &>(1+\epsilon)^2(n+\delta_{k+1})(1+\alpha_{k+1}) + n (\alpha_{k+1} (2/3 - 1/2 - 1/n) - 16\epsilon) \\
		  &>(1+\epsilon)^2(n+\delta_{k+1})(1+\alpha_{k+1}) \ \ (\text{as $\epsilon < \tfrac{\alpha_{k+1}}{192}$})
		\end{align*} 
		However, total (money supply) price of $b^{k+1}_{2i}$ is at most $(n+\delta_{k+1})p(b^{k+1}_{2i}) \le (n+\delta_{k+1}) (1+\epsilon)^2 (1+\alpha_{k+1})$ since it's total endowment is $n + \delta_{k+1}$ by Lemma~\ref{total-endowments}. This contradicts demand equals supply for chore $b^{k+1}_{2i}$.
		\item \textbf{$p(b^{k+1}_{2i})\le p(b^{k+1}_{2i-1})$:} In this case, $p(b^{k+1}_{2i})\le (1+\epsilon)^2$. Since the total endowment of $b^{k+1}_{2i}$ is $n+ \delta_{k+1}$ by Lemma~\ref{total-endowments}, the total (money supply) price of chore $b^{k+1}_{2i}$ is $(1+\epsilon)^2 (n+ \delta_{k+1})$. Next we show that this is strictly less than the demand from agent $a^k_{2i}$. Agent $a^k_{2i}$ owns $n$ units of $b^k_{2i}$ and hence earns at least $(1-\epsilon) n (1+\alpha_k)$ at \eA CE. 

		\begin{align*}
		(1-\epsilon) n (1+\alpha_k) & =  \frac{(1-\epsilon)}{(1+3\epsilon)} (1+3\epsilon)n (1+\frac{2}{3}\alpha_{k+1}) \\ 
		&= (1-\tfrac{4\epsilon}{(1+3\epsilon)} (1+3\epsilon) (n +\frac{n}{2}\alpha_{k+1} +\frac{n}{6}\alpha_{k+1})\\
		&>(1+3\epsilon)(n+\delta_{k+1}) + \frac{n}{6}\alpha_{k+1}) - 4\epsilon (2n)\\
		&>(1+\epsilon)^2(n+\delta_{k+1}) \ \ (\text{as $\epsilon < \frac{\alpha_{k+1}}{48}$}) \qedhere
		\end{align*}
	\end{itemize}
\end{proof}

%\clearpage
Since $K$ is even, a repeated application of Lemma~\ref{rev-ratio-amplification1} will yield the following lemma, 

\begin{lemma}
	\label{rev-ratio-amplification2}
	We have,
	\begin{enumerate}
		\item if $\tfrac{p(b^{1}_{2i-1})}{p(b^1_{2i})} = \tfrac{1 - \alpha_1}{1+\alpha_1}$, then $\tfrac{p(b^{K}_{2i-1})}{p(b^{K}_{2i})} = \tfrac{1 + \alpha_{K}}{1-\alpha_{K}}$, and  
		\item if $\tfrac{p(b^{1}_{2i-1})}{p(b^1_{2i})} = \tfrac{1 + \alpha_1}{1 - \alpha_1}$, then $\tfrac{p(b^{K}_{2i-1})}{p(b^{K}_{2i})} = \tfrac{1 - \alpha_{K}}{1 + \alpha_{K}}$.
	\end{enumerate}	
\end{lemma}

Now that we have shown that our instance satisfies the desired properties of {pairwise equal endowments}, {(approximate) fixed earning}, {(approximate) price equality}, {price regulation} and {reverse ratio amplification}, we are ready to outline how to determine the equilibrium strategy vector $x$ for the instance $I$ of the polymatrix game, given the \eA CE prices of the instance $E(I)$ of chore division:
\begin{align*}
\forall i\in[n],\ \ &  x_{2i-1} &= \frac{2p(b^K_{2i-1}) - (1-\alpha_K)\pi^K_i}{2 \cdot \pi^K_i \cdot \alpha_K}
 &  x_{2i} &= \frac{2p(b^K_{2i}) - (1-\alpha_K)\pi^K_i}{2 \cdot \pi^K_i \cdot \alpha_K}
\end{align*}

It is clear that given the prices of chores at a  CE,  the equilibrium strategy  vector can be obtained in linear time. We will now show that $x$ is the desired equilibrium strategy vector for instance $I$ of the polymatrix game. 

\begin{lemma}
	\label{polymatrix-reduction}
	$x = \langle x_1,x_2, \dots, x_{2n} \rangle $ is an equilibrium strategy vector for the polymatrix game instance $I$.
\end{lemma}

\begin{proof}
First, observe that since $\pi^K_i =p(b^K_{2i-1})+p(b^K_{2i})$ and our instance satisfies price regulation (Lemma~\ref{price-regulation}) we have that for all $i \in [n]$, $(1 -\alpha_K)\frac{\pi^K_i}{2} \leq p(b^K_{2i-1}), p(b^K_{2i}) \leq (1 + \alpha_K)\frac{\pi^K_i}{2}$. Therefore, for all $i \in [2n]$   $x_i \geq 0$ . Furthermore, for all $i \in [n]$ we have $x_{2i-1} +x_{2i} =  \tfrac{2p(b^K_{2i-1})  + p(b^K_{2i}) - 2(1-\alpha_K)\pi^K_i}{2 \cdot \pi^K_i\cdot \alpha_K} = \tfrac{2\pi^K_i \alpha_K}{2\pi^K_i \alpha_K} = 1$.

Now we will show that if $x^T \cdot \M_{*,{2i}} > x^T \cdot \M_{*,2i-1} + \tfrac{1}{n}$, then $ x_{2i-1} = 0$. The proof for the other symmetric condition will be similar. So let us assume that  $x^T \cdot \M_{*,{2i}} > x^T \cdot \M_{*,2i-1} + \tfrac{1}{n}$. Note that if $p(b^K_{2i-1})=(1 -\alpha_K)\frac{\pi^K_i}{2}$, then $x_{2i-1}=0$. Hence, by Lemma \ref{rev-ratio-amplification2} it suffices to show that $\frac{p(b^1_{2i-1})}{p(b^1_{2i})} = \frac{1+\alpha_1}{1-\alpha_1}$. Wlog, assume that $\pi^1_i = p(b^1_{2i-1})+p(b^1_{2i})=2$, implying $(1-\alpha_1)\le p(b^1_{2i-1}), p(b^1_{2i})\le (1+\alpha_1)$. Next we will show that if $\frac{p(b^1_{2i-1})}{p(b^1_{2i})} < \frac{1+\alpha_1}{1-\alpha_1}$, then $p(b^1_{2i})> (1+\alpha_1)$, which is a contradiction. This will complete the proof.

By Lemma \ref{price-equality} we have $2(1-O(n)\epsilon) \le \pi^K_i \le 2(1+O(n)\epsilon)$ for all $i\in [n]$. This implies, $p(b^K_{i})= (1\pm O(n)\epsilon) (2x_i\alpha_K +(1-\alpha_K))$ for all $i\in [2n]$.
Observe that the agents that have a disutility of $1- \alpha_1$ towards chore $b^1_{2i}$ are $\left\{a^K_{j,2i}\ |\ j \in [2n]\right\} \cup a'_{2i}$. Observe that at the given prices the sum of prices of chores owned by these agents is, 
%a CE, the total earning of the agents $\left\{ \cup_{j \in [2n]} a^K_{j,2i} \right\} \cup a'_{2i}$ equals the sum of prices of chores they own, which is, 
\begin{align*}
&=\sum_{j \in [2n]} \M_{j,2i} \cdot p(b^K_j) + (1- \alpha_K) \cdot (2n - \sum_{j \in [2n]} \M_{j,2i}) \quad \quad (\text{by Lemma~\ref{fixed-earning}})\\
&\ge\sum_{j \in [2n]} \M_{j,2i} \cdot (1-O(n)\epsilon)(2\alpha_K \cdot x_{j} + (1-\alpha_K)) +  (1- \alpha_K) \cdot (2n - \sum_{j \in [2n]} \M_{j,2i}) \quad \quad \text{(substituting $p(b^K_j)$)}\\ 
&=(1-O(n)\epsilon) \sum_{j \in [2n]} 2\alpha_K \cdot x_{j} \cdot \M_{j,2i} + (1-\alpha_K) (1-O(n)\epsilon)\cdot\sum_{j \in [2n]}  \M_{j,2i} +  (1- \alpha_K) \cdot (2n - \sum_{j \in [2n]} \M_{j,2i})\\          
&\ge (1-O(n)\epsilon) 2\alpha_K x^T \cdot \M_{*,2i} +  2n \cdot (1-\alpha_K) - \epsilon O(n^2).               
\end{align*}
Similarly, the total earning of the agents that have a disutility of $1 - \alpha_1$ towards $b^1_{2i-1}$ is {\em at most} $(1+O(n)\epsilon)2 \alpha_K x^T \cdot \M_{*,2i-1} +  2n \cdot (1-\alpha_K)+\epsilon O(n^2)$. Observe that the agents with disutility $1-\alpha_1$ towards $b^1_{2i}$ can earn all of their money only from the chores $b^1_{2i}$ or $b^1_{2i-1}$ (as these are the only chores towards which they have finite disutility). Also note that both chores $b^1_{2i-1}$ and $b^1_{2i}$ have the same total endowment which is $n + n \cdot (1 - \alpha_K)$ by Lemma~\ref{total-endowments}(part 1). Now if, $\frac{p(b^1_{2i-1})}{p(b^1_{2i})} < \frac{1+\alpha_1}{1-\alpha_1}$, then agents with disutility $1-\alpha_1$ towards $b^1_{2i}$ earn all of their money entirely from $b^1_{2i}$ (they earn nothing from $b^1_{2i-1}$). Since at \eA CE they earn at least $(1-\epsilon)$ times price of their endowments, we have 
\begin{equation}\label{eq:le-1}
p(b^1_{2i}) \geq (1-\epsilon) \tfrac{2(1-O(n)\epsilon) \alpha_K x^T \cdot \M_{*,2i} +  2n \cdot (1-\alpha_K) - \epsilon O(n^2)}{n + n\cdot (1- \alpha_K)}
\end{equation}

Agents with disutility $1-\alpha_1$ towards $b^1_{2i-1}$ may earn from both $b^1_{2i}$ and $b^1_{2i-1}$. Since they earn at most $(1+\epsilon)$ times price of their endowments, we have
\begin{equation}\label{eq:le-2}
p(b^1_{2i-1}) \leq (1+\epsilon)\tfrac{2(1+O(n)\epsilon)\alpha_K x^T \cdot \M_{*,2i-1} +  2n \cdot (1-\alpha_K)+\epsilon O(n^2)}{n + n\cdot (1- \alpha_K)}
\end{equation}

Using the above two inequalities together with $x^T \cdot \M_{*,{2i}} > x^T \cdot \M_{*,2i-1} + \tfrac{1}{n}$ we will next show that $p(b^1_{2i}) > (1+\alpha_1)$. 
{\small
\begin{align*}
(n + n\cdot (1- \alpha_K)) p(b^1_{2i})& \ge  (1-\epsilon) \left[2(1-O(n)\epsilon) \alpha_K x^T \cdot \M_{*,2i} +  2n \cdot (1-\alpha_K) - \epsilon O(n^2)\right]\\
& > (1-\epsilon) \left[2(1-O(n)\epsilon) \alpha_K x^T \cdot \M_{*,2i-1} + \tfrac{2(1-O(n)\epsilon)\alpha_K}{n} + 2n (1-\alpha_K) -\epsilon O(n^2)\right]\\
& > (1-\epsilon)(1-O(n)\epsilon)^2\left[(1-\epsilon)(n+n(1-\alpha_K))p(b^1_{2i-1}) -2n(1-\alpha_K) - \epsilon O(n^2) + \frac{2\alpha_K}{n} \right. \\
& \hspace{4cm} \left. +2n(1-\alpha_K) -\frac{\epsilon O(n^2)}{(1-O(n)\epsilon)^2}\right] \quad \quad\\
& \hspace{6.5cm}(\text{Using \eqref{eq:le-2} and $\tfrac{1}{(1+x)} > (1-x)$})\\
& > (1-\epsilon)^2 (1-O(n)\epsilon) \left[(n+n(1-\alpha_K))(1-\alpha_1) + \frac{2\alpha_K}{n} - \epsilon O(n^2) (1+O(n)\epsilon) \right] \\
& \quad \quad (\text{Using $p(b^1_{2i-1})> (1-\alpha_1)$ and $\tfrac{1}{(1-x)} < (1+2x)$})\\
& \ge (1-O(n)\epsilon)\left[(n+n(1-\alpha_K)) (1-\alpha_1 + \frac{\alpha_K}{n^2}\right] +(1-O(n)\epsilon) \left[\frac{\alpha_K}{n} -\epsilon O(n^2)(1+O(n)\epsilon)\right] \\
& > (1-O(n)\epsilon)(n+n(1-\alpha_K)(1-\alpha_1+n\alpha_1) \ \ (\text{as $\epsilon < \frac{\alpha_1}{4}\le \frac{\alpha_K}{4n^3}$})\\
& \ge (n+n(1-\alpha_K)) (1+\alpha_1) + (n+n(1-\alpha_K)) (1+(n-2)\alpha_1)\\
& \hspace{4cm} - O(n) \epsilon (n+n(1-\alpha_K)(1+(n-1)\alpha_1)\\
& > (n+n(1-\alpha_K)) (1+\alpha_1)
\end{align*}
}
The above implies $p(b^1_{2i}) > (1+\alpha_1)$ as aimed, which is a contradiction to $p(b^1_{2i}) \le (1+\alpha_1)$. Therefore, agents with $(1-\alpha_1)$ disutility towards $b^1_{2i}$ must be consuming both $b^1_{2i}$ and $b^1_{2i-1}$. This is possible only if $\frac{p(b^1_{2i-1})}{p(b^1_{2i})} = \frac{1+\alpha_1}{1-\alpha_1}$, implying $\frac{p(b^K_{2i-1})}{p(b^K_{2i})} = \frac{1-\alpha_K}{1+\alpha_K}$ by Lemma \ref{rev-ratio-amplification2} and thereby $p(b^K_{2i-1})=(1 -\alpha_K)\frac{\pi^K_i}{2}$. Replacing this in the expression for $x_{2i-1}$ we get,

\begin{align*}
x_{2i-1} &=\frac{2p(b^K_{2i}) - (1-\alpha_K)\pi^K_i}{2 \cdot \pi^K_i \cdot \alpha_K}\\
& = \frac{\pi^K_i(1- \alpha_K) - (1-\alpha_K)\pi^K_i}{2 \cdot \alpha_K}\\
&=0.
\end{align*}
A very similar argument will show that when $x^T \cdot \M_{*,{2i-1}} > x^T \cdot \M_{*,2i} + \tfrac{1}{n}$, then $ x_{2i} = 0$. Thus, $x = \langle x_1, x_2, \dots , x_n \rangle$ is a desired Nash equilibrium strategy vector for the polymatrix game $I$.  
\end{proof}

Note that, since $\epsilon=\frac{\alpha_1}{200 n}$, we have $\epsilon=\frac{1}{200 n^{3c+1}}=\frac{1}{poly(n)}$. Thus, this immediately implies the main result of this section, where the sufficiency conditions to guarantee existence of CE are of Theorem~\ref{sufficientcondition}.

\begin{theorem}
	\label{PPADhard}
	Finding $(1-\frac{1}{poly(n)})$ approximate CE for chore division in the exchange model is PPAD-hard even for instances satisfying conditions $SC_1$ and $SC_2$ of Theorem~\ref{sufficientcondition}. %when restricted to the instances that satisfy condtions of Theorem~\ref{}. 
	%Let $\mathcal{I}$ be the set of all chore division instances under the exchange model that satisfy Conditions 1 and 2 in Section~\ref{sufficiency}. Finding $\frac{1}{poly(n)}$ approximate CE is PPAD-hard even when restricted to the set of instances $\mathcal{I}$.
\end{theorem}

\begin{proof}
	We bring all the points together. Normalized polymatrix game is PPAD-hard~\cite{chen2017complexity}. Given an instance $I$ of the normalized polymatrix game, in polynomial time we can determine the instance $E(I)$. $E(I)$ satisfies the sufficiency conditions $SC_1$ and $SC_2$ of Theorem~\ref{sufficientcondition} and therefore admits an exact as well as approximate CE. Given the prices at a $(1-\frac{1}{poly(n)})$-approximate CE of $E(I)$, in polynomial time we can determine the equilibrium strategy vector for the polymatrix game by Lemma~\ref{polymatrix-reduction}. Therefore, chore division is PPAD-hard even on instances that satisfy $SC_1$ and $SC_2$ of Theorem~\ref{sufficientcondition} sufficient to guarantee existence of CE.
\end{proof}

%\bibliographystyle{alpha}
%\bibliography{chores}

\newcommand{\etalchar}[1]{$^{#1}$}

\end{document}

%% file: price-update.tex
\begin{tikzpicture}[
roundnode/.style={circle, draw=green!60, fill=green!5, very thick, minimum size=7mm},
roundnode2/.style={circle, draw=red!60, fill=red!5, very thick, minimum size=7mm},
roundnode3/.style={circle, draw=green!60, fill=green!55, very thick, minimum size=7mm},
roundnode4/.style={circle, draw=red!60, fill=red!55, very thick, minimum size=7mm},
liner/.style={red,->},
lineb/.style={blue,->},
tliner/.style={red,->,ultra thick},
tlineb/.style={blue,->,ultra thick},
tlinebl/.style={black,->,ultra thick},
]

%All in Instant 1 
%Nodes

    \node (b1) [roundnode3] {$b_1$};
    \node (b2) [roundnode3] [below=of b1]  {$b_2$};
    \node (b3) [roundnode3] [below=of b2]  {$b_3$};
    \node (b4) [roundnode] [below=of b3]  {$b_4$};
    \node (b5) [roundnode] [below=of b4]  {$b_5$};

    \node [above=0.5pt of b1] {$\scriptstyle{0.5}$};
    \node [above=0.5pt of b2] {$\scriptstyle{0.5}$};
    \node [above=0.5pt of b3] {$\scriptstyle{0.5}$};
    \node [above=0.5pt of b4] {$\scriptstyle{0}$};
    \node [above=0.5pt of b5] {$\scriptstyle{0}$};

    \node (s) [roundnode] [left=of b3] {$s$};

    \node (g1) [roundnode2] [right=of b1]  {$g_1$};
    \node (g2) [roundnode4] [right=of b2]  {$g_2$};
    \node (g3) [roundnode4] [right=of b3]  {$g_3$};
    \node (g4) [roundnode4] [right=of b4]  {$g_4$};
    \node (g5) [roundnode2] [right=of b5]  {$g_5$};

    \node (t) [roundnode2] [right=of g3] {$t$};

%Lines

%source and sink edges
    \draw[lineb] (s)--node[pos=0.53,below] {$\scriptstyle{1}$} (b1);  \draw[lineb] (g1)--node[pos=0.53,sloped,above] {$\scriptstyle{1.75}$}(t);
    \draw[lineb] (s)--node[pos=0.33,below] {$\scriptstyle{1}$}(b2);  \draw[lineb] (g2)--node[pos=0.33,sloped,below] {$\scriptstyle{0.5}$}(t);
    \draw[lineb] (s)--node[pos=0.33,below] {$\scriptstyle{1}$}(b3);  \draw[lineb] (g3)--node[pos=0.33,sloped,below] {$\scriptstyle{0.5}$}(t);
    \draw[lineb] (s)--node[pos=0.45,below] {$\scriptstyle{1}$}(b4);  \draw[lineb] (g4)--node[pos=0.33,sloped,below] {$\scriptstyle{0.5}$}(t);
    \draw[lineb] (s)--node[pos=0.56, below] {$\scriptstyle{1}$}(b5);  \draw[lineb] (g5)--node[pos=0.33,sloped,below] {$\scriptstyle{1.75}$}(t);

%Edges from Agents to goods
    \draw[liner] (b1)--(g2);
    \draw[liner] (b1)--(g3);
    \draw[liner] (b2)--(g3);
    \draw[liner] (b2)--(g4);
    \draw[liner] (b3)--(g2);
    \draw[liner] (b3)--(g3);
    \draw[liner] (b4)--(g5);
    \draw[liner] (b5)--(g1);
    \draw[liner] (b5)--(g3);

%After price update  
%Nodes

\node (c1) [roundnode3] [right=200pt of g1]{$b_1$};
\node (c2) [roundnode3] [below=of c1]  {$b_2$};
\node (c3) [roundnode3] [below=of c2]  {$b_3$};
\node (c4) [roundnode] [below=of c3]  {$b_4$};
\node (c5) [roundnode] [below=of c4]  {$b_5$};

\node (s') [roundnode] [left=of c3] {$s$};

\node (h1) [roundnode2] [right=of c1]  {$g_1$};
\node (h2) [roundnode4] [right=of c2]  {$g_2$};
\node (h3) [roundnode4] [right=of c3]  {$g_3$};
\node (h4) [roundnode4] [right=of c4]  {$g_4$};
\node (h5) [roundnode2] [right=of c5]  {$g_5$};

\node [above=0.5pt of c1] {$\scriptstyle{0.4}$};
\node [above=0.5pt of c2] {$\scriptstyle{0.4}$};
\node [above=0.5pt of c3] {$\scriptstyle{0.4}$};
\node [above=0.5pt of c4] {$\scriptstyle{0.08}$};
\node [above=0.5pt of c5] {$\scriptstyle{0.08}$};

\node (t') [roundnode2] [right=of h3] {$t$};

%Lines

%source and sink edges
\draw[lineb] (s')--node[pos=0.53,below] {$\scriptstyle{1}$} (c1);  \draw[lineb] (h1)--node[pos=0.53,sloped,above] {$\scriptstyle{1.6}$}(t');
\draw[lineb] (s')--node[pos=0.33,below] {$\scriptstyle{1}$}(c2);  \draw[tlineb] (h2)--node[pos=0.33,sloped,below] {$\scriptstyle{0.6}$}(t');
\draw[lineb] (s')--node[pos=0.33,below] {$\scriptstyle{1}$}(c3);  \draw[tlineb] (h3)--node[pos=0.33,sloped,below] {$\scriptstyle{0.6}$}(t');
\draw[lineb] (s')--node[pos=0.45,below] {$\scriptstyle{1}$}(c4);  \draw[tlineb] (h4)--node[pos=0.33,sloped,below] {$\scriptstyle{0.6}$}(t');
\draw[lineb] (s')--node[pos=0.56, below] {$\scriptstyle{1}$}(c5);  \draw[lineb] (h5)--node[pos=0.33,sloped,below] {$\scriptstyle{1.6}$}(t');

%Edges from Agents to goods
\draw[tliner] (c1)--(h2);
\draw[tliner] (c1)--(h3);
\draw[tliner] (c2)--(h3);
\draw[tliner] (c2)--(h4);
\draw[tliner] (c3)--(h2);
\draw[tliner] (c3)--(h3);
\draw[liner] (c4)--(h5);
\draw[liner] (c5)--(h1);
\draw[->,black] (c3)--(h5);

\node (x) [right=5pt of t] {};
\node (y) [left=5pt of s'] {};
\draw[blue, ->, ultra thick] (x)-- node[below] {$\scriptstyle{\text{Price adjustment}}$} (y);

\end{tikzpicture}

%% file: main.bbl
\begin{thebibliography}{CGMM22}
	
	\bibitem[BCM22]{BoodaghiansCM22}
	Shant Boodaghians, Bhaskar~Ray Chaudhury, and Ruta Mehta.
	\newblock Polynomial time algorithms to find an approximate competitive
	equilibrium for chores.
	\newblock In {\em {SODA}}, pages 2285--2302. {SIAM}, 2022.
	
	\bibitem[BKV18]{BarmanKV18}
	Siddharth Barman, Sanath~Kumar Krishnamurthy, and Rohit Vaish.
	\newblock Finding fair and efficient allocations.
	\newblock In {\em Proc.\ 19th Conf.\ Economics and Computation (EC)}, pages
	557--574, 2018.
	
	\bibitem[BMSY17]{BogomolnaiaMSY17}
	Anna Bogomolnaia, Herv{\'{e}} Moulin, Fedor Sandomirskiy, and Elena Yanovskaia.
	\newblock Competitive division of a mixed manna.
	\newblock {\em Econometrica}, 85(6):1847--1871, 2017.
	
	\bibitem[BS19]{BranzeiS19}
	Simina Branzei and Fedor Sandomirskiy.
	\newblock Algorithms for competitive division of chores.
	\newblock arXiv:1907.01766 (To appear in Mathematics of Operations Research),
	2019.
	
	\bibitem[CCPY22]{ChenCPY22}
	Thomas Chen, Xi~Chen, Binghui Peng, and Mihalis Yannakakis.
	\newblock Computational hardness of the {Hylland-Zeckhauser} scheme.
	\newblock In {\em Proc.\ 33rd Symp.\ Discrete Algorithms (SODA)}, 2022.
	\newblock To appear.
	
	\bibitem[CDDT09]{ChenDDT09}
	Xi~Chen, Decheng Dai, Ye~Du, and Shang{-}Hua Teng.
	\newblock Settling the complexity of {A}rrow-{D}ebreu equilibria in markets
	with additively separable utilities.
	\newblock In {\em Proc.\ 50th Symp.\ Foundations of Computer Science (FOCS)},
	pages 273--282, 2009.
	
	\bibitem[CDG{\etalchar{+}}17]{ColeDGJMVY17}
	Richard Cole, Nikhil Devanur, Vasilis Gkatzelis, Kamal Jain, Tung Mai, Vijay
	Vazirani, and Sadra Yazdanbod.
	\newblock Convex program duality, {F}isher markets, and {N}ash social welfare.
	\newblock In {\em Proc.\ 18th Conf.\ Economics and Computation (EC)}, 2017.
	
	\bibitem[CGMM21]{ChaudhuryGMM21}
	Bhaskar~Ray Chaudhury, Jugal Garg, Peter McGlaughlin, and Ruta Mehta.
	\newblock Competitive allocation of a mixed manna.
	\newblock In {\em Proc.\ 32nd Symp.\ Discrete Algorithms (SODA)}, 2021.
	
	\bibitem[CGMM22]{ChaudhuryGMM22}
	Bhaskar~Ray Chaudhury, Jugal Garg, Peter McGlaughlin, and Ruta Mehta.
	\newblock On the existence of competitive equilibrium with chores.
	\newblock In {\em {ITCS}}, volume 215 of {\em LIPIcs}, pages 41:1--41:13.
	Schloss Dagstuhl - Leibniz-Zentrum f{\"{u}}r Informatik, 2022.
	
	\bibitem[CM18]{ChaudhuryM18}
	Bhaskar~Ray Chaudhury and Kurt Mehlhorn.
	\newblock Combinatorial algorithms for general linear {Arrow-Debreu} markets.
	\newblock In {\em Proc.\ 38th Conf.\ Foundations of Software Tech.\ and
		Theoret.\ Comp.\ Sci. (FSTTCS)}, volume 122, pages 26:1--26:16, 2018.
	
	\bibitem[CMPV05]{CodenottiMPV05}
	Bruno Codenotti, Benton McCune, Sriram Penumatcha, and Kasturi~R. Varadarajan.
	\newblock Market equilibrium for {CES} exchange economies: {E}xistence,
	multiplicity, and computation.
	\newblock In {\em Proc.\ 25th Conf.\ Foundations of Software Tech.\ and
		Theoret.\ Comp.\ Sci. (FSTTCS)}, pages 505--516, 2005.
	
	\bibitem[Cor89]{cornet89}
	Bernard Cornet.
	\newblock Linear exchange economies.
	\newblock Technical report, Cahier Eco-Math, Universit{\'e} de Paris, 1989.
	
	\bibitem[CPY17]{chen2017complexity}
	Xi~Chen, Dimitris Paparas, and Mihalis Yannakakis.
	\newblock The complexity of non-monotone markets.
	\newblock {\em Journal of the ACM (JACM)}, 64(3):1--56, 2017.
	
	\bibitem[CSVY06]{CodenottiSVY06}
	Bruno Codenotti, Amin Saberi, Kasturi Varadarajan, and Yinyu Ye.
	\newblock {L}eontief economies encode two-player zero-sum games.
	\newblock In {\em Proc.\ 17th Symp.\ Discrete Algorithms (SODA)}, pages
	659--667, 2006.
	
	\bibitem[CT09]{ChenT09}
	Xi~Chen and Shang{-}Hua Teng.
	\newblock Spending is not easier than trading: {O}n the computational
	equivalence of {F}isher and {A}rrow-{D}ebreu equilibria.
	\newblock In {\em Proc.\ 20th Intl.\ Symp.\ Algorithms and Computation
		(ISAAC)}, pages 647--656, 2009.
	
	\bibitem[DGM16]{DuanGM16}
	Ran Duan, Jugal Garg, and Kurt Mehlhorn.
	\newblock An improved combinatorial polynomial algorithm for the linear
	{A}rrow-{D}ebreu market.
	\newblock In {\em Proc.\ 27th Symp.\ Discrete Algorithms (SODA)}, pages
	90--106, 2016.
	
	\bibitem[DGV16]{DevanurGV16}
	Nikhil Devanur, Jugal Garg, and L{\'{a}}szl{\'{o}} V{\'{e}}gh.
	\newblock A rational convex program for linear {A}rrow-{D}ebreu markets.
	\newblock {\em ACM Trans.\ Econom.\ Comput.}, 5(1):6:1--6:13, 2016.
	
	\bibitem[DM15]{DuanM15}
	Ran Duan and Kurt Mehlhorn.
	\newblock A combinatorial polynomial algorithm for the linear {A}rrow-{D}ebreu
	market.
	\newblock {\em Inf.\ Comput.}, 243:112--132, 2015.
	
	\bibitem[DPSV08]{DevanurPSV08}
	Nikhil Devanur, Christos Papadimitriou, Amin Saberi, and Vijay Vazirani.
	\newblock Market equilibrium via a primal--dual algorithm for a convex program.
	\newblock {\em J. ACM}, 55(5), 2008.
	
	\bibitem[Eav76]{Eaves76}
	B.~Curtis Eaves.
	\newblock A finite algorithm for the linear exchange model.
	\newblock {\em J. Math.\ Econom.}, 3:197--203, 1976.
	
	\bibitem[EG59]{EisenbergG59}
	Edmund Eisenberg and David Gale.
	\newblock Consensus of subjective probabilities: {T}he {P}ari-{M}utuel method.
	\newblock {\em Ann.\ Math.\ Stat.}, 30(1):165--168, 1959.
	
	\bibitem[Eis61]{Eisenberg61}
	Edmund Eisenberg.
	\newblock Aggregation of utility functions.
	\newblock {\em Management Sci.}, 7(4):337--350, 1961.
	
	\bibitem[EPS22]{EbadianPS22}
	Soroush Ebadian, Dominik Peters, and Nisarg Shah.
	\newblock How to fairly allocate easy and difficult chores.
	\newblock In {\em Proc.\ 21st Conf.\ Auton.\ Agents and Multi-Agent Systems
		(AAMAS)}, 2022.
	\newblock To appear.
	
	\bibitem[FGHS21]{FearnleyGHS21}
	John Fearnley, Paul~W. Goldberg, Alexandros Hollender, and Rahul Savani.
	\newblock The complexity of gradient descent: {CLS} = {PPAD} {\(\cap\)} {PLS}.
	\newblock In {\em Proc.\ 53rd Symp.\ Theory of Computing (STOC)}, pages 46--59,
	2021.
	
	\bibitem[Fuj80]{fujishige1980lexicographically}
	Satoru Fujishige.
	\newblock Lexicographically optimal base of a polymatroid with respect to a
	weight vector.
	\newblock {\em Mathematics of Operations Research}, 5(2):186--196, 1980.
	
	\bibitem[Gal76]{Gale76}
	David Gale.
	\newblock The linear exchange model.
	\newblock {\em Journal of Mathematical Economics}, 3(2):205--209, l976.
	
	\bibitem[GHV21]{GargHV21}
	Jugal Garg, Edin Husi{\'{c}}, and L{\'{a}}szl{\'{o}} V{\'{e}}gh.
	\newblock Auction algorithms for market equilibrium with weak gross substitute
	demands.
	\newblock In {\em Proc.\ 38th Symp.\ Theoret.\ Aspects of Computer Science
		(STACS)}, volume 187, pages 33:1--33:19, 2021.
	
	\bibitem[GL91]{GoldfarbL91}
	Donald Goldfarb and Shucheng Liu.
	\newblock An {O(n\({}^{\mbox{3}}\)L)} primal interior point algorithm for
	convex quadratic programming.
	\newblock {\em Math. Program.}, 49:325--340, 1991.
	
	\bibitem[GM20]{GargM20}
	Jugal Garg and Peter McGlaughlin.
	\newblock Computing competitive equilibria with mixed manna.
	\newblock In {\em Proc.\ 19th Conf.\ Auton.\ Agents and Multi-Agent Systems
		(AAMAS)}, pages 420--428, 2020.
	
	\bibitem[GMQ22]{GargMQ22}
	Jugal Garg, Aniket Murhekar, and John Qin.
	\newblock Fair and efficient allocations of chores under bivalued preferences.
	\newblock In {\em Proc.\ 36th Conf.\ Artif.\ Intell.\ (AAAI)}, 2022.
	\newblock To appear.
	
	\bibitem[GMVY17]{GargMVY17}
	Jugal Garg, Ruta Mehta, Vijay~V. Vazirani, and Sadra Yazdanbod.
	\newblock Settling the complexity of {L}eontief and {PLC} exchange markets
	under exact and approximate equilibria.
	\newblock In {\em STOC}, pages 890--901, 2017.
	
	\bibitem[GP14]{GoldmanP14}
	Jonathan~R. Goldman and Ariel~D. Procaccia.
	\newblock Spliddit: unleashing fair division algorithms.
	\newblock {\em SIGecom Exchanges}, 13(2):41--46, 2014.
	
	\bibitem[GV19]{GargV19}
	Jugal Garg and L{\'a}szl{\'o}~A V{\'e}gh.
	\newblock A strongly polynomial algorithm for linear exchange markets.
	\newblock In {\em Proc.\ 51st Symp.\ Theory of Computing (STOC)}, 2019.
	
	\bibitem[HZ79]{HyllandZ79}
	Aanund Hylland and Richard Zeckhauser.
	\newblock The efficient allocation of individuals to positions.
	\newblock {\em J. Political Economy}, 87(2):293--314, 1979.
	
	\bibitem[Jai07]{Jain07}
	Kamal Jain.
	\newblock A polynomial time algorithm for computing the {A}rrow-{D}ebreu market
	equilibrium for linear utilities.
	\newblock {\em SIAM J. Comput.}, 37(1):306--318, 2007.
	
	\bibitem[NP83]{nenakov83}
	E~I Nenakov and M~E Primak.
	\newblock One algorithm for finding solutions of the {Arrow-Debreu} model.
	\newblock {\em Kibernetica}, 3:127--128, 1983.
	
	\bibitem[Orl10]{Orlin10}
	James Orlin.
	\newblock Improved algorithms for computing {F}isher's market clearing prices.
	\newblock In {\em Proc.\ 42nd Symp.\ Theory of Computing (STOC)}, pages
	291--300, 2010.
	
	\bibitem[Shm09]{Shmyrev09}
	Vadim Shmyrev.
	\newblock An algorithm for finding equilibrium in the linear exchange model
	with fixed budgets.
	\newblock {\em J. Appl.\ Indust.\ Math.}, 3(4):505--518, 2009.
	
	\bibitem[V{\'{e}}g12]{Vegh12}
	L{\'{a}}szl{\'{o}} V{\'{e}}gh.
	\newblock Strongly polynomial algorithm for a class of minimum-cost flow
	problems with separable convex objectives.
	\newblock In {\em Proc.\ 44th Symp.\ Theory of Computing (STOC)}, pages 27--40,
	2012.
	
	\bibitem[Ye08]{Ye08}
	Yinyu Ye.
	\newblock A path to the {A}rrow-{D}ebreu competitive market equilibrium.
	\newblock {\em Math.\ Prog.}, 111(1-2):315--348, 2008.
	
\end{thebibliography}
